\documentclass[11pt]{article}

\usepackage[margin=1in]{geometry}
\usepackage{amsmath,amssymb,amsthm,mathtools}
\usepackage{booktabs}
\usepackage{graphicx}
\usepackage{enumitem}
\usepackage{microtype}
\usepackage[numbers,sort&compress]{natbib}
\usepackage[colorlinks=true,linkcolor=blue,citecolor=blue,urlcolor=blue]{hyperref}
\usepackage{pdflscape}
\usepackage{threeparttable}
\usepackage{natbib}
\usepackage{tikz}
\usetikzlibrary{positioning}

\newcommand{\bfepsilon}{\mbox{\boldmath $\epsilon$}}

\newcommand{\bftheta}{\mbox{\boldmath $\theta$}}

\newcommand{\bfd}{{\bf d}}

\newcommand{\bff}{{\bf f}}

\newcommand{\bfI}{{\bf I}}

\newcommand{\bfW}{{\bf W}}
\newcommand{\bfx}{{\bf x}}

\newcommand{\bfy}{{\bf y}}

\newcommand{\bfzero}{{\bf 0}}

\newcommand{\R}{\mathbb{R}}
\newcommand{\N}{\mathbb{N}}
\newcommand{\E}{\operatorname{E}}
\newcommand{\Var}{\operatorname{Var}}

\newcommand{\Normal}{\operatorname{N}}
\newcommand{\Exp}{\operatorname{Exp}}

\newtheorem{proposition}{proposition}
\newtheorem{remark}{remark}

\title{Resolution-Adaptive Compact-Support Priors for Bayesian Wavelet Denoising}
\author{Nilotpal Sanyal\\[2pt] The University of Texas at El Paso}
\date{}

\begin{document}

\maketitle

\begin{abstract}
We propose a resolution-adaptive Bayesian wavelet denoising method for noisy one-dimensional signals. The main contribution is a spike-and-slab prior whose continuous slab is a mixture of a compactly supported Wendland-type polynomial kernel and the semicircle density, with data-adaptive, resolution-specific mixture weights, produced by a low-dimensional empirical-Bayes trend. The Wendland component concentrates mass near zero and vanishes smoothly at the support boundary, whereas the semicircle component is more dispersed. This construction combines explicit sparsity and support control with an interpretable mechanism for adapting the shrinkage shape across resolutions. Under squared-error loss, we derive the posterior-mean estimator, establish key symmetry, boundedness, continuity, and limiting properties, define pointwise fixed-hyperparameter bias, variance, and risk, and develop an empirical-Bayes estimation procedure. The Wendland contribution has finite-sum expressions under a Laplace working likelihood, while the semicircle contribution is evaluated by stable one-dimensional integration. Simulations using the Bumps, Blocks, Doppler, and HeaviSine signals compare the proposed Gaussian- and Laplace-likelihood versions with universal thresholding, false-discovery-rate (FDR) thresholding, cross-validation (CV), Stein's unbiased risk estimate (SURE), the Bayesian adaptive multiresolution shrinker (BAMS), and a nonlocal-prior (NLP)-based method. In the primary Gaussian-error simulation study, the Gaussian-likelihood version was the strongest non-NLP method in 24 of the 36 design cells, including 11 of the 12 low signal-to-noise ratio (SNR) cells, and had a substantially more favorable computational profile than the Laplace-likelihood version. Analysis of a seismic acceleration trace from the 2008 Chino Hills earthquake illustrates attenuation of rapid fluctuations and preservation of the dominant acceleration event under the chosen diagnostics. Using the processed channel-1 trace as surrogate truth, the corresponding semi-synthetic validation showed that WS--Gaussian improved on the noisy observation at lower and moderate SNRs but not at the highest SNR.
\end{abstract}


\section{Introduction}
Wavelet methods provide a localized representation of an unknown function in both the observation and scale domains.  This localization is particularly useful for signals containing discontinuities, isolated peaks, sharp transitions, and rapid oscillations.  A second important feature is sparsity, which means that, at a given resolution level, most wavelet coefficients are often close to zero, whereas a relatively small number of coefficients describe the principal features of the signal.  Wavelet shrinkage exploits this structure by reducing noisy empirical coefficients before reconstructing the signal through the inverse discrete wavelet transform \citep{DonohoJohnstone1994,DonohoJohnstone1995,Daubechies1992,Mallat2009}. We propose a Bayesian wavelet-shrinkage model whose continuous slab is a bounded mixture of a compactly supported Wendland-type density and the semicircle density.  The model permits the mixture weight and the support scale to vary with resolution, while the extent to which the data identify such variation is assessed empirically.

Classical wavelet estimators include hard and soft thresholding, false-discovery-rate (FDR) thresholding, cross-validation (CV), and Stein's unbiased risk estimation (SURE) \citep{DonohoJohnstone1994,AbramovichBenjamini1996,Nason1996,DonohoJohnstone1995}.  These procedures are useful and often fast, but their shrinkage behavior is governed by a threshold-selection rule rather than by an explicit prior distribution for the unknown coefficient.  Universal hard or soft thresholds can be conservative for moderate signals; hard thresholding also introduces a discontinuous decision at the threshold, whereas soft thresholding imposes a fixed shrinkage bias on retained coefficients.  FDR thresholding controls a multiple-testing criterion rather than directly minimizing reconstruction loss, CV can be computationally expensive and sensitive to the validation scheme, and SURE can be variable in low-SNR settings when the estimated risk is noisy.  These are not universal failures, but they motivate an explicit prior-based alternative.  A Bayesian prior can encode sparsity through a spike at zero, symmetry and bounded support through the slab, and uncertainty about the magnitude of nonzero coefficients through its continuous shape.  Early Bayesian wavelet models used mixtures of normal distributions or mixtures involving a point mass at zero and a continuous slab \citep{ChipmanKolaczykMcCulloch1997,VidakovicRuggeri2001,JohnstoneSilverman2005}. Subsequent work considered double-exponential, uniform, double-Weibull, beta, logistic, compact-support, and nonlocal priors \citep{AngeliniVidakovic2004,RemenyiVidakovic2015,SousaGarciaVidakovic2021,Sousa2022,SanyalFerreira2017,Sanyal2025,BarriosSousa2025,ReinaSousa2026}.  Low-signal-to-noise ratio (SNR) settings are particularly challenging because an observed moderate-size coefficient may represent either a genuine local feature or noise.  Bayesian shrinkage rules are attractive in this setting because they can incorporate uncertainty about both possibilities \citep{VimalajeewaDasguptaRuggeriVidakovic2023,BarriosSousa2025}.

Compactly supported slab priors are appealing because their support and shape can control the attenuation of moderate wavelet coefficients and the behavior of the shrinkage rule near the support boundary.  The contribution here is the specific Wendland--semicircle bounded mixture slab for Bayesian wavelet coefficient shrinkage, together with a resolution-indexed mixture parameter and a data-adaptive support scale.  This is a narrower novelty claim than a claim about Wendland functions, semicircle laws, compact support, or resolution-adaptive sparsity in general.  The Wendland component is a normalized compactly supported polynomial density with a high-order zero at the support boundary.  Wendland functions were originally introduced as compactly supported positive-definite functions for scattered-data approximation \citep{Wendland1995,Wendland2005}.  In the present model, this component is concentrated near zero and represents small nonzero coefficients that are difficult to distinguish from noise.  The semicircle component is also symmetric and compactly supported, but is more dispersed and places relatively more mass away from zero.  The semicircle law is classically associated with the limiting eigenvalue distribution of random matrices \citep{Wigner1955}; here it is used as a bounded slab density rather than as a spectral limit. To the best of our knowledge, the particular combination of these two bounded densities as a mixture slab for Bayesian wavelet shrinkage has not previously been studied in this form. Henceforth, we use the abbreviation WS for Wendland--semicircle.

The proposed prior combines the two slab shapes with a point mass at zero.  At each resolution level, a mixture weight determines the relative contribution of the Wendland component, while a support-scale parameter determines the common bounded truncation interval for the two slab densities. The weight is parameterized through a low-dimensional empirical Bayes trend in resolution. Thus, the model permits the relative slab contribution to vary with resolution; whether the data materially identify that variation is examined in the component-and-adaptivity comparison.  The prior retains an explicit spike at zero, allows a single specification to express both strong concentration near zero and preservation of moderate coefficients, and maintains an interpretable bounded-support structure. The posterior-mean rule is computationally tractable because the Wendland contribution benefits from its finite polynomial form, whereas the semicircle contribution reduces to a one-dimensional integral.  We consider a Gaussian observation model as the primary specification and a Laplace working likelihood as an alternative sensitivity specification.

Our empirical evaluation assesses reconstruction accuracy, normalized error, paired performance, and computational cost.  The primary Gaussian-error experiment considers the Bumps, Blocks, Doppler, and HeaviSine test signals at \(n=512\), \(1024\), and \(2048\), with amplitude SNR values \(\rho\in\{0.2,1,3\}\) and 100 replications per design cell.  An expanded \(n=1024\) experiment additionally uses \(\rho\in\{0.5,2,5\}\).  WS--Gaussian and WS--Laplace are compared with universal thresholding, FDR thresholding, CV, SURE, the Bayesian adaptive multiresolution shrinker (BAMS), and a nonlocal-prior (NLP)-based method. In the original Gaussian grid, WS--Gaussian is the best non-NLP method in 24 of the 36 cells and in 11 of the 12 low-SNR signal cells. These are cellwise statements rather than claims of universal dominance; pooled and normalized summaries, including the computational cost of the numerical methods, are reported separately.

{\color{red}We also analyze a 77-second, 200-Hz strong-motion accelerogram from the 29 July 2008 Chino Hills earthquake recorded at station 24763.  The real-data analysis compares WS fits, the two component-specific endpoint fits, and classical thresholding rules using residual variability, residual roughness, retained squared energy, preservation and timing of the dominant acceleration peak, and waveform plots.  Because the observed record has no trusted noise-free reference, these diagnostics describe the amount and character of the modification rather than denoising risk.  A separate semi-synthetic validation treats the processed channel-1 trace as a surrogate truth, adds controlled Gaussian noise at prespecified amplitude SNR values, and evaluates recovery against that surrogate. It therefore complements, but does not replace, the real-data illustration.}


The remainder of the article is organized as follows.  Section~\ref{sec:model} introduces the nonparametric regression and wavelet-domain models.  Section~\ref{sec:prior} defines the Wendland and semicircle slab densities and the mixture prior.  Section~\ref{sec:properties} establishes properties of the prior and the resulting shrinkage rule.  Section~\ref{sec:inference} develops posterior inference and the empirical-Bayes fitting procedure.  Section~\ref{sec:simulation} presents the simulation design and results, and Section~\ref{sec:real-data-seismic} gives the seismic-trace application; the semi-synthetic validation is reported in Supplementary Section~3.2. Section~\ref{sec:conclusion} summarizes the findings, limitations, and extensions.

\section{Statistical model and wavelet representation}
\label{sec:model}
Let \(\R\) denote the real line and let \(\N=\{1,2,\ldots\}\). Let \(N\) denote the number of observations in the original record and let \(\widetilde N\) denote the length of the vector supplied to the discrete wavelet transform. In the idealized dyadic setting considered here, \(N=\widetilde N=n=2^J\), where \(J\in\N\) is the number of dyadic resolution steps. Suppose that the \(N\) observations are collected at equally spaced design points \(x_i=(i-1)/(N-1)\), \(i=1,\ldots,N\). We observe
\begin{equation}
    y_i=f(x_i)+\epsilon_i,
    \qquad i=1,\ldots,N,
    \label{eq:time-model}
\end{equation}
where \(f:[0,1]\to\R\) is an unknown regression function and \(\epsilon_i\) are independent and identically distributed Gaussian errors satisfying
\begin{equation}
    \epsilon_i\sim\Normal(0,\sigma^2),
    \qquad \sigma^2>0,
    \label{eq:gaussian-errors}
\end{equation}
where \(\sigma^2\) is the observational noise variance.

Let \(\bfx=(x_1,\ldots,x_N)^T\), \(\bfy=(y_1,\ldots,y_N)^T\), \(\bff=(f(x_1),\ldots,f(x_N))^T\), and \(\bfepsilon=(\epsilon_1,\ldots,\epsilon_N)^T\). Then \eqref{eq:time-model} can be written as
\begin{equation*}
    \bfy=\bff+\bfepsilon,
    \qquad
    \bfepsilon\sim\Normal_N(\bfzero,\sigma^2\bfI_N),
    \label{eq:vector-model}
\end{equation*}
where \(\Normal_N\) denotes the \(N\)-variate normal distribution and \(\bfI_N\) is the \(N\times N\) identity matrix. 

Let \(\phi\) be a scaling function and let \(\psi\) be a mother wavelet generating an orthonormal wavelet basis. For a location index \(k\) and resolution level \(j\), define
\begin{equation*}
    \phi_{j,k}(x)=2^{j/2}\phi(2^jx-k),
    \qquad
    \psi_{j,k}(x)=2^{j/2}\psi(2^jx-k).
    \label{eq:wavelet-basis}
\end{equation*}
For a fixed primary resolution level \(J_0\), the wavelet expansion of \(f\) is written formally as
\begin{equation*}
    f(x)=\sum_k c_{J_0,k}\phi_{J_0,k}(x)
    +\sum_{j=J_0}^{\infty}\sum_k\theta_{j,k}\psi_{j,k}(x),
    \qquad
    \theta_{j,k}=\int_0^1 f(x)\psi_{j,k}(x)\,dx,
    \label{eq:wavelet-expansion}
\end{equation*}
where \(c_{J_0,k}\) are scaling coefficients and \(\theta_{j,k}\) are detail wavelet coefficients. In the finite-sample analysis below, only levels \(j=J_0,\ldots,J-1\) are represented, and the scaling coefficients are retained without shrinkage.

Let \(\mathcal P:\mathbb{R}^{N}\to\mathbb{R}^{\widetilde N}\) denote a prescribed deterministic padding or extension operator, where \(\widetilde N=2^J\). Define the vectors supplied to the DWT by
\[
    \widetilde{\bfy}=\mathcal P(\bfy),
    \qquad
    \widetilde{\bff}=\mathcal P(\bff),
    \qquad
    \widetilde{\bfepsilon}=\mathcal P(\bfepsilon).
\]
Thus, \(\widetilde{\bfy}=\widetilde{\bff}+\widetilde{\bfepsilon}\). In the ideal dyadic setting, \(\mathcal P\) is the identity and \(N=\widetilde N\). For a padded record, \(\widetilde{\bfy}\) is the computational DWT input. The padding does not represent additional independent observations. In particular, when \(\mathcal P\) is not the identity, \(\widetilde{\bfepsilon}\) need not have independent coordinates. For a non-dyadic record, the discrete wavelet transform is applied after a documented extension to a dyadic length \(\widetilde N\geq N\), and the reconstructed values are truncated to the original \(N\) observations before reporting application-level diagnostics.  The specific extension rule used for the seismic record is given in Section~\ref{sec:real-data-seismic}.

Let \(\bfW\) be an orthogonal \(\widetilde N\times\widetilde N\) matrix representing a DWT, so that \(\bfW^T\bfW=\bfI_{\widetilde N}\). For each detail level \(j=J_0,\ldots,J-1\), define the observed and true detail-coefficient vectors by \(\bfd_j=(d_{j,k}:k\in\mathcal K_j)^T\) and \(\bftheta_j=(\theta_{j,k}:k\in\mathcal K_j)^T\), respectively. Let \(\mathbf c_{J_0}^{(y)}\) and \(\mathbf c_{J_0}^{(f)}\) denote the observed and true scaling-coefficient vectors at the primary resolution level. With the DWT coefficients ordered by placing the scaling coefficients first, define
\begin{equation*}
    \bfd=\bfW\widetilde{\bfy}=
    \begin{pmatrix}
        \mathbf c_{J_0}^{(y)}\\
        \bfd_{J_0}\\
        \vdots\\
        \bfd_{J-1}
    \end{pmatrix},
    \qquad
    \bftheta=\bfW\widetilde{\bff}=
    \begin{pmatrix}
        \mathbf c_{J_0}^{(f)}\\
        \bftheta_{J_0}\\
        \vdots\\
        \bftheta_{J-1}
    \end{pmatrix},
    \qquad
    \bfepsilon^*=\bfW\widetilde{\bfepsilon}.
    \label{eq:dwt-vectors}
\end{equation*}
Thus, \(\bfd\), \(\bftheta\), and \(\bfepsilon^*\) are \(\widetilde N\)-dimensional vectors. In the ideal dyadic setting,
\[
    \bfepsilon^*\sim\Normal_{\widetilde N}
    (\bfzero,\sigma^2\bfI_{\widetilde N}).
\]
For a nontrivial padding operator, no independent multivariate-normal assertion is made for \(\bfepsilon^*\). The coefficientwise model below is used as the working likelihood for the padded analysis.

The prior in \eqref{eq:proposed-prior}, below, is assigned only to the scalar detail coefficients \(\theta_{j,k}\); the scaling coefficients \(\mathbf c_{J_0}^{(f)}\) are retained without shrinkage. In scalar expressions below, \(\theta\) denotes a generic value of a detail coefficient, whereas \(\theta_{j,k}\) denotes the particular detail coefficient at level \(j\) and location \(k\). For a detail coefficient at level \(j\) and location \(k\), the coefficientwise wavelet-domain model is
\begin{equation}
    d_{j,k}=\theta_{j,k}+\epsilon^*_{j,k},
    \qquad
    \epsilon^*_{j,k}\sim\Normal(0,\sigma^2),
    \label{eq:coefficient-model}
\end{equation}
where \(j=J_0,\ldots,J-1\) and \(k\in\mathcal K_j=\{0,1,\ldots,2^j-1\}\). Here \(d_{j,k}\) is the observed empirical detail coefficient.

\section{Adaptive Wendland--semicircle prior}
\label{sec:prior}

\subsection{Two normalized slab densities}
For a set \(A\subseteq\R\), let \(I(A)\) denote its indicator function. Let \(u\) be a dimensionless coefficient and let \(\beta>0\) be a scale parameter. We define the Wendland-type kernel
\begin{equation}
    K_W(u)=\frac{3}{2}(1-|u|)^4(1+4|u|)I(|u|<1).
    \label{eq:wendland-kernel}
\end{equation}
The polynomial in \eqref{eq:wendland-kernel} can equivalently be written, for \(|u|<1\), as
\begin{equation}
    (1-|u|)^4(1+4|u|)
    =1-10|u|^2+20|u|^3-15|u|^4+4|u|^5.
    \label{eq:wendland-polynomial}
\end{equation}
The normalization follows from
\begin{equation}
    \int_{-1}^{1}K_W(u)\,du
    =3\int_0^1(1-u)^4(1+4u)\,du=1.
    \label{eq:wendland-normalization}
\end{equation}

The second slab is the semicircle kernel
\begin{equation}
    K_S(u)=\frac{2}{\pi}\sqrt{1-u^2}\,I(|u|<1).
    \label{eq:semicircle-kernel}
\end{equation}
Since \(\int_{-1}^{1}\sqrt{1-u^2}\,du=\pi/2\), \(K_S\) is also a density. The scale versions of these kernels are
\begin{align*}
    g_W(\theta;\beta)&=\frac{1}{\beta}K_W\left(\frac{\theta}{\beta}\right)
    =\frac{3}{2\beta}\left(1-\frac{|\theta|}{\beta}\right)^4
    \left(1+4\frac{|\theta|}{\beta}\right)I(|\theta|<\beta),
    \\
    g_S(\theta;\beta)&=\frac{1}{\beta}K_S\left(\frac{\theta}{\beta}\right)
    =\frac{2}{\pi\beta^2}\sqrt{\beta^2-\theta^2}\,I(|\theta|<\beta).
\end{align*}
Both densities are symmetric about zero and have support \([-\beta,\beta]\). The Wendland density is more concentrated near zero and has a high-order zero at the support boundary. The semicircle density is less concentrated at the center and allocates relatively more prior mass to moderate coefficients. The differing concentration and dispersion of the two slab densities are illustrated in Figure~\ref{fig:slab-densities}, which also displays their equal-weight continuous mixture on the standardized support.

\begin{figure}[t]
    \centering
    \includegraphics[width=0.8\textwidth]{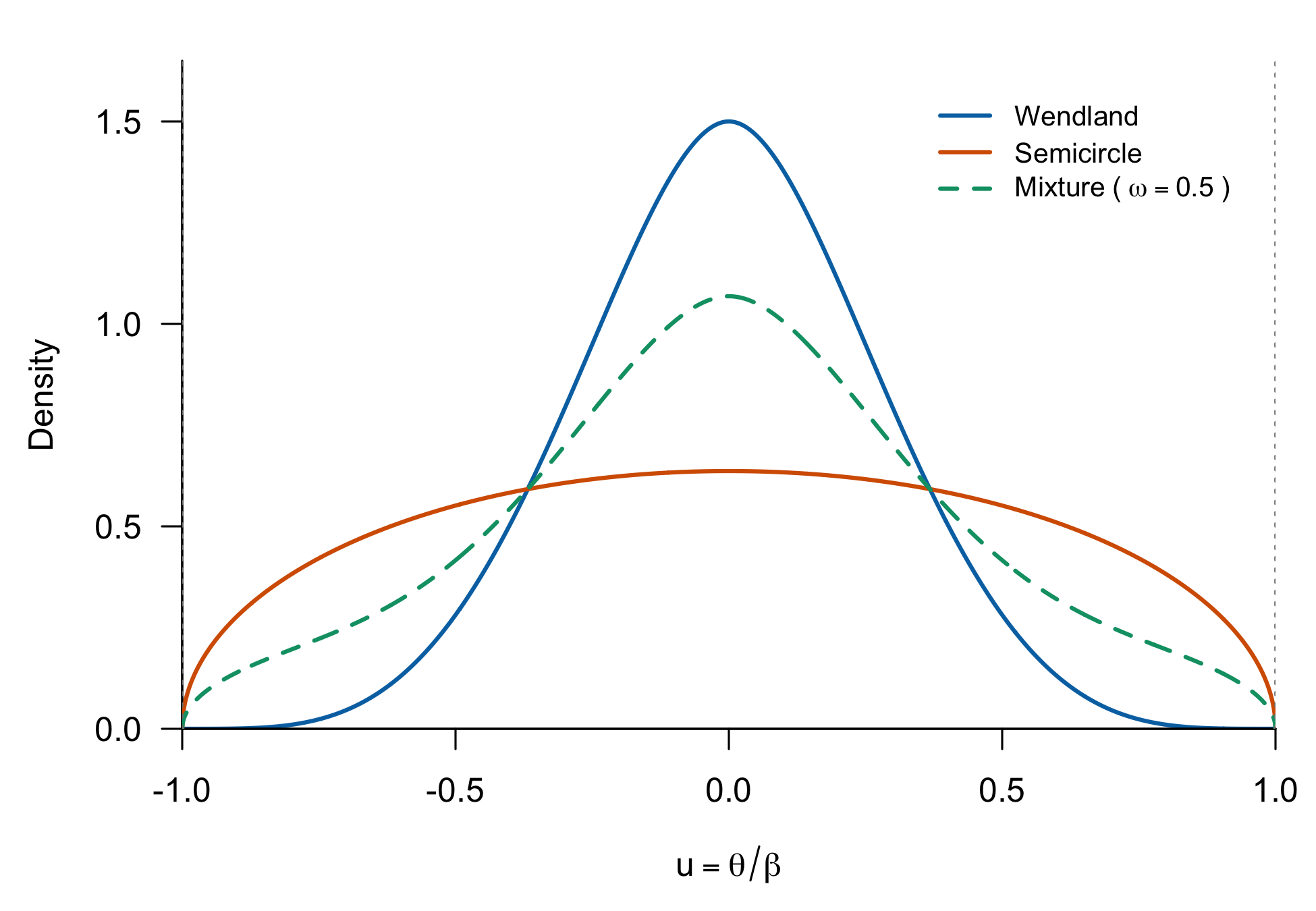}
    \caption{Normalized Wendland-type and semicircle slab densities on the common standardized support $u=\theta/\beta$, together with their continuous equal-weight mixture $0.5K_W(u)+0.5K_S(u)$. The mixture curve is shown for illustration only. In the proposed model, the mixture weight $\omega_j$ is estimated separately at each resolution level. The point-mass spike at zero is not shown.}
    \label{fig:slab-densities}
\end{figure}

\begin{proposition}[Normalization and dispersion]
\label{prop:kernel-properties}
For every \(\beta>0\), \(g_W(\cdot;\beta)\) and \(g_S(\cdot;\beta)\) are probability densities on \(\R\), both have mean zero, and
\begin{equation}
    \Var_W(\theta)=\frac{\beta^2}{14},
    \qquad
    \Var_S(\theta)=\frac{\beta^2}{4}.
    \label{eq:kernel-variances}
\end{equation}
\end{proposition}

\begin{proof}
Normalization and symmetry follow from \eqref{eq:wendland-normalization} and the semicircle area formula. For the Wendland density, using symmetry and \eqref{eq:wendland-polynomial},
\begin{align*}
    \E_W(\theta^2)
    &=3\beta^2\int_0^1 u^2(1-u)^4(1+4u)\,du
      =\frac{\beta^2}{14}.
\end{align*}
For the semicircle density,
\begin{align*}
    \E_S(\theta^2)
    &=\frac{2\beta^2}{\pi}\int_{-1}^{1}u^2\sqrt{1-u^2}\,du
      =\frac{\beta^2}{4}.
\end{align*}
The means are zero by symmetry, so the second moments equal the variances.
\end{proof}

\subsection{Spike-and-mixture slab prior}
Let \(\delta_0\) denote the Dirac probability measure concentrated at zero. For each resolution level \(j\), let \(\pi_j\in(0,1)\) be the prior probability that a coefficient is exactly zero, let \(\omega_j\in[0,1]\) be the conditional probability of the Wendland component given that the coefficient is nonzero, and let \(\beta_j>0\) be the level-specific support scale. Define the continuous mixture slab
\begin{equation}
    g_j(\theta;\omega_j,\beta_j)
    =\omega_j g_W(\theta;\beta_j)
    +(1-\omega_j)g_S(\theta;\beta_j).
    \label{eq:mixture-slab}
\end{equation}
The proposed prior for a detail coefficient is
\begin{equation}
    \theta_{j,k}\mid\pi_j,\omega_j,\beta_j
    \sim
    \pi_j\delta_0+(1-\pi_j)g_j(\theta;\omega_j,\beta_j).
    \label{eq:proposed-prior}
\end{equation}

The main theoretical development considers the interior mixture case \(0<\omega_j<1\), in which both slab components receive positive prior weight. The endpoint values define reduced one-component models:
\[
    g_j(\theta;1,\beta_j)=g_W(\theta;\beta_j),
    \qquad
    g_j(\theta;0,\beta_j)=g_S(\theta;\beta_j).
\]
Thus, \(\omega_j=1\) is the Wendland-only specification and \(\omega_j=0\) is the Semicircle-only specification.  These endpoint specifications are fitted separately in the component-and-adaptivity comparison.  They are also used in the numerical endpoint checks, where the absent component is omitted rather than evaluated with a zero numerical weight.

Equivalently, introduce a latent component indicator \(Z_{j,k}\in\{0,W,S\}\), where \(Z_{j,k}=0\) denotes the spike, \(Z_{j,k}=W\) denotes the Wendland slab, and \(Z_{j,k}=S\) denotes the semicircle slab. Its prior probabilities are
\begin{equation*}
    \Pr(Z_{j,k}=0)=\pi_j,
    \quad
    \Pr(Z_{j,k}=W)=(1-\pi_j)\omega_j,
    \quad
    \Pr(Z_{j,k}=S)=(1-\pi_j)(1-\omega_j).
    \label{eq:component-probabilities}
\end{equation*}
Conditional on \(Z_{j,k}=0\), \(\theta_{j,k}=0\). Conditional on \(Z_{j,k}=W\) or \(Z_{j,k}=S\), the coefficient has density \(g_W(\cdot;\beta_j)\) or \(g_S(\cdot;\beta_j)\), respectively.

The prior variance of a nonzero coefficient at level \(j\) is
\begin{equation}
    \Var(\theta_{j,k}\mid Z_{j,k}\ne0)
    =\beta_j^2\left\{\frac{\omega_j}{14}+\frac{1-\omega_j}{4}\right\},
    \label{eq:mixture-variance}
\end{equation}
which also applies to the two endpoint models. The prior variance of a nonzero coefficient at level \(j\) is given by \eqref{eq:mixture-variance}, and the unconditional prior variance is \((1-\pi_j)\) times that quantity.  Thus, \(\pi_j\) controls sparsity, \(\beta_j\) controls the coefficient scale, and \(\omega_j\) controls the shape of the nonzero slab.  The endpoint values \(\omega_j=0\) and \(\omega_j=1\) are handled by the corresponding one-component posterior formulas.

\subsection{Resolution-dependent hyperparameters}
Following the principle that fine-scale coefficients are more likely to be noise, we use the level-dependent spike probability
\begin{equation}
    \pi_j=1-\frac{1}{(j-J_0+\ell)^\gamma},
    \qquad j=J_0,\ldots,J-1,
    \label{eq:spike-probability}
\end{equation}
where \(\ell>1\) is the spike-offset parameter and \(\gamma>0\) is the resolution sparsity exponent.  Larger values of \(\gamma\) make the spike probability increase more rapidly with resolution level, whereas \(\ell\) controls its starting value.  These quantities are fixed in the primary analysis and are varied one at a time in the supplementary sensitivity study.

To adapt the slab shape across resolution levels, define
\begin{equation*}
    t_j=\frac{j-J_0}{J-1-J_0},
    \qquad
    \boldsymbol{\eta}=(\eta_0,\eta_1)^T\in\R^2,
    \label{eq:level-coordinate}
\end{equation*}
for \(J>J_0+1\), and set
\begin{equation}
    \omega_j(\boldsymbol{\eta})
    =\operatorname{logit}^{-1}(\eta_0+\eta_1t_j)
    =\frac{\exp(\eta_0+\eta_1t_j)}{1+\exp(\eta_0+\eta_1t_j)}.
    \label{eq:omega-logistic}
\end{equation}
Here \(\eta_0\) controls the Wendland weight at the coarsest detail level and \(\eta_1\) controls its change across resolution levels.  A positive \(\eta_1\) describes an increasing fitted Wendland weight under this parameterization; it does not imply that such an increase is present in every dataset.  If no trend is desired, setting \(\eta_1=0\) yields a common slab mixture weight across levels.  The fitted weights and their empirical identification are assessed in the component-and-adaptivity analysis.

The scale \(\beta_j\) defines the common prior truncation interval, but its fitted value is treated as a data-adaptive truncation-scale choice rather than as a formal estimator of an unknown true support. It may be specified from prior knowledge or selected from the empirical coefficients. We recommend the robust upper-quantile rule
\begin{equation}
    \widehat\beta_j^{(0)}
    =\max\left\{\widehat\sigma,
    Q_{\varkappa}\left(\{|d_{j,k}|:k\in\mathcal K_j\}\right)\right\},
    \qquad \varkappa\in(0,1),
    \label{eq:beta-quantile}
\end{equation}
where \(Q_{\varkappa}\) denotes the empirical \(\varkappa\)-quantile and \(\widehat\sigma\) is the robust noise-scale estimate defined in Section~\ref{sec:inference}. To ensure compatibility with the assumption \(\beta_j>0\), define
\[
    s_\beta=\max\left\{1,\widehat\sigma,
    \max_{\ell,k}|d_{\ell,k}|\right\},
    \qquad
    \beta_{\min}=\varepsilon_\beta s_\beta,
\]
where the maximum is taken over all detail levels and locations, and \(\varepsilon_\beta>0\) is a small numerical constant. We use \(\varepsilon_\beta=10^{-8}\). The fitted support scale is therefore
\begin{equation}
    \widehat\beta_j
    =\max\left\{\widehat\beta_j^{(0)},\beta_{\min}\right\}.
    \label{eq:positive-beta-floor}
\end{equation}
In the implementation, \(Q_{\varkappa}\) is the empirical quantile calculated by the exact quantile convention used in the R analysis (the default stats::quantile type 7 if unchanged).  The value \(\varkappa=0.99\) is therefore a data-adaptive truncation rule, not a formal estimator or guarantee of the true support.  At coarse resolutions, relatively few coefficients determine this quantile, and a future coefficient may exceed the fitted interval. This floor preserves \(\widehat\beta_j\geq\widehat\sigma\) while preventing a degenerate zero support when \(\widehat\sigma=0\) and all relevant coefficients vanish. The alternative \(\widehat\beta_j=\max_{k\in\mathcal K_j}|d_{j,k}|\) can be used when a conservative truncation scale is preferred, but it is more sensitive to isolated noise outliers and must likewise be subject to the positive floor in \eqref{eq:positive-beta-floor}. The default hyperparameter choices and their interpretations are summarized in Table~\ref{tab:hyperparameters}.  These settings are used throughout the primary simulation and real-data analyses unless a sensitivity analysis is explicitly identified.

\begin{table}[htbp]
    \centering
    \begin{threeparttable}
        \caption{Default prior hyperparameters and likelihood-specific tuning quantities.}
        \label{tab:hyperparameters}

        \begin{tabular}{lll}
            \toprule
            Quantity & Default & Interpretation \\
            \midrule
            \(J_0\)
                & \(0\)
                & Primary detail resolution level \\
            \(\ell\)
                & \(2\)
                & Spike-offset parameter controlling the initial spike probability \\
            \(\gamma\)
                & \(2.4\)
                & Controls the increase of sparsity with resolution \\
            \(\varkappa\)
                & \(0.99\)
                & Quantile used for \(\widehat{\beta}_j\) \\
            \(\widehat{\sigma}\)
                & Equation~\eqref{eq:mad-estimator}
                & Robust noise-scale estimate used in both likelihood versions \\
            \(\lambda\)
                & \(1/(2\widehat{\sigma}^{2})\)
                & Rate for the Laplace working likelihood \\
            \bottomrule
        \end{tabular}
    \end{threeparttable}
\end{table}

\section{Bayesian shrinkage rule and its properties}
\label{sec:properties}
Unless stated otherwise, all theoretical results in this section condition on fixed values of the prior and likelihood hyperparameters. Thus, the propositions describe the fixed-hyperparameter, or oracle, rule \(\delta_j(d;\pi_j,\omega_j,\beta_j,\lambda)\). The effect of estimating these quantities from the observed wavelet coefficients is treated separately in Section~\ref{sec:inference}.

\subsection{Gaussian likelihood formulation}
Under \eqref{eq:coefficient-model}, the coefficientwise likelihood is
\begin{equation}
    L_N(d\mid\theta,\sigma^2)
    =\frac{1}{\sqrt{2\pi\sigma^2}}
      \exp\left\{-\frac{(d-\theta)^2}{2\sigma^2}\right\}.
    \label{eq:normal-likelihood}
\end{equation}
For any nonnegative likelihood \(L(d\mid\theta)\), define the mixture-slab moments
\begin{equation*}
    M_{r,j}(d)
    =\int_{-\beta_j}^{\beta_j}
       \theta^r g_j(\theta;\omega_j,\beta_j)L(d\mid\theta)\,d\theta,
    \qquad r\in\{0,1\}.
    \label{eq:mixture-moments}
\end{equation*}
The corresponding component moments are
\begin{align*}
    M_{r,W,j}(d)
    &=\int_{-\beta_j}^{\beta_j}
      \theta^r g_W(\theta;\beta_j)L(d\mid\theta)\,d\theta,
    \\
    M_{r,S,j}(d)
    &=\int_{-\beta_j}^{\beta_j}
      \theta^r g_S(\theta;\beta_j)L(d\mid\theta)\,d\theta.
\end{align*}
By \eqref{eq:mixture-slab},
\begin{equation*}
    M_{r,j}(d)=\omega_jM_{r,W,j}(d)+(1-\omega_j)M_{r,S,j}(d).
    \label{eq:moment-decomposition}
\end{equation*}

\begin{proposition}[Posterior-mean shrinkage rule]
\label{prop:posterior-mean}
Under squared-error loss \(L_2(a,\theta)=(a-\theta)^2\), the Bayes estimator of \(\theta_{j,k}\) based on \(d_{j,k}=d\) is
\begin{equation}
    \delta_j(d)
    =\frac{(1-\pi_j)M_{1,j}(d)}
    {\pi_jL(d\mid0)+(1-\pi_j)M_{0,j}(d)}.
    \label{eq:general-shrinkage-rule}
\end{equation}
The posterior probabilities of the three latent components are
\begin{align}
    p_{0,j}(d)&=\frac{\pi_jL(d\mid0)}{D_j(d)},
    \label{eq:posterior-spike-probability}\\
    p_{W,j}(d)&=\frac{(1-\pi_j)\omega_jM_{0,W,j}(d)}{D_j(d)},
    \label{eq:posterior-wendland-probability}\\
    p_{S,j}(d)&=\frac{(1-\pi_j)(1-\omega_j)M_{0,S,j}(d)}{D_j(d)},
    \label{eq:posterior-semicircle-probability}
\end{align}
where
\begin{equation}
    D_j(d)=\pi_jL(d\mid0)+(1-\pi_j)M_{0,j}(d).
    \label{eq:posterior-denominator}
\end{equation}
Consequently,
\begin{equation}
    \delta_j(d)=p_{W,j}(d)\mu_{W,j}(d)+p_{S,j}(d)\mu_{S,j}(d),
    \label{eq:posterior-component-decomposition}
\end{equation}
where \(\mu_{W,j}(d)=M_{1,W,j}(d)/M_{0,W,j}(d)\) and \(\mu_{S,j}(d)=M_{1,S,j}(d)/M_{0,S,j}(d)\) are the posterior means conditional on the corresponding continuous component.
\end{proposition}

\begin{proof}
The marginal likelihood of \(d\) is \(D_j(d)\). The spike contributes zero to the posterior first moment, while the continuous mixture contributes \((1-\pi_j)M_{1,j}(d)\). Dividing the posterior first moment by the marginal likelihood gives \eqref{eq:general-shrinkage-rule}. The component probabilities follow by Bayes' theorem, and \eqref{eq:posterior-component-decomposition} follows by separating the two continuous components.
\end{proof}

\begin{remark}
Under squared-error loss, the proposed estimator is generally a continuous shrinkage rule rather than an exact hard threshold, so except at \(d=0\), the posterior mean is usually nonzero even though the posterior assigns positive probability to the exact spike. The quantity \(p_{0,j}(d)\) is the appropriate posterior measure of evidence for an exactly zero coefficient.
\end{remark}

\begin{proposition}[Continuity of moments and shrinkage rule]
\label{prop:continuity}
Assume that \(0<\pi_j<1\), \(\beta_j>0\), and that the prior slab is supported on \([ -\beta_j,\beta_j ]\). For either the Gaussian likelihood in \eqref{eq:normal-likelihood} with fixed \(\sigma^2>0\), or the fixed-rate Laplace likelihood in \eqref{eq:laplace-likelihood} with fixed \(\lambda>0\), the functions \(M_{r,j}(d)\), \(M_{r,W,j}(d)\), and \(M_{r,S,j}(d)\) are continuous in \(d\) for \(r\in\{0,1\}\). Moreover, \(D_j(d)\) is continuous and strictly positive for every \(d\in\mathbb{R}\), and the posterior-mean rule \(\delta_j(d)\) in \eqref{eq:general-shrinkage-rule} is continuous on \(\mathbb{R}\).
\end{proposition}

\begin{proof}
Fix \(d_0\in\mathbb{R}\) and let \(d\to d_0\). For either likelihood, the function \(L(d\mid\theta)\) is continuous in \(d\) for every \(\theta\), and it is bounded uniformly over \(d\) and \(\theta\) by a finite constant that depends only on the fixed likelihood parameter. Since \(|\theta|\leq\beta_j\) on the slab support, for \(r\in\{0,1\}\) the integrands defining the moments are dominated by an integrable function proportional to \(\beta_j^r g_j(\theta;\omega_j,\beta_j)\). Dominated convergence therefore gives continuity of \(M_{r,j}(d)\). The same argument applies separately to the Wendland and semicircle moments.

The likelihood is strictly positive at \(\theta=0\), and \(\pi_j>0\), so \(D_j(d)\geq \pi_jL(d\mid0)>0\) for every finite \(d\). Hence the numerator and denominator in \eqref{eq:general-shrinkage-rule} are continuous and the denominator never vanishes. Their quotient, \(\delta_j(d)\), is consequently continuous on \(\mathbb{R}\).
\end{proof}

\begin{proposition}[Monotonicity of the fixed Gaussian shrinkage rule]
\label{prop:gaussian-monotonicity}
Assume the Gaussian likelihood in \eqref{eq:normal-likelihood} with fixed \(\sigma^2>0\), and hold \(\pi_j\), \(\omega_j\), and \(\beta_j\) fixed.  Let
\[
    m_j(d)=\pi_jL_N(d\mid0,\sigma^2)
    +(1-\pi_j)M_{0,j}^{(N)}(d)
\]
denote the marginal density of \(d\).  Then
\[
    \delta_j(d)
    =d+\sigma^2\frac{d}{dd}\log m_j(d),
    \qquad
    \delta_j'(d)
    =\frac{\operatorname{Var}(\theta_{j,k}\mid d)}{\sigma^2}\geq0.
\]
Consequently, the posterior-mean map is nondecreasing in \(d\) for fixed hyperparameters.  This conclusion concerns the fixed-hyperparameter rule; it does not by itself establish monotonicity of an empirical-Bayes map whose fitted hyperparameters change with the complete coefficient vector.
\end{proposition}

\begin{proof}
For the Gaussian likelihood,
\[
    \frac{d}{dd}\log m_j(d)
    =\frac{\mathbb{E}(\theta_{j,k}\mid d)-d}{\sigma^2}.
\]
Consequently, \(\delta_j(d)=d+\sigma^2(\log m_j)'(d)\).  Differentiating the posterior mean under the integral sign gives
\[
    \delta_j'(d)
    =\operatorname{Cov}\left(\theta_{j,k},
      \frac{\theta_{j,k}-d}{\sigma^2}\,\middle|\,d\right)
    =\frac{\operatorname{Var}(\theta_{j,k}\mid d)}{\sigma^2}\geq0.
\]
The calculation holds for fixed hyperparameters.  It does not apply directly to an empirical-Bayes map after \(\widehat{\boldsymbol\eta}\), \(\widehat\beta_j\), or \(\widehat\sigma\) has been refitted as a function of the complete observed coefficient vector.
\end{proof}

\begin{remark}[Endpoint cases]
If \(\pi_j=1\), the prior is degenerate at zero and \(\delta_j(d)=0\) for every \(d\); \(\omega_j\) and \(\beta_j\) are irrelevant.  If \(\pi_j=0\), the spike is absent and
\[
    \delta_j(d)=\frac{M_{1,j}(d)}{M_{0,j}(d)}.
\]
If \(\omega_j=1\) or \(\omega_j=0\), the absent semicircle or Wendland component is omitted, respectively.  These endpoint formulas are used in the component-specific fits and in the numerical checks.
\end{remark}

Proposition~\ref{prop:continuity} establishes continuity of the component moments, the posterior denominator, and the resulting shrinkage rule for both likelihood specifications. Figure~\ref{fig:shrinkage-curves} illustrates how the posterior-mean shrinkage rule changes with the prior spike probability, the Wendland mixture weight, the support scale, and the estimated noise level under the Gaussian likelihood. The plots also show the resulting attenuation relative to the unshrunk rule.

The differences among the panels have distinct interpretations.  Increasing \(\pi_j\) increases the prior probability of the spike at zero and therefore strengthens shrinkage toward zero.  Increasing \(\omega_j\) gives more weight to the Wendland slab, which is more concentrated near zero and vanishes at the support boundary. Relative to the more dispersed semicircle slab, this generally produces greater attenuation of moderate coefficients.  Increasing \(\beta_j\) widens the admissible slab support and delays boundary saturation, allowing larger posterior means when the data support them.  Increasing \(\widehat{\sigma}\) makes a fixed observed coefficient less informative relative to the noise and therefore increases shrinkage.  For small \(|d|\), all curves remain close to zero because the spike is plausible; as \(|d|\) increases, the posterior probability of the slab increases and the curves move toward the diagonal.  The curves are odd because the slab and likelihood are symmetric, and the bounded slab prevents the posterior mean from exceeding the fitted support in magnitude.

\begin{figure}[htbp]
    \centering
    \includegraphics[width=\textwidth]{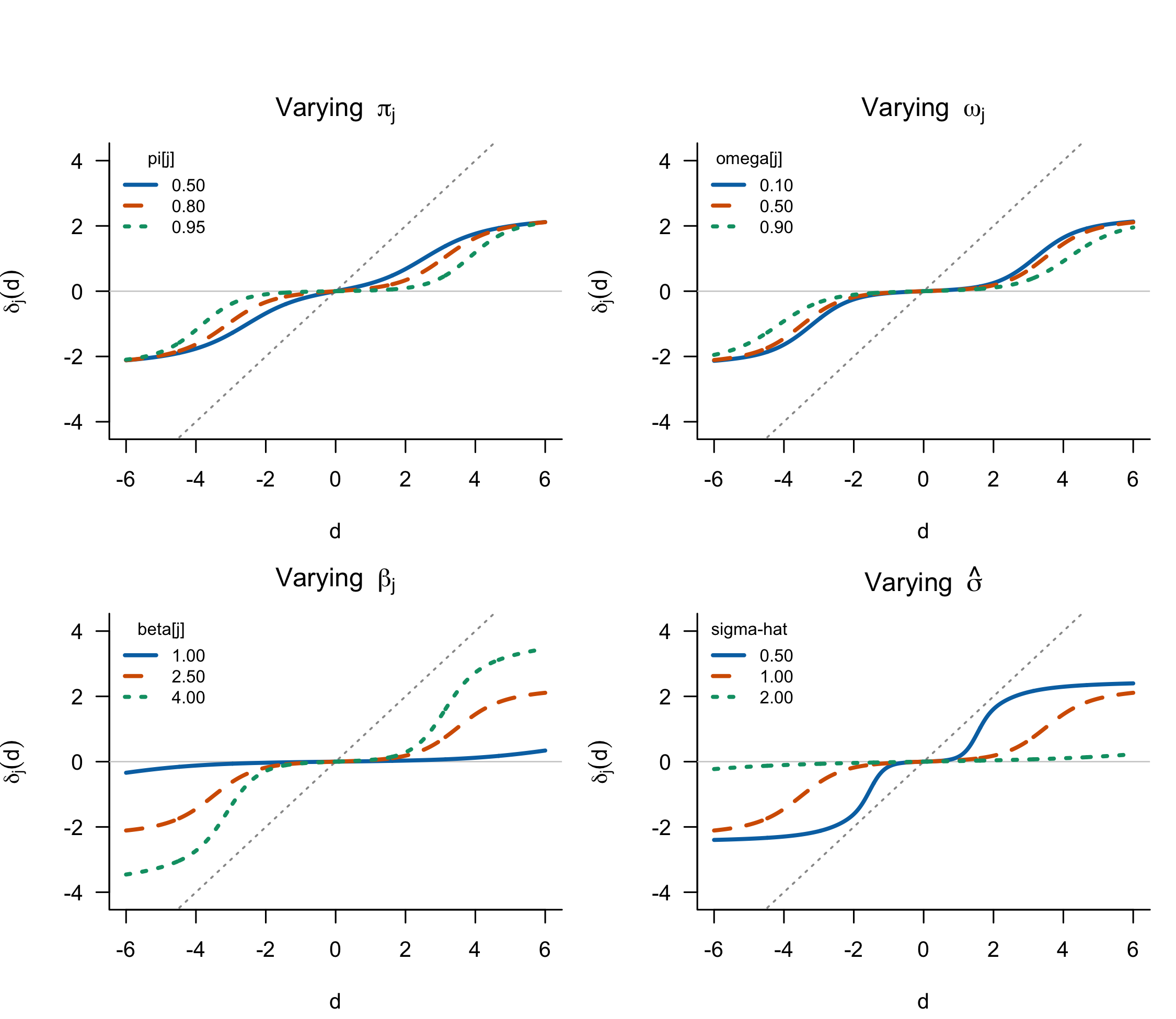}
    \caption{Posterior-mean shrinkage curves $\delta_j(d)$ for the proposed Wendland--semicircle prior under a Gaussian likelihood. Each panel varies one parameter while holding the others fixed: the prior spike probability $\pi_j$, the Wendland mixture weight $\omega_j$, the slab support scale $\beta_j$, and the noise-scale estimate $\widehat{\sigma}$. The gray diagonal line represents the unshrunk rule $\delta_j(d)=d$.}
    \label{fig:shrinkage-curves}
\end{figure}

\subsection{Laplace working likelihood and analytic Wendland contribution}
The normal likelihood in \eqref{eq:normal-likelihood} is the primary observation model. For a closed-form coefficientwise rule, we also consider the variance-mixture working likelihood used in related Bayesian wavelet-shrinkage constructions. Specifically, let \(v>0\) denote a coefficientwise latent variance and suppose
\begin{equation}
    d\mid\theta,v\sim\Normal(\theta,v),
    \qquad
    v\sim\Exp(\lambda),
    \qquad \lambda>0,
    \label{eq:variance-mixture}
\end{equation}
where \(\Exp(\lambda)\) denotes the exponential distribution with density
\[
    p(v\mid\lambda)=\lambda\exp(-\lambda v)I(v>0),
    \qquad \lambda>0.
\]
Thus, \(\lambda\) is the rate of the exponential distribution assigned to the latent variance \(v\); it is not the Gaussian noise variance \(\sigma^2\).
Integrating out \(v\) gives
\begin{equation}
    L_L(d\mid\theta,\lambda)
    =\frac{a}{2}\exp\{-a|d-\theta|\},
    \qquad a=\sqrt{2\lambda},
    \label{eq:laplace-likelihood}
\end{equation}
which is the density of a Laplace distribution centered at \(\theta\) with scale \(1/a\). The parameter \(\lambda\) is relevant only to the Laplace working likelihood.
In the empirical-Bayes implementation, \(\lambda\) is fixed at the value in \eqref{eq:lambda-estimator} and is not re-estimated. The numerical safeguard for a zero median absolute deviation (MAD) is described in the noise-scale subsection. The Laplace construction is therefore used as a coefficientwise working likelihood when the data are generated under the Gaussian model \eqref{eq:gaussian-errors}. The use of \eqref{eq:variance-mixture} should be interpreted as a coefficientwise working likelihood when the original data are generated under the common-variance Gaussian model \eqref{eq:gaussian-errors}.

For the Laplace likelihood, define the elementary exponential moment
\begin{equation*}
    \mathcal I_m(b;u,v)=\int_u^v t^m e^{bt}\,dt,
    \qquad m\in\{0,1,2,\ldots\},\quad b\ne0,\quad u<v.
    \label{eq:elementary-integral}
\end{equation*}
It has the finite-sum representation
\begin{equation}
    \mathcal I_m(b;u,v)
    =\left[
      e^{bt}\sum_{r=0}^{m}
      \frac{(-1)^r m!}{(m-r)!}\frac{t^{m-r}}{b^{r+1}}
      \right]_{t=u}^{t=v}.
    \label{eq:elementary-integral-closed}
\end{equation}

For a fixed \(d\), let \(\mathcal P(d,\beta)\) be the partition of \([-\beta,\beta]\) obtained by inserting the points \(0\) and \(d\) whenever they lie strictly inside the interval. On every interval \(I=[u_I,v_I]\in\mathcal P(d,\beta)\), define \(s_{\theta,I}\in\{-1,1\}\) as the constant sign of \(\theta\) on the interior of \(I\), and define \(s_{d,I}\in\{-1,1\}\) as the constant sign of \(d-\theta\) on the interior of \(I\). Values at partition points are irrelevant because they have Lebesgue measure zero. Let
\begin{equation*}
    (c_0,c_1,c_2,c_3,c_4,c_5)=(1,0,-10,20,-15,4).
    \label{eq:wendland-coefficients}
\end{equation*}

\begin{proposition}[Finite-sum Wendland moments]
\label{prop:wendland-moments}
Under the Laplace likelihood \eqref{eq:laplace-likelihood}, for \(r\in\{0,1\}\),
\begin{equation}
\begin{split}
    M_{r,W}(d;\beta,\lambda)
    ={}&\frac{3a}{4\beta}
    \sum_{I=[u_I,v_I]\in\mathcal P(d,\beta)}
    e^{-a s_{d,I}d}
    \sum_{m=0}^{5}c_m\beta^{-m}s_{\theta,I}^{m}
    \mathcal I_{r+m}(a s_{d,I};u_I,v_I),
\end{split}
\label{eq:wendland-moment-closed}
\end{equation}
where \(a=\sqrt{2\lambda}\) and \(\mathcal I_m\) is given by \eqref{eq:elementary-integral-closed}. Thus, both the zero-order marginal contribution and the first-order posterior numerator for the Wendland component are available without numerical integration.
\end{proposition}

\begin{proof}
By \eqref{eq:wendland-polynomial}, the Wendland density on each interval of the partition is
\[
    g_W(\theta;\beta)=\frac{3}{2\beta}
    \sum_{m=0}^{5}c_m\beta^{-m}s_{\theta,I}^{m}\theta^m.
\]
On the same interval, \(|d-\theta|=s_{d,I}(d-\theta)\), and hence
\[
    \exp\{-a|d-\theta|\}
    =e^{-as_{d,I}d}e^{as_{d,I}\theta}.
\]
Multiplication by \(\theta^r(a/2)\) and integration over each interval gives \eqref{eq:wendland-moment-closed}.
\end{proof}

For the semicircle component, the substitution \(\theta=\beta\cos t\), \(0\le t\le\pi\), yields the stable one-dimensional representation
\begin{equation}
    M_{r,S}(d;\beta,\lambda)
    =\frac{a}{\pi}\int_0^\pi
      (\beta\cos t)^r
      \exp\{-a|d-\beta\cos t|\}
      \sin^2(t)\,dt,
    \qquad r\in\{0,1\}.
    \label{eq:semicircle-moment-integral}
\end{equation}
For \(|d|\geq\beta\), the integrand has no interior kink. The treatment of the \(|d|<\beta\) case is specified in the numerical implementation paragraph below. Standard adaptive Gauss--Legendre quadrature provides a deterministic and accurate evaluation.

In our implementation, we use several numerical safeguards. In the finite-sum Wendland calculation, the case \(|b|<10^{-10}\) is evaluated by the limiting polynomial integral \(\mathcal I_m(0;u,v)=(v^{m+1}-u^{m+1})/(m+1)\), avoiding the powers of \(b^{-1}\). Otherwise, the finite-sum expression in \eqref{eq:elementary-integral-closed} is used. If it is non-finite, the code falls back to adaptive integration with relative tolerance \(10^{-10}\) and absolute tolerance \(10^{-12}\). Quadrature-based component moments are evaluated on the log scale using row-wise log-sum-exp calculations. The semicircle integral is split at \(t_d=\arccos(d/\beta)\) when the absolute-value kink is inside the integration interval. The independent quadrature study selected \(q=128\) for the remaining simulation-suite analyses.  The real-data and semi-synthetic seismic analyses used \(q=48\), as specified in Table~\ref{tab:seismic-settings}; these orders were fixed before evaluating the corresponding outcomes.  The selected orders were not chosen by inspecting the corresponding reconstruction results.

\subsection{Structural properties of the shrinkage rule}

\begin{proposition}[Symmetry, boundedness, and tail behavior]
\label{prop:shrinkage-properties}
Consider the prior in \eqref{eq:proposed-prior}, where \(0<\pi_j<1\), \(\beta_j>0\), and \(g_j(\cdot;\omega_j,\beta_j)\) is a symmetric probability density supported on \([-\beta_j,\beta_j]\).  Assume that the likelihood is finite and strictly positive for finite \(d\) and \(\theta\), and satisfies \(L(-d\mid-\theta)=L(d\mid\theta)\).  Then:
\begin{enumerate}[label=(\roman*)]
    \item \(\delta_j(-d)=-\delta_j(d)\) for every \(d\in\mathbb{R}\), and \(\delta_j(0)=0\);
    \item \(|\delta_j(d)|<\beta_j\) for every finite \(d\), and \(|\delta_j(d)|\leq\beta_j\) for all \(d\);
    \item if \(L=L_N\) is the Gaussian likelihood in \eqref{eq:normal-likelihood} with fixed \(\sigma^2>0\), and \(g_j(\theta;\omega_j,\beta_j)>0\) for every \(\theta\in(-\beta_j,\beta_j)\), then
    \[
        \lim_{d\to\infty}\delta_j(d)=\beta_j,
        \qquad
        \lim_{d\to-\infty}\delta_j(d)=-\beta_j;
    \]
    These endpoint values are limits of the posterior mean and are not attained at finite \(d\);
    \item if \(L=L_L\) is the fixed-rate Laplace likelihood in \eqref{eq:laplace-likelihood}, with \(a>0\), then the tail limits are
    \[
        \lim_{d\to\infty}\delta_j(d)
        =\frac{(1-\pi_j)A_{1,j}^{+}}
        {\pi_j+(1-\pi_j)A_{0,j}^{+}},
        \qquad
        \lim_{d\to-\infty}\delta_j(d)
        =\frac{(1-\pi_j)A_{1,j}^{-}}
        {\pi_j+(1-\pi_j)A_{0,j}^{-}},
    \]
    where
    \[
        A_{r,j}^{+}=\int_{-\beta_j}^{\beta_j}\theta^r g_j(\theta;\omega_j,\beta_j)e^{a\theta}\,d\theta,
        \qquad
        A_{r,j}^{-}=\int_{-\beta_j}^{\beta_j}\theta^r g_j(\theta;\omega_j,\beta_j)e^{-a\theta}\,d\theta,
        \qquad r\in\{0,1\};
    \]
    \item \(p_{0,j}(d)\), \(p_{W,j}(d)\), and \(p_{S,j}(d)\) are even functions of \(d\), whereas \(\mu_{W,j}(d)\) and \(\mu_{S,j}(d)\) are odd functions of \(d\).
\end{enumerate}
\end{proposition}

The proof is given in the Appendix. 

\begin{remark}
The bounded-support property is a modeling choice, not a universal truth. It can reduce the influence of extreme noisy coefficients in low-SNR problems, but it also implies saturation, i.e., a genuinely large coefficient is eventually estimated near \(\pm\beta_j\). For signals with unusually large dynamic range, the value of \(\beta_j\) should therefore be examined carefully, and an unbounded slab should be considered as a sensitivity analysis.
\end{remark}

\subsection{Fixed-hyperparameter bias, variance, and risk}
For the two likelihood specifications, define the likelihood-specific moments as:
\begin{align*}
    M_{r,j}^{(N)}(d;\omega_j,\beta_j,\sigma^2)
    &=
    \int_{-\beta_j}^{\beta_j}
    \theta^r g_j(\theta;\omega_j,\beta_j)
    L_N(d\mid\theta,\sigma^2)\,d\theta,
    \\
    M_{r,j}^{(L)}(d;\omega_j,\beta_j,\lambda)
    &=
    \int_{-\beta_j}^{\beta_j}
    \theta^r g_j(\theta;\omega_j,\beta_j)
    L_L(d\mid\theta,\lambda)\,d\theta,
\end{align*}
for \(r\in\{0,1\}\). The fixed-hyperparameter shrinkage maps under the Gaussian and Laplace likelihoods are, respectively,
\begin{align}
    \delta_{j,N}(d;\pi_j,\omega_j,\beta_j,\sigma^2)
    &=
    \frac{(1-\pi_j)M_{1,j}^{(N)}(d;\omega_j,\beta_j,\sigma^2)}
    {\pi_jL_N(d\mid0,\sigma^2)
    +(1-\pi_j)M_{0,j}^{(N)}(d;\omega_j,\beta_j,\sigma^2)},
    \label{eq:normal-fixed-shrinkage-map}\\
    \delta_{j,L}(d;\pi_j,\omega_j,\beta_j,\lambda)
    &=
    \frac{(1-\pi_j)M_{1,j}^{(L)}(d;\omega_j,\beta_j,\lambda)}
    {\pi_jL_L(d\mid0,\lambda)
    +(1-\pi_j)M_{0,j}^{(L)}(d;\omega_j,\beta_j,\lambda)}.
    \label{eq:laplace-fixed-shrinkage-map}
\end{align}
Let
\[
    \boldsymbol{\zeta}_{j,N}
    =(\pi_j,\omega_j,\beta_j,\sigma^2),
    \qquad
    \boldsymbol{\zeta}_{j,L}
    =(\pi_j,\omega_j,\beta_j,\lambda)
\]
denote the fixed hyperparameter vectors for the Gaussian and Laplace
likelihoods, respectively. For \(q\in\{N,L\}\), let
\(\delta_{j,q}(d;\boldsymbol{\zeta}_{j,q})\) denote the corresponding map in
\eqref{eq:normal-fixed-shrinkage-map} or
\eqref{eq:laplace-fixed-shrinkage-map}.

Let \(D_{j,k}\) denote a random empirical coefficient generated under a fixed
true coefficient value \(\theta\). Because the simulations use the Gaussian
data-generating model, define \(\E_{\theta}^{G}\) and
\(\Var_{\theta}^{G}\) with respect to
\(D_{j,k}\sim\Normal(\theta,\sigma^2)\), including when the Laplace working
likelihood is used. For \(q\in\{N,L\}\), the pointwise fixed-hyperparameter
bias, variance, and frequentist risk are
\begin{align}
    B_{j,q}^{G}(\theta;\boldsymbol{\zeta}_{j,q})
    &=
    \E_{\theta}^{G}
    \left\{
    \delta_{j,q}(D_{j,k};\boldsymbol{\zeta}_{j,q})
    \right\}
    -\theta,
    \label{eq:bias}\\
    V_{j,q}^{G}(\theta;\boldsymbol{\zeta}_{j,q})
    &=
    \Var_{\theta}^{G}
    \left\{
    \delta_{j,q}(D_{j,k};\boldsymbol{\zeta}_{j,q})
    \right\},
    \label{eq:variance}\\
    R_{j,q}^{G}(\theta;\boldsymbol{\zeta}_{j,q})
    &=
    \E_{\theta}^{G}
    \left[
    \left\{
    \delta_{j,q}(D_{j,k};\boldsymbol{\zeta}_{j,q})-\theta
    \right\}^{2}
    \right] \nonumber\\
    &=
    \left\{
    B_{j,q}^{G}(\theta;\boldsymbol{\zeta}_{j,q})
    \right\}^{2}
    +
    V_{j,q}^{G}(\theta;\boldsymbol{\zeta}_{j,q}),
    \label{eq:risk}
\end{align}
where the superscript \(G\) denotes the data-generating law \(D_{j,k}\sim\mathcal{N}(\theta,\sigma^2)\), not necessarily the likelihood used by the fitted shrinkage rule.  Thus, \(B_{j,N}^{G}\), \(V_{j,N}^{G}\), and \(R_{j,N}^{G}\) describe the Gaussian likelihood under Gaussian data, whereas \(B_{j,L}^{G}\), \(V_{j,L}^{G}\), and \(R_{j,L}^{G}\) describe the Laplace working likelihood when the data are still generated under Gaussian errors.  When the data-generating law is Laplace, expectations and variances must instead be taken under the variance-matched Laplace distribution; the Gaussian variance identity and the Gaussian monotonicity result are not used for those rows.

For completeness, let
\[
    \mathcal{D}
    =
    \{d_{j,k}:j=J_0,\ldots,J-1,\ k\in\mathcal K_j\}
\]
denote the complete collection of observed detail coefficients. The two
empirical-Bayes plug-in rules are
\begin{align}
    \widehat\delta_{j,k,N}^{\mathrm{EB}}(\mathcal D)
    &=\delta_{j,N}\left(
    d_{j,k};\pi_j,\widehat\omega_{j,N},\widehat\beta_j,\widehat\sigma^2
    \right),
    \label{eq:normal-eb-shrinkage-rule}\\
    \widehat\delta_{j,k,L}^{\mathrm{EB}}(\mathcal D)
    &=\delta_{j,L}\left(
    d_{j,k};\pi_j,\widehat\omega_{j,L},\widehat\beta_j,\widehat\lambda
    \right).
    \label{eq:laplace-eb-shrinkage-rule}
\end{align}
Here \(\widehat\omega_{j,N}\) and \(\widehat\omega_{j,L}\) are defined in \eqref{eq:eb-omega-estimates}, and \(\widehat\lambda\) is the fixed rate in \eqref{eq:lambda-estimator}. For \(q\in\{N,L\}\), the end-to-end empirical-Bayes risk is
\begin{equation}
    R_{j,k,q}^{\mathrm{EB},G}(\boldsymbol{\theta})
    =
    \E_{\boldsymbol{\theta}}^{G}
    \left[
    \left\{
    \widehat\delta_{j,k,q}^{\mathrm{EB}}(\mathcal D)
    -\theta_{j,k}
    \right\}^{2}
    \right],
    \label{eq:eb-risk}
\end{equation}
where the expectation is over the joint distribution of all empirical coefficients used to construct \(\mathcal D\). The primary Gaussian-error simulation study estimates $R_{j,k,q}^{\mathrm{EB},G}$ through repeated data generation and refitting. The supplementary Laplace-error sensitivity analysis uses a different data-generating law and is reported separately; it should not be interpreted as an estimate of $R_{j,k,q}^{\mathrm{EB},G}$. The supplementary risk tables therefore identify the data-generating law and the working likelihood separately for every row.

\section{Inference and implementation}
\label{sec:inference}

\subsection{Noise-scale estimation}

The noise standard deviation is estimated from the finest-resolution detail coefficients. Let \(\operatorname{med}(\cdot)\) denote the sample median. Define
\begin{equation}
    \widehat\sigma
    =\frac{\operatorname{med}\left\{
       |d_{J-1,k}-\operatorname{med}(d_{J-1,k})|:
       k\in\mathcal K_{J-1}\right\}}{0.6745}.
    \label{eq:mad-estimator}
\end{equation}
If the raw MAD estimate is zero, replace \(\widehat\sigma\) by \(\max\{\widehat\sigma,\sqrt{\epsilon_{\mathrm{mach}}}\}\) before calculating any likelihood-specific quantity.  The resulting positive value is used below.  The constant \(0.6745\) is the approximate median of \(|Z|\) for \(Z\sim\mathcal{N}(0,1)\).  A pooled MAD over the two or three finest levels is a prespecified alternative when the finest level may contain substantial signal; it is reported only for analyses in which that alternative is actually fitted.

For the Gaussian likelihood, \(\widehat\sigma^2\) is inserted into \eqref{eq:normal-likelihood}. For the Laplace working likelihood, the default rate is
\begin{equation}
    \widehat\lambda=\frac{1}{2\widehat\sigma^2},
    \qquad
    \widehat a=\sqrt{2\widehat\lambda}=\frac{1}{\widehat\sigma},
    \label{eq:lambda-estimator}
\end{equation}
so that the working Laplace scale is \(\widehat\sigma\). For the processed seismic trace, \(\widehat\sigma\) is a wavelet-domain scale estimate conditional on the supplied preprocessing; it is not interpreted as a direct estimate of physical measurement uncertainty.

\subsection{Empirical-Bayes estimation of the slab-shape parameters}
For fixed \(\pi_j\), \(\beta_j\), and likelihood-specific parameters, define
the Gaussian and Laplace marginal densities by
\begin{align}
    m_{j,N}(d;\boldsymbol{\eta})
    &
    =\pi_jL_N(d\mid0,\sigma^2)
    +(1-\pi_j)M_{0,j}^{(N)}
    \left(d;\omega_j(\boldsymbol{\eta}),\beta_j,\sigma^2\right),
    \label{eq:normal-marginal-density}\\
    m_{j,L}(d;\boldsymbol{\eta})
    &
    =\pi_jL_L(d\mid0,\lambda)
    +(1-\pi_j)M_{0,j}^{(L)}
    \left(d;\omega_j(\boldsymbol{\eta}),\beta_j,\lambda\right),
    \label{eq:laplace-marginal-density}
\end{align}
where \(M_{0,j}^{(N)}\) and \(M_{0,j}^{(L)}\) are the mixture-slab moments
defined in Section~\ref{sec:properties}. The empirical-Bayes estimate of
\(\boldsymbol{\eta}\) is obtained separately for the two likelihood
specifications by maximizing
\begin{align}
    \widehat{\boldsymbol{\eta}}_{N}
    &=\arg\max_{\boldsymbol{\eta}\in[-12,12]^2}
    \ell_N(\boldsymbol{\eta}),
    &
    \ell_N(\boldsymbol{\eta})
    &=\sum_{j=J_0}^{J-1}\sum_{k\in\mathcal K_j}
    \log m_{j,N}(d_{j,k};\boldsymbol{\eta}),
    \label{eq:normal-eb-objective}\\
    \widehat{\boldsymbol{\eta}}_{L}
    &=\arg\max_{\boldsymbol{\eta}\in[-12,12]^2}
    \ell_L(\boldsymbol{\eta}),
    &
    \ell_L(\boldsymbol{\eta})
    &=\sum_{j=J_0}^{J-1}\sum_{k\in\mathcal K_j}
    \log m_{j,L}(d_{j,k};\boldsymbol{\eta}).
    \label{eq:laplace-eb-objective}
\end{align}
For numerical stability, each marginal density is evaluated on the log scale
using a log-sum-exp operation. In all reported analyses, the optimization is
performed with a bounded quasi-Newton method over \([-12,12]^2\). This
restriction prevents numerical probabilities indistinguishable from zero or
one while having negligible practical effect on the fitted mixture weights. For every fitted replicate, the implementation records the optimizer convergence status, whether either optimization coordinate reaches its bound, whether the objective or any component moment is non-finite, whether a numerical fallback is used, the endpoint branch if applicable, and the fitted values of \(\widehat{\boldsymbol\eta}\), \(\widehat\omega_j\), \(\widehat\beta_j\), and the likelihood-specific noise-scale quantity.  These fit-level diagnostics are summarized in the Supplementary Materials together with the warning messages and the number of finite performance measures.

The fitted estimators are obtained by an empirical-Bayes plug-in operation. The Gaussian and Laplace mixture weights are
\begin{equation}
    \widehat\omega_{j,N}
    =\omega_j(\widehat{\boldsymbol{\eta}}_N),
    \qquad
    \widehat\omega_{j,L}
    =\omega_j(\widehat{\boldsymbol{\eta}}_L).
    \label{eq:eb-omega-estimates}
\end{equation}
The values of \(\pi_j\), determined by the fixed hyperparameters \(\ell\) and
\(\gamma\), are not estimated from the data. The denoised signal is reconstructed by applying the inverse DWT to the retained scaling coefficients and the appropriate estimator in \eqref{eq:normal-eb-shrinkage-rule} or \eqref{eq:laplace-eb-shrinkage-rule}.

\subsection{Fitting algorithm}

The complete fitting procedure is as follows.
\begin{enumerate}[label=\textbf{Step \arabic*:}]
    \item Apply an orthogonal DWT to \(\bfy\), obtaining empirical detail coefficients \(d_{j,k}\) and scaling coefficients.
    \item Estimate \(\widehat\sigma\) using \eqref{eq:mad-estimator} and select the Gaussian likelihood \((q=N)\) or the Laplace working likelihood \((q=L)\), with the corresponding quantities defined above.
    \item Set \(\pi_j\), \(\widehat\beta_j^{(0)}\), and \(\widehat\beta_j\) using \eqref{eq:spike-probability}, \eqref{eq:beta-quantile}, and \eqref{eq:positive-beta-floor}; retain \(\varkappa\), \(\varepsilon_\beta\), and the selected quadrature order as part of the fitting specification.
    \item Apply the endpoint branches before evaluating unnecessary component expressions.  If \(\pi_j=1\), set the posterior mean and continuous-component probabilities to zero for that level.  If \(\pi_j=0\), omit the spike term.  If \(\omega_j=1\) or \(\omega_j=0\), evaluate only the Wendland or semicircle component, respectively.
     \item For an interior mixture, evaluate the mixture moments and marginal densities for a candidate \(\boldsymbol\eta\) as defined in \eqref{eq:normal-marginal-density}--\eqref{eq:laplace-marginal-density.  Use log-sum-exp calculations and the stated numerical safeguards.}
     \item If the Gaussian likelihood is selected, maximize \eqref{eq:normal-eb-objective} to obtain \(\widehat{\boldsymbol\eta}_N\); if the Laplace working likelihood is selected, maximize \eqref{eq:laplace-eb-objective} to obtain \(\widehat{\boldsymbol\eta}_L\).  Skip this optimization for a fixed endpoint specification.
    \item Compute the posterior component probabilities and posterior mean for every detail coefficient using the appropriate interior or endpoint map, and save the branch, convergence status, boundary-hit status, warnings, and support diagnostics.
    \item Apply the inverse DWT, retaining the empirical scaling coefficients, to obtain the denoised estimate \(\widehat f\). For a padded input, retain only the original observations before calculating reported application-level summaries.
\end{enumerate}

The overall fitting workflow is summarized in Figure~\ref{fig:fitting-algorithm}.

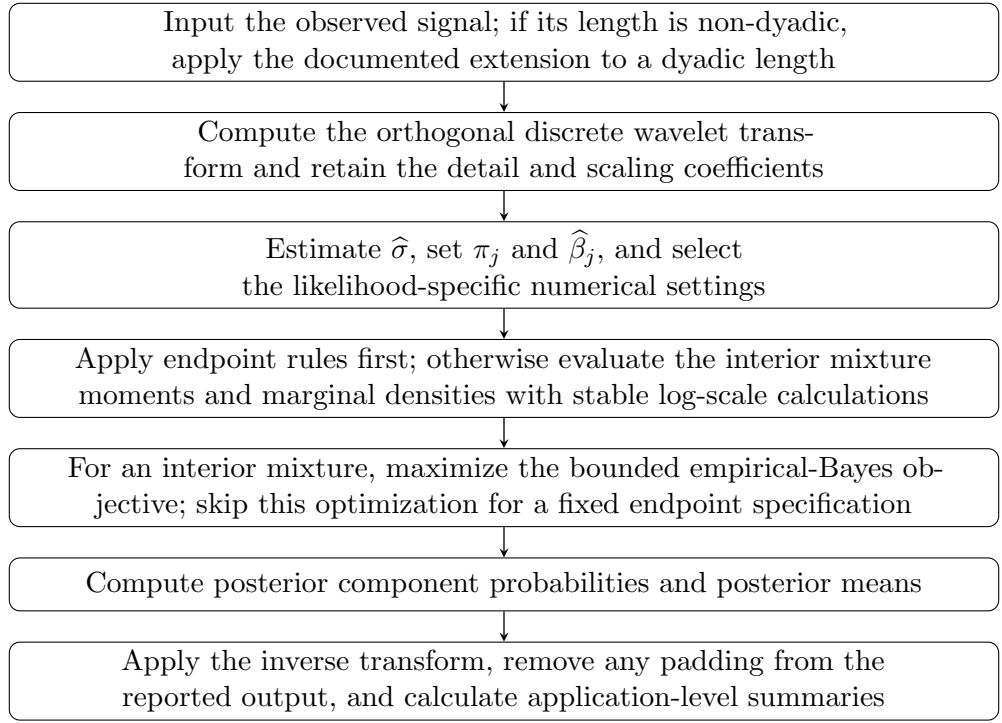
\begin{figure}[htbp]
    \centering
    \begin{tikzpicture}[
        node distance=4mm,
        box/.style={draw, rounded corners, align=center, text width=0.78\linewidth, minimum height=7mm, inner sep=3pt},
        line/.style={->,>=stealth, line width=0.5pt}
    ]
        \node[box] (input) {Input the observed signal; if its length is non-dyadic, apply the documented extension to a dyadic length};
        \node[box, below=of input] (dwt) {Compute the orthogonal discrete wavelet transform and retain the detail and scaling coefficients};
        \node[box, below=of dwt] (hyper) {Estimate \(\widehat{\sigma}\), set \(\pi_j\) and \(\widehat{\beta}_j\), and select the likelihood-specific numerical settings};
        \node[box, below=of hyper] (branch) {Apply endpoint rules first; otherwise evaluate the interior mixture moments and marginal densities with stable log-scale calculations};
        \node[box, below=of branch] (fit) {For an interior mixture, maximize the bounded empirical-Bayes objective; skip this optimization for a fixed endpoint specification};
        \node[box, below=of fit] (posterior) {Compute posterior component probabilities and posterior means};
        \node[box, below=of posterior] (output) {Apply the inverse transform, remove any padding from the reported output, and calculate application-level summaries};
        \draw[line] (input) -- (dwt);
        \draw[line] (dwt) -- (hyper);
        \draw[line] (hyper) -- (branch);
        \draw[line] (branch) -- (fit);
        \draw[line] (fit) -- (posterior);
        \draw[line] (posterior) -- (output);
    \end{tikzpicture}
    \caption{Flowchart for the empirical-Bayes wavelet fitting procedure.  Endpoint branches are evaluated before unnecessary component calculations; for a padded record, truncation to the original observations occurs before application-level diagnostics are reported.}
    \label{fig:fitting-algorithm}
\end{figure}

\clearpage


\section{Simulation study}
\label{sec:simulation}
The simulation study had two purposes. The primary experiment compared reconstruction accuracy and computational cost over standard test signals, sample sizes, and signal-to-noise ratios.  The supplementary experiments then examined the contribution of the two slab components, the resolution-adaptive mixture weight, numerical quadrature, fixed hyperparameters, support saturation, the wavelet basis and coarsest shrunk resolution, non-Gaussian errors, and semi-synthetic seismic recovery (Supplementary Section~3.2).  All methods were evaluated on the same simulated realization within each scenario, so that paired comparisons could be formed directly.

\subsection{Test functions and data generation}
Let \(f\) be one of the four standard test functions of \citet{DonohoJohnstone1994}---Bumps, Blocks, Doppler, and HeaviSine.  On \([0,1]\), the functions used in the simulations were as follows.  Let \(\operatorname{sgn}(u)\) denote the sign function.  The Blocks function was
\begin{equation*}
    f_{\mathrm{Blocks}}(x)=\sum_{r=1}^{11}h_r\frac{1+\operatorname{sgn}(x-t_r)}{2},
    \label{eq:blocks-function}
\end{equation*}
where
\begin{align*}
    (t_1,\ldots,t_{11})={}&(0.10,0.13,0.15,0.23,0.25,0.40,0.44,0.65,0.76,0.78,0.81),\\
    (h_1,\ldots,h_{11})={}&(4,-5,3,-4,5,-4.2,2.1,4.3,-3.1,2.1,-4.2).
\end{align*}
The Bumps function was
\begin{equation*}
    f_{\mathrm{Bumps}}(x)=\sum_{r=1}^{11}h_r\left\{1+\left|\frac{x-t_r}{w_r}\right|\right\}^{-4},
    \label{eq:bumps-function}
\end{equation*}
where
\begin{align*}
    (t_1,\ldots,t_{11})={}&(0.10,0.13,0.15,0.23,0.25,0.40,0.44,0.65,0.76,0.78,0.81),\\
    (h_1,\ldots,h_{11})={}&(4,5,3,4,5,4.2,2.1,4.3,3.1,5.1,4.2),\\
    (w_1,\ldots,w_{11})={}&(0.005,0.005,0.006,0.010,0.010,0.030,0.010,0.010,0.005,0.008,0.005).
\end{align*}
The Doppler and HeaviSine functions were
\begin{align*}
    f_{\mathrm{Doppler}}(x)&=\sqrt{x(1-x)}\sin\left\{\frac{2\pi(1+0.05)}{x+0.05}\right\},\\
    f_{\mathrm{HeaviSine}}(x)&=4\sin(4\pi x)-\operatorname{sgn}(x-0.3)-\operatorname{sgn}(0.72-x).
\end{align*}

For each function, the observations were generated at \(x_i=(i-1)/(n-1)\), with \(n\in\{512,1024,2048\}\), according to
\begin{equation*}
    y_i^{(r)}=f(x_i)+\varepsilon_i^{(r)},
    \qquad \varepsilon_i^{(r)}\stackrel{\mathrm{iid}}{\sim}\mathcal{N}(0,\sigma^2),
    \qquad i=1,\ldots,n,
    \label{eq:simulation-data}
\end{equation*}
where \(r=1,\ldots,R\) indexes the \(R=100\) independent replications.  The primary grid used \(\rho\in\{0.2,1,3\}\) for all three sample sizes.  To address the sensitivity of conclusions to the very low-SNR setting, the expanded primary grid additionally used \(\rho\in\{0.5,2,5\}\) at \(n=1024\).  For every signal and sample size, the noise scale was set to
\begin{equation*}
    \sigma=\frac{\operatorname{SD}\{f(x_i):i=1,\ldots,n\}}{\rho}.
    \label{eq:snr-definition}
\end{equation*}
Thus, \(\rho\) is an amplitude or standard-deviation SNR, not a power ratio.  The corresponding power SNR is \(\rho^2\), and the decibel representation is \(20\log_{10}(\rho)\).  For example, \(\rho=0.2\) corresponds to a power SNR of \(0.04\) and \(-13.98\) dB, whereas \(\rho=5\) corresponds to a power SNR of \(25\) and \(13.98\) dB.  The value \(0.2\) was retained as a deliberately difficult stress test, while the additional values \(0.5\), \(2\), and \(5\) provide a wider range of more moderate and favorable conditions.  A common random-number construction was used within every scenario: one noise vector was generated for the scenario and passed to all fitted methods.

\subsection{Methods and implementation}
The primary method was the adaptive Wendland--semicircle mixture with a Gaussian coefficientwise likelihood, denoted WS--Gaussian or WS--G.  The same prior fitted with a Laplace working likelihood, denoted WS--Laplace or WS--L, was included as a robustness and likelihood-sensitivity comparison.  The benchmark methods were universal soft thresholding (Univ) \citep{DonohoJohnstone1994}, false-discovery-rate thresholding (FDR) \citep{AbramovichBenjamini1996}, cross-validation thresholding (CV) \citep{Nason1996}, Stein's unbiased risk estimate thresholding (SURE) \citep{DonohoJohnstone1995}, the Bayesian adaptive multiresolution shrinker (BAMS) \citep{VidakovicRuggeri2001}, and the nonlocal-prior method (NLP) \citep{Sanyal2025}.  The comparison was restricted to methods with working software implementations in the analysis environment.

All methods used a Daubechies extremal-phase wavelet with filter number 10, periodic boundary handling, and \(J_0=0\), with the alternative basis and \(J_0\) checks reported in Supplementary Section~2.6.  The proposed fits used \(\ell=2\), \(\gamma=2.4\), \(\varkappa=0.99\), and the quadrature order selected by the independent quadrature study, \(q=128\).  The BAMS calculation used 32 quadrature points.  The finite-sum Wendland calculation was used where available under the Laplace working likelihood, and the specified quadrature order was used for the remaining component calculations.  The four classical procedures used soft thresholding.  NLP was fitted with the \texttt{NLPwavelet} implementation using its nonlocal-prior mixture, generalized-logit mixture probabilities, double-exponential scale specification, and default values for other settings.  The fixed hyperparameters and their perturbations are examined in Supplementary Section~2.5.

\subsection{Performance measures}
Let \(\widehat f_i^{(r)}\) denote the reconstruction in replication \(r\).  We recorded the replication-specific mean squared error (MSE), mean absolute error (MAE), maximum absolute error, truth-based output SNR, and elapsed fitting time.  In addition to the raw MSE, we used the normalized MSE
\begin{equation*}
    \operatorname{NMSE}^{(r)}=\frac{\operatorname{MSE}^{(r)}}{\operatorname{Var}\{f(x_i):i=1,\ldots,n\}},
    \label{eq:simulation-nmse}
\end{equation*}
which permits a more interpretable descriptive comparison across test functions with different scales.  The truth-based output SNR was
\begin{equation*}
    \operatorname{SNR}_{\mathrm{out}}^{(r)}=\frac{\operatorname{SD}\{f(x_i):i=1,\ldots,n\}}{\operatorname{SD}\{\widehat f_i^{(r)}-f(x_i):i=1,\ldots,n\}}.
    \label{eq:simulation-output-snr}
\end{equation*}
For a competitor \(m\), the paired MSE and NMSE differences were defined as
\begin{equation*}
    \Delta_{m}^{\mathrm{MSE},(r)}=\operatorname{MSE}_{\mathrm{WS-G}}^{(r)}-\operatorname{MSE}_{m}^{(r)},
    \qquad
    \Delta_{m}^{\mathrm{NMSE},(r)}=\operatorname{NMSE}_{\mathrm{WS-G}}^{(r)}-\operatorname{NMSE}_{m}^{(r)}.
    \label{eq:simulation-paired-difference}
\end{equation*}
The reporting script summarizes each mean with its Monte Carlo standard error and a replicate-level 95\% \(t\)-interval.  Negative paired differences favor WS--G, whereas a paired win rate is the proportion of common replications for which WS--G has the smaller error.  Study-level fit counts and numerical warnings are reported in Supplementary Table~3.

\subsection{Primary Gaussian-error results}
\label{sec:simulation-results}
The original 36-cell Gaussian-error grid contains four signals, three sample sizes, and three SNR values.  Table~\ref{tab:simulation-mse} reports the complete cellwise MSE comparison; entries in parentheses are Monte Carlo standard errors.  The WS--G fit was the best non-NLP method in 24 of the 36 original cells and in 11 of the 12 low-SNR cells (the four signals at each of the three sample sizes). The cellwise results also make explicit how the relative standing of the proposed fits and the benchmark methods changes with signal structure, sample size, and noise level.

\begin{landscape}
\begin{table}[p]
    \centering
    \caption{Mean MSE across 100 independent replications for the original Gaussian-error grid.  Each entry is the mean MSE with its Monte Carlo standard error in parentheses; the smallest MSE in each cell is shown in bold.}
    \label{tab:simulation-mse}
    \scriptsize
    \renewcommand{\arraystretch}{0.92}
    \resizebox{0.85\linewidth}{!}{%
    \begin{tabular}{llrrrrrrrr}
        \toprule
        Signal & SNR & WS--G & WS--L & Univ & FDR & CV & SURE & BAMS & NLP \\
        \midrule
        \multicolumn{10}{c}{\(n=512\)} \\
                \midrule
        Bumps & 0.2 & 0.446 (0.005) & 0.444 (0.003) & 0.450 (0.003) & 0.451 (0.003) & 0.456 (0.004) & 0.554 (0.026) & \textbf{0.441 (0.003)} & 0.627 (0.023) \\
        Bumps & 1.0 & 0.227 (0.003) & 0.357 (0.002) & 0.370 (0.002) & 0.319 (0.003) & 0.180 (0.002) & 0.187 (0.003) & 0.382 (0.002) & \textbf{0.133 (0.001)} \\
        Bumps & 3.0 & 0.049 (0.001) & 0.148 (0.002) & 0.179 (0.001) & 0.091 (0.001) & 0.051 (0.000) & 0.038 (0.001) & 0.234 (0.002) & \textbf{0.026 (0.000)} \\
        Blocks & 0.2 & \textbf{3.305 (0.053)} & 3.610 (0.030) & 3.770 (0.027) & 3.822 (0.024) & 3.716 (0.042) & 4.445 (0.146) & 3.558 (0.028) & 3.946 (0.213) \\
        Blocks & 1.0 & 0.898 (0.012) & 1.439 (0.010) & 1.667 (0.014) & 1.554 (0.018) & 0.871 (0.009) & 0.933 (0.013) & 1.514 (0.010) & \textbf{0.582 (0.007)} \\
        Blocks & 3.0 & 0.307 (0.003) & 0.524 (0.005) & 0.715 (0.005) & 0.481 (0.006) & 0.217 (0.002) & 0.241 (0.004) & 0.706 (0.008) & \textbf{0.172 (0.002)} \\
        Doppler & 0.2 & \textbf{0.077 (0.001)} & 0.081 (0.001) & 0.087 (0.001) & 0.087 (0.001) & 0.087 (0.001) & 0.108 (0.004) & 0.081 (0.001) & 0.088 (0.004) \\
        Doppler & 1.0 & 0.015 (0.000) & 0.031 (0.001) & 0.038 (0.000) & 0.033 (0.001) & 0.017 (0.000) & 0.018 (0.000) & 0.050 (0.001) & \textbf{0.011 (0.000)} \\
        Doppler & 3.0 & 0.003 (0.000) & 0.008 (0.000) & 0.011 (0.000) & 0.007 (0.000) & 0.003 (0.000) & 0.003 (0.000) & 0.017 (0.000) & \textbf{0.002 (0.000)} \\
        HeaviSine & 0.2 & \textbf{4.798 (0.256)} & 7.060 (0.161) & 8.072 (0.162) & 9.080 (0.093) & 7.095 (0.231) & 10.096 (0.829) & 7.784 (0.116) & 5.310 (0.314) \\
        HeaviSine & 1.0 & 0.475 (0.017) & 0.904 (0.013) & 1.403 (0.023) & 1.746 (0.074) & 0.791 (0.018) & 0.851 (0.022) & 0.927 (0.018) & \textbf{0.323 (0.017)} \\
        HeaviSine & 3.0 & 0.111 (0.002) & 0.158 (0.002) & 0.319 (0.004) & 0.362 (0.015) & 0.153 (0.002) & 0.163 (0.003) & 0.312 (0.005) & \textbf{0.075 (0.002)} \\
                \addlinespace
        \multicolumn{10}{c}{\(n=1024\)} \\
                \midrule
        Bumps & 0.2 & \textbf{0.431 (0.003)} & 0.447 (0.001) & 0.456 (0.001) & 0.456 (0.001) & 0.453 (0.002) & 0.485 (0.007) & 0.444 (0.001) & 0.460 (0.010) \\
        Bumps & 1.0 & 0.151 (0.001) & 0.301 (0.002) & 0.313 (0.002) & 0.248 (0.002) & 0.130 (0.001) & 0.140 (0.001) & 0.350 (0.001) & \textbf{0.092 (0.001)} \\
        Bumps & 3.0 & 0.029 (0.000) & 0.083 (0.001) & 0.108 (0.001) & 0.058 (0.001) & 0.025 (0.000) & 0.025 (0.000) & 0.172 (0.002) & \textbf{0.017 (0.000)} \\
        Blocks & 0.2 & 2.604 (0.053) & 3.334 (0.028) & 3.574 (0.027) & 3.707 (0.018) & 3.311 (0.048) & 3.769 (0.116) & 3.354 (0.024) & \textbf{2.497 (0.111)} \\
        Blocks & 1.0 & 0.619 (0.006) & 1.147 (0.010) & 1.253 (0.010) & 1.147 (0.014) & 0.673 (0.005) & 0.701 (0.006) & 1.331 (0.007) & \textbf{0.420 (0.004)} \\
        Blocks & 3.0 & 0.209 (0.002) & 0.409 (0.002) & 0.516 (0.003) & 0.351 (0.003) & 0.159 (0.001) & 0.172 (0.002) & 0.520 (0.003) & \textbf{0.122 (0.001)} \\
        Doppler & 0.2 & 0.068 (0.001) & 0.078 (0.000) & 0.084 (0.000) & 0.085 (0.000) & 0.081 (0.001) & 0.089 (0.002) & 0.079 (0.000) & \textbf{0.054 (0.002)} \\
        Doppler & 1.0 & 0.010 (0.000) & 0.020 (0.000) & 0.027 (0.000) & 0.024 (0.000) & 0.012 (0.000) & 0.013 (0.000) & 0.032 (0.000) & \textbf{0.007 (0.000)} \\
        Doppler & 3.0 & 0.002 (0.000) & 0.004 (0.000) & 0.006 (0.000) & 0.005 (0.000) & 0.002 (0.000) & 0.002 (0.000) & 0.013 (0.000) & \textbf{0.001 (0.000)} \\
        HeaviSine & 0.2 & \textbf{2.619 (0.090)} & 4.915 (0.141) & 6.072 (0.142) & 8.447 (0.122) & 4.840 (0.141) & 5.788 (0.246) & 6.422 (0.151) & 2.993 (0.233) \\
        HeaviSine & 1.0 & 0.253 (0.007) & 0.603 (0.015) & 0.961 (0.016) & 1.120 (0.028) & 0.519 (0.010) & 0.553 (0.012) & 0.707 (0.015) & \textbf{0.204 (0.008)} \\
        HeaviSine & 3.0 & 0.080 (0.001) & 0.118 (0.001) & 0.211 (0.003) & 0.230 (0.016) & 0.109 (0.001) & 0.114 (0.002) & 0.281 (0.004) & \textbf{0.052 (0.001)} \\
                \addlinespace
        \multicolumn{10}{c}{\(n=2048\)} \\
                \midrule
        Bumps & 0.2 & 0.406 (0.002) & 0.434 (0.001) & 0.446 (0.001) & 0.448 (0.001) & 0.441 (0.002) & 0.466 (0.006) & 0.434 (0.001) & \textbf{0.346 (0.007)} \\
        Bumps & 1.0 & 0.092 (0.001) & 0.207 (0.002) & 0.225 (0.001) & 0.174 (0.001) & 0.092 (0.001) & 0.096 (0.001) & 0.288 (0.001) & \textbf{0.058 (0.001)} \\
        Bumps & 3.0 & 0.017 (0.000) & 0.050 (0.000) & 0.065 (0.000) & 0.037 (0.000) & 0.016 (0.000) & 0.017 (0.000) & 0.107 (0.000) & \textbf{0.010 (0.000)} \\
        Blocks & 0.2 & 1.904 (0.034) & 2.967 (0.030) & 3.246 (0.032) & 3.487 (0.027) & 2.720 (0.045) & 2.898 (0.057) & 3.090 (0.027) & \textbf{1.739 (0.050)} \\
        Blocks & 1.0 & 0.472 (0.003) & 0.789 (0.007) & 0.960 (0.006) & 0.850 (0.007) & 0.520 (0.004) & 0.535 (0.004) & 1.020 (0.009) & \textbf{0.298 (0.003)} \\
        Blocks & 3.0 & 0.129 (0.001) & 0.305 (0.001) & 0.364 (0.002) & 0.246 (0.002) & 0.114 (0.001) & 0.120 (0.001) & 0.413 (0.002) & \textbf{0.077 (0.001)} \\
        Doppler & 0.2 & 0.052 (0.001) & 0.072 (0.000) & 0.081 (0.000) & 0.083 (0.000) & 0.071 (0.001) & 0.074 (0.001) & 0.077 (0.000) & \textbf{0.033 (0.001)} \\
        Doppler & 1.0 & 0.006 (0.000) & 0.013 (0.000) & 0.018 (0.000) & 0.016 (0.000) & 0.008 (0.000) & 0.008 (0.000) & 0.023 (0.000) & \textbf{0.004 (0.000)} \\
        Doppler & 3.0 & 0.001 (0.000) & 0.003 (0.000) & 0.004 (0.000) & 0.003 (0.000) & 0.001 (0.000) & 0.001 (0.000) & 0.011 (0.000) & \textbf{0.001 (0.000)} \\
        HeaviSine & 0.2 & 1.654 (0.053) & 2.891 (0.092) & 4.499 (0.093) & 6.986 (0.188) & 3.313 (0.085) & 3.780 (0.111) & 4.812 (0.177) & \textbf{1.516 (0.104)} \\
        HeaviSine & 1.0 & 0.158 (0.003) & 0.273 (0.005) & 0.609 (0.008) & 0.717 (0.019) & 0.342 (0.005) & 0.366 (0.007) & 0.469 (0.009) & \textbf{0.138 (0.005)} \\
        HeaviSine & 3.0 & 0.051 (0.001) & 0.094 (0.001) & 0.145 (0.001) & 0.136 (0.002) & 0.075 (0.001) & 0.078 (0.001) & 0.287 (0.003) & \textbf{0.034 (0.000)} \\
        \bottomrule
    \end{tabular}
    }
\end{table}
\end{landscape}

Figure~\ref{fig:simulation-nmse-snr} displays normalized-error profiles over the expanded SNR range at \(n=1024\).  The additional SNR values test whether the low-SNR ranking persists away from \(\rho=0.2\), while the normalization prevents the large-scale Blocks and HeaviSine signals from dominating the visual comparison.  The figure is descriptive; the detailed cellwise MSE values remain the primary accuracy comparison.

\begin{figure}[htbp]
    \centering
    \includegraphics[width=\linewidth]{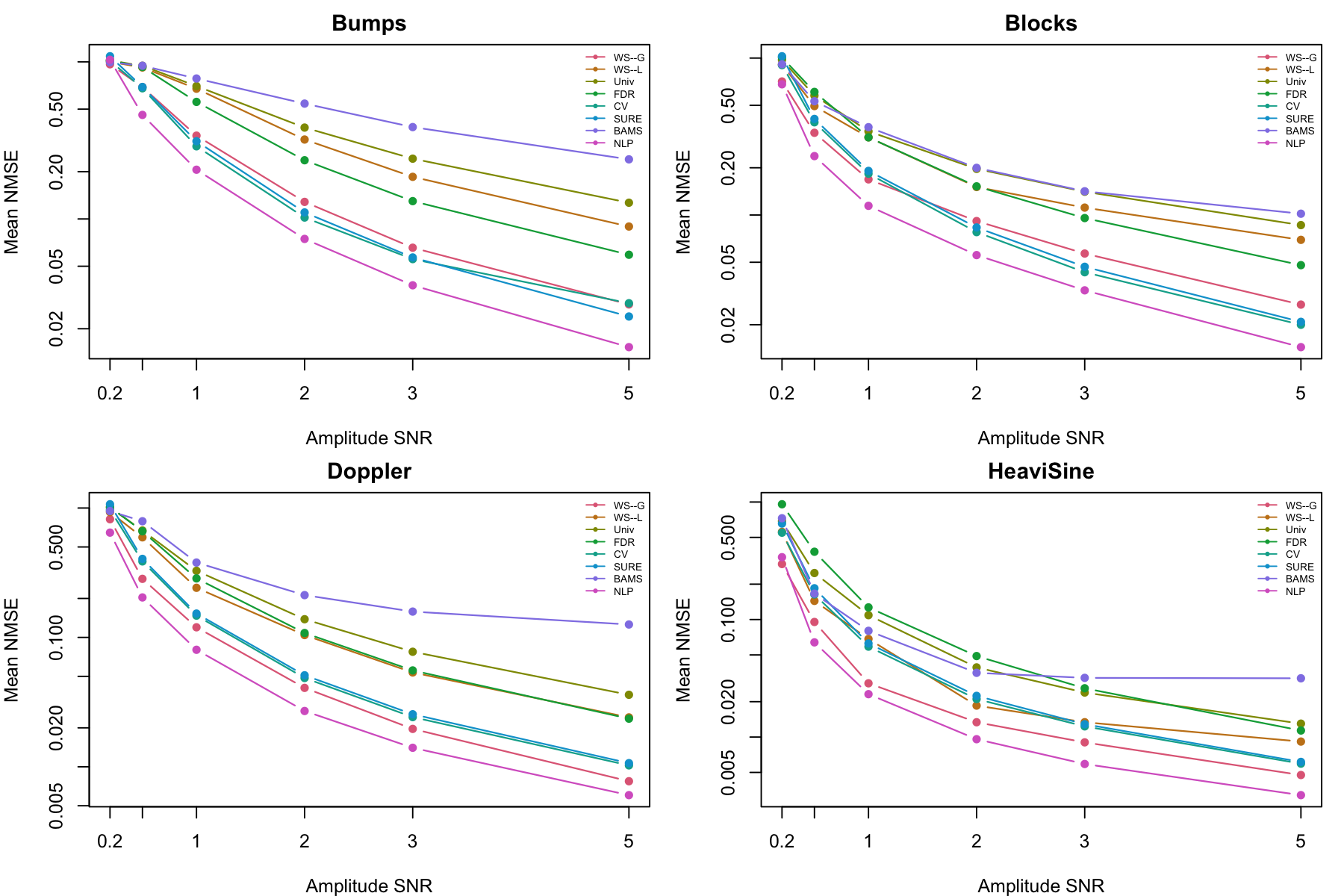}
    \caption{Mean NMSE as a function of amplitude SNR at \(n=1024\) over the expanded primary Gaussian-error grid.  Panels correspond to Bumps, Blocks, Doppler, and HeaviSine; each point averages 100 common-random-number replications, and lower values indicate better reconstruction accuracy.}
    \label{fig:simulation-nmse-snr}
\end{figure}

{\color{red}
Table~\ref{tab:simulation-overall} summarizes the primary Gaussian-error simulations based on the four deterministic Donoho--Johnstone test functions, averaged over the 36 signal--sample-size--SNR cells with 100 common-random-number replications per cell; MSE and NMSE are computed against the known generating functions.}  
Because the test functions have different scales, these pooled averages should not be interpreted as a universal method ranking; normalized and cellwise summaries are more informative.  NLP attains the smallest pooled mean MSE and NMSE, whereas WS--G is the best non-NLP method in 24 of the 36 cells.  WS--G has mean MSE 0.631 and mean NMSE 0.306, compared with mean MSE 0.624 and mean NMSE  0.287 for NLP.
This accuracy advantage for NLP comes at a substantial computational cost: its mean fitting time is \(116.442~\mathrm{s}\) per replicate, compared with \(0.281~\mathrm{s}\) for WS--G.  Thus, the table indicates an accuracy--computational-cost trade-off rather than a universal pooled-MSE ordering.

\begin{table}[p]
    \centering
    \caption{Descriptive averages over the original 36 Gaussian-error design cells.  NMSE is normalized by the variance of the corresponding true signal; runtime is the mean elapsed time in seconds per fitted replicate; best values for each numerical measure are shown in bold.}
    \label{tab:simulation-overall}
    \scriptsize
    \renewcommand{\arraystretch}{0.92}
    \resizebox{\linewidth}{!}{%
    \begin{tabular}{lrrrrrr}
        \toprule
        Method & Mean MSE & Mean NMSE & Mean MAE & Mean output SNR & Runtime (s) & Best non-NLP cells \\
        \midrule
        WS--G & 0.631 & 0.306 & 0.435 & 3.807 & 0.281 & 24 \\
        WS--L & 0.953 & 0.423 & 0.557 & 2.714 & 1.816 & 0 \\
        Univ & 1.148 & 0.476 & 0.631 & 2.268 & 0.003 & 0 \\
        FDR & 1.308 & 0.468 & 0.649 & 2.433 & 0.003 & 0 \\
        CV & 0.879 & 0.358 & 0.511 & 3.381 & 0.011 & 10 \\
        SURE & 1.054 & 0.406 & 0.529 & 3.302 & \textbf{0.003} & 1 \\
        BAMS & 1.132 & 0.495 & 0.626 & 2.077 & 0.012 & 1 \\
        NLP & \textbf{0.624} & \textbf{0.287} & \textbf{0.408} & \textbf{4.593} & 116.442 & -- \\
        \bottomrule
    \end{tabular}
    }
\end{table}

\subsubsection{What is gained by the adaptive mixture?}
\label{sec:simulation-ablation}
The central ablation compares four Gaussian WS specifications: Wendland-only (\(\omega_j=1\)), Semicircle-only (\(\omega_j=0\)), a constant mixture with only the intercept estimated, and the proposed adaptive mixture with both the intercept and resolution trend estimated.  The comparison uses the component/adaptivity study over all three sample sizes, all three original SNR values, and all four signals.  Table~\ref{tab:simulation-ablation} reports normalized and raw error summaries together with cellwise wins and paired differences.  The adaptive fit was the lowest-NMSE specification in 8 of the 36 ablation cells, whereas the Semicircle-only endpoint was lowest in 26 cells.  Its paired win rate against the constant mixture was 72.6\%.  Thus, the ablation does not indicate a uniform accuracy gain from resolution-dependent mixing; rather, it shows that the adaptive mixture is a flexible specification whose benefit depends on the signal and noise level.

\begin{table}[p]
    \centering
    \caption{Gaussian WS component and adaptivity ablation over the 36 original signal--sample-size--SNR cells.  The paired difference is adaptive WS--G minus the indicated specification; negative values favor the adaptive mixture.}
    \label{tab:simulation-ablation}
    \scriptsize
    \renewcommand{\arraystretch}{0.92}
    \resizebox{\linewidth}{!}{%
    \begin{tabular}{lrrrcc}
        \toprule
        Specification & Mean NMSE & Mean MSE & Lowest-NMSE cells & Paired NMSE difference [95\% CI] & Adaptive win rate (\%) \\
        \midrule
        WS--G & 0.306 & 0.631 & 8 & -- & -- \\
        WS--G constant & 0.306 & 0.632 & 1 & -0.000224 [-0.000344, -0.000104] & 72.6 \\
        Wendland-only & 0.349 & 0.766 & 1 & -0.042523 [-0.044524, -0.040521] & 97.1 \\
        Semicircle-only & \textbf{0.306} & \textbf{0.631} & 26 & 0.000284 [0.000205, 0.000363] & 25.7 \\
        \bottomrule
    \end{tabular}
    }
\end{table}

Figure~\ref{fig:simulation-omega-profiles} shows the fitted adaptive Wendland weight profiles at \(n=1024\) and \(\rho=0.2\).  The solid curve is the replicate median and the shaded region is the interquartile range.  Although the weights are positive, they are extremely close to zero in this slice: all median weights are below \(2.4\times10^{-5}\), and all upper-quartile values are below \(2.3\times10^{-4}\).  Thus, the adaptive fit is close to the Semicircle-only endpoint \(\omega_j=0\); the figure should be interpreted as showing weak data-driven variation around that endpoint, not substantial resolution-varying use of both components.

\begin{figure}[htbp]
    \centering
    \includegraphics[width=\linewidth]{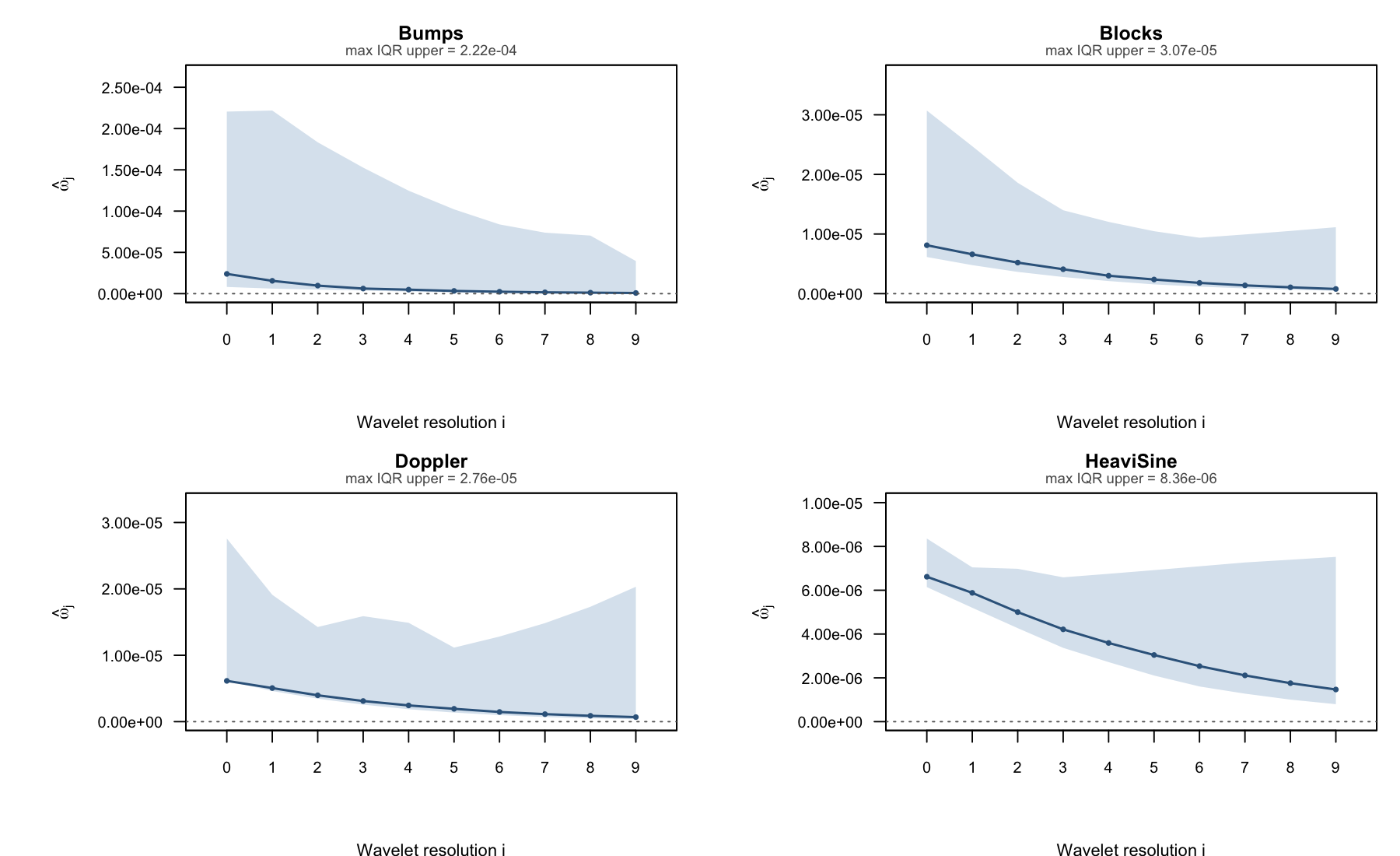}
    \caption{Fitted Wendland mixture-weight profiles \(\widehat\omega_j\) for the adaptive WS--Gaussian fit at \(n=1024\) and amplitude SNR \(\rho=0.2\).  Panels correspond to the four test signals.  Curves show replicate medians and ribbons show interquartile ranges.  Signal-specific zoomed vertical axes display the small positive weights; the dotted horizontal line marks the Semicircle-only endpoint \(\omega_j=0\), while the Wendland-only endpoint \(\omega_j=1\) lies outside the plotted ranges.}
    \label{fig:simulation-omega-profiles}
\end{figure}

\subsubsection{Paired and normalized comparisons}
At \(n=1024\) and \(\rho=0.2\), Table~\ref{tab:simulation-paired} reports paired MSE and NMSE differences pooled over the four signals and 100 common replications per signal.  The interval for each difference uses the empirical standard deviation of the paired replicate differences, rather than separate intervals for the two method means.  The 95\% intervals for the WS--G versus NLP comparison include zero for both MSE and NMSE, and the WS--G win rate is 38.5\%; therefore, this low-SNR slice does not establish an advantage for WS--G over NLP.  The negative differences and win rates for the other comparisons indicate where WS--G has a reproducible low-SNR advantage and where competing methods remain competitive.  Figure~\ref{fig:simulation-paired-differences} gives the corresponding MSE-difference intervals on a common scale.

\begin{table}[p]
    \centering
    \caption{Paired comparison with WS--Gaussian at (n=1024) and amplitude SNR (0.2), pooled over the four signals and 100 common-random-number replications per signal.  Negative differences favor WS--Gaussian.}
    \label{tab:simulation-paired}
    \scriptsize
    \renewcommand{\arraystretch}{0.92}
    \resizebox{\linewidth}{!}{%
    \begin{tabular}{lccc}
        \toprule
        Competitor & MSE difference [95\% CI] & NMSE difference [95\% CI] & WS--G win rate (\%) \\
        \midrule
        WS--L & -0.763 [-0.869, -0.657] & -0.151 [-0.163, -0.139] & 94.2 \\
        Univ & -1.116 [-1.267, -0.965] & -0.225 [-0.241, -0.209] & 98.0 \\
        FDR & -1.743 [-1.990, -1.496] & -0.305 [-0.329, -0.280] & 98.8 \\
        CV & -0.741 [-0.841, -0.640] & -0.162 [-0.172, -0.151] & 97.5 \\
        SURE & -1.102 [-1.277, -0.927] & -0.262 [-0.287, -0.237] & 98.2 \\
        BAMS & -1.144 [-1.320, -0.969] & -0.198 [-0.217, -0.180] & 93.2 \\
        NLP & -0.071 [-0.205, 0.064] & 0.024 [-0.006, 0.053] & 38.5 \\
        \bottomrule
    \end{tabular}
    }
\end{table}

\begin{figure}[htbp]
    \centering
    \includegraphics[width=0.92\linewidth]{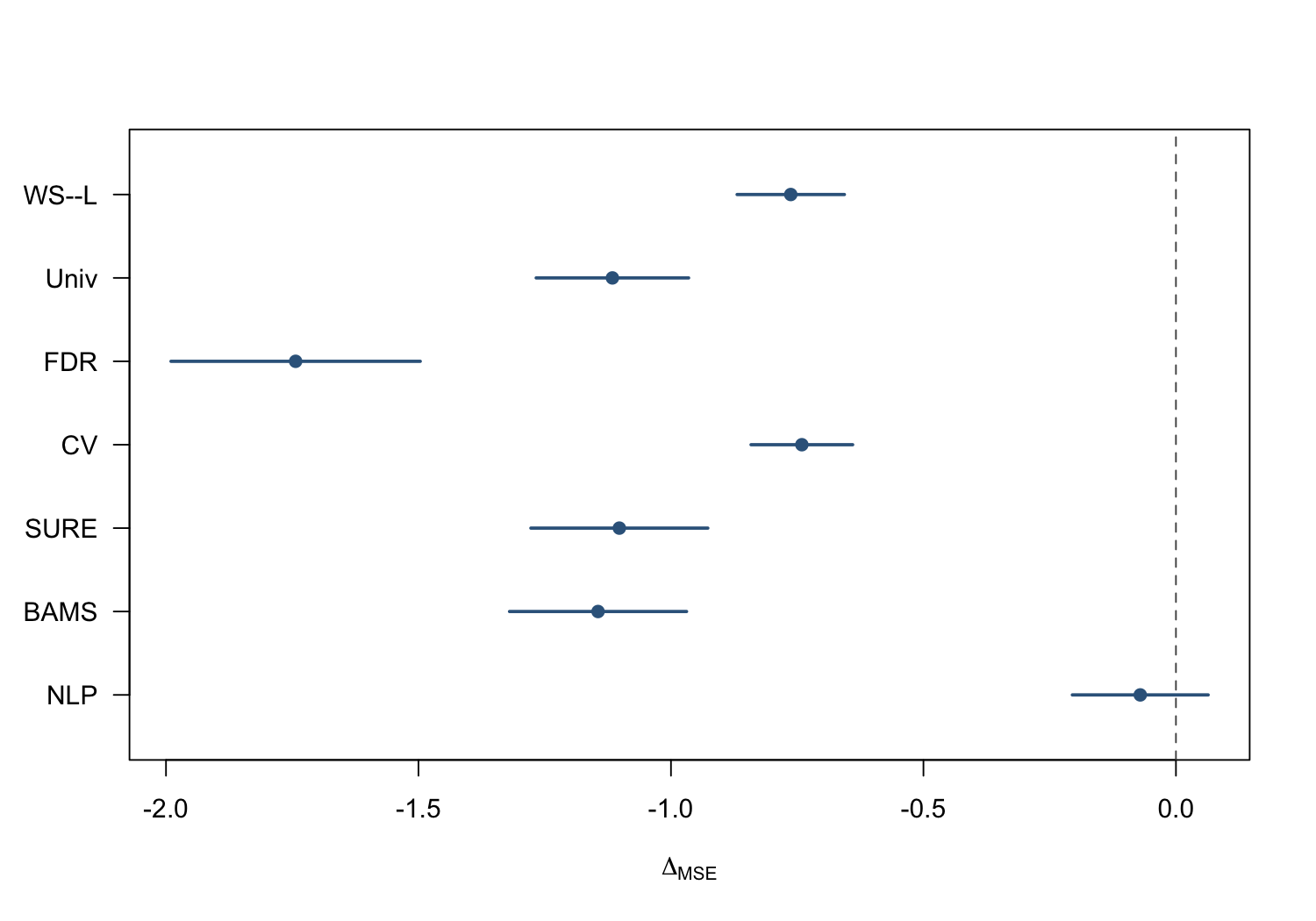}
    \caption{Paired mean MSE differences at \(n=1024\) and amplitude SNR \(\rho=0.2\), pooled over the four test signals.  Points are mean differences \(\operatorname{MSE}_{\mathrm{WS-G}}-\operatorname{MSE}_{m}\), and horizontal intervals are replicate-level 95\% \(t\)-intervals.  Negative values favor WS--Gaussian.}
    \label{fig:simulation-paired-differences}
\end{figure}

\begin{figure}[htbp]
    \centering
    \includegraphics[width=\linewidth]{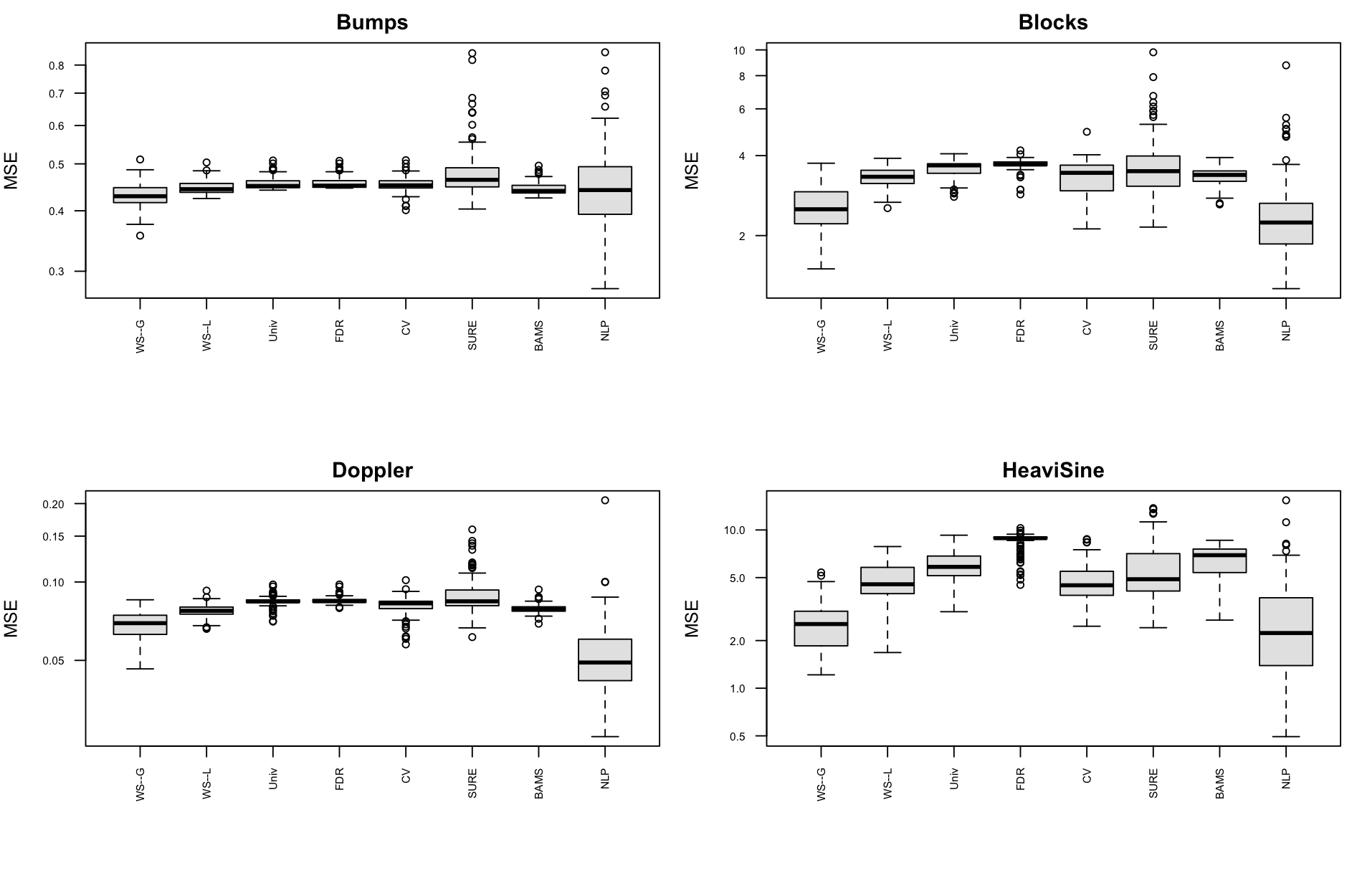}
    \caption{Replicate-specific MSE distributions at \(n=1024\) and amplitude SNR \(\rho=0.2\), shown separately for the four test signals.  Each box summarizes the 100 common-random-number replications for each method.}
    \label{fig:simulation-mse-boxplots}
\end{figure}

Figure~\ref{fig:simulation-mse-boxplots} shows the replicate-level MSE distributions and the spread and outliers underlying the paired summaries in Table~\ref{tab:simulation-paired}.

\subsection{Interpretation}
\label{sec:simulation-interpretation}
The primary simulation supports a qualified conclusion.  At the representative low-SNR slice, the pooled MSE leader is WS-G, while the full cellwise table shows which signals and resolutions drive that summary.  WS--G is the best non-NLP method in 24 of the original cells, and its performance should therefore be interpreted as a cellwise advantage in the difficult settings where it occurs rather than as uniform dominance.  The ablation and fitted-weight results show that the model provides two alternative compact-support slab shapes and permits resolution-dependent weighting, but the fitted weights in the examined low-SNR slice were close to the Semicircle-only endpoint; thus, the present results do not demonstrate a substantial accuracy gain from resolution-dependent mixing.

The supplementary analyses quantify how sensitive this interpretation is to numerical and tuning choices.  The quadrature study selects the smallest order whose reconstruction, component-moment, and MSE changes remain below prespecified tolerances.  The hyperparameter and support analyses assess the fixed sparsity schedule, empirical support quantile, and possible saturation of large coefficients.  The Laplace-error experiment evaluates a different data-generating mechanism and therefore is a sensitivity analysis rather than an estimate of the Gaussian-error risk.  Finally, the semi-synthetic seismic experiment uses a processed trace as a known surrogate truth after adding controlled noise; it complements, but does not replace, the real-data illustration for which no noise-free reference waveform is available.

\section{Seismic trace data application}
\label{sec:real-data-seismic}

The simulation study evaluates denoising accuracy against a known signal. For a measured seismic record, however, a noise-free reference acceleration is not available.  We therefore use the real-data analysis as an illustration of the reconstructed waveform and compare the methods using signal-level diagnostics---residual variability, sample-to-sample residual roughness, retained squared-energy, preservation and timing of the largest absolute acceleration, and visual agreement with the observed trace.  These quantities describe how aggressively each method modifies the record. They should not be interpreted as out-of-sample accuracy measures in the absence of a trusted reference signal.

\subsection{Record and data extraction}
The record was obtained through the United States Geological Survey (USGS) Earthquake Data portal, which directs processed strong-motion waveform downloads to the Center for Engineering Strong Motion Data (CESMD). 
\footnote{USGS Earthquake Data: \url{https://www.usgs.gov/programs/earthquake-hazards/science/earthquake-data}; CESMD: \url{https://www.strongmotioncenter.org/}.} Specifically, we downloaded the 29 July 2008 Chino Hills earthquake record from station 24763, which corresponds to Sylmar--Olive View Hospital Grounds.  The station coordinates are \(34.327^\circ\mathrm{N}\) and \(118.444^\circ\mathrm{W}\), and the record identifier is 24763-K0186-08211.01.

The time-domain data were extracted from \texttt{CE24763.V2}, which contains the corrected accelerogram blocks.  The companion file \texttt{CE24763.V3} contains response and Fourier-amplitude spectra and was not used as the time-domain input.  In the main text, we analyze channel 1, the horizontal component labeled $360^\circ$. Results for channels 2 and 3, labeled \textit{Up} and $90^\circ$, respectively, are provided in the Supplementary Materials.

The source file reports instrument- and baseline-corrected acceleration and a bandpass with 3-dB points at \(0.30\) and \(40.00\) Hz.  We used the acceleration block supplied in that file without applying an additional filter, detrending operation, or unit conversion.  We used a fixed-width parser that read eight values per line using 10-character fields and retained the first \(N=15{,}400\) acceleration values. Let \(a_i\) denote the \(i\)th observed acceleration, measured in \(\mathrm{cm}\,\mathrm{s}^{-2}\), and let \(\Delta t=0.005\) s denote the sampling interval.  The corresponding sampling frequency is \(f_s=1/\Delta t=200\) Hz, and the observation times are \(t_i=(i-1)\Delta t\), \(i=1,\ldots,N\).  Thus, the record spans \(0\) to \(76.995\) s, with a nominal duration of \(77.000\) s.  The observed acceleration ranges from \(-44.27435\) to \(36.89636\) \(\mathrm{cm}\,\mathrm{s}^{-2}\); the largest observed absolute acceleration is \(44.27435\ \mathrm{cm}\,\mathrm{s}^{-2}\) at approximately \(36.05\) s. The data and common analysis settings are summarized in Table~\ref{tab:seismic-settings}.

\begin{table}[htbp]
    \centering
    \caption{Data and common analysis settings for the seismic application.}
    \label{tab:seismic-settings}
    \small
    \begin{tabular}{@{}p{0.31\linewidth}p{0.62\linewidth}@{}}
        \toprule
        Quantity & Specification \\
        \midrule
        Earthquake and date & Chino Hills earthquake, 29 July 2008 \\
        Station and record & Station 24763; record 24763-K0186-08211.01 \\
        Input file and component & \texttt{CE24763.V2}, channel 1, \(360^\circ\) horizontal component \\
        Original observations & \(N=15{,}400\) acceleration values \\
        Sampling interval and frequency & \(\Delta t=0.005\) s and \(f_s=200\) Hz \\
        Units and duration & \(\mathrm{cm}\,\mathrm{s}^{-2}\); 77.000 s nominal duration \\
        Source processing & Instrument- and baseline-corrected; 0.30--40.00 Hz bandpass reported by CESMD \\
        Wavelet input length & \(2^{14}=16{,}384\), obtained by appending 984 reverse-padded values \\
        Wavelet basis & Daubechies extremal-phase family, \texttt{DaubExPhase}, filter number 10 \\
        Boundary rule and coarsest level & Periodic boundary rule; \(J_0=0\) \\
        WS tuning settings & \(\beta\)-quantile \(=0.99\), spike offset \(=2\), spike exponent \(\gamma=2.4\), quadrature order \(=48\) \\
        Fitted methods & WS--Gaussian, WS--Laplace, Wendland-only, Semicircle-only, Univ, FDR, CV, and SURE \\
        \bottomrule
    \end{tabular}
\end{table}

\subsection{Padding and wavelet analysis}
Because the original record has length \(N=15{,}400\), which is not a power of two, the wavelet input was extended to \(\widetilde N=16{,}384=2^{14}\) observations. Thus, \(J=\log_2(\widetilde N)=14\), and the detail levels in the padded analysis are \(j=J_0,\ldots,13\). Specifically, 984 values were appended in reverse order from the end of the observed record. This reverse-padding rule avoids discarding observations and reduces the artificial endpoint jump that would result from simply appending zeros. The periodic boundary rule was then applied to the padded signal. After reconstruction, only the first \(N\) fitted values were retained, so all reported diagnostics refer to the original 77-second record.

All methods used the Daubechies extremal phase wavelet family with filter number 10, periodic boundary handling, and \(J_0=0\).  The proposed Wendland--semicircle model was fitted with a Gaussian likelihood (WS--Gaussian) and a Laplace likelihood (WS--Laplace).  For the two component-specific fits, Wendland-only fixes the resolution-specific Wendland mixture weight at \(\omega_j=1\), whereas Semicircle-only fixes it at \(\omega_j=0\), where \(\omega_j\) denotes the mixture weight assigned to the Wendland component at wavelet resolution \(j\). The support-scale quantile, spike prior offset, spike exponent, and quadrature settings were \(0.99\), \(2\), \(2.4\), and \(48\), respectively, as summarized in Table~\ref{tab:seismic-settings}. For comparison, Univ, FDR, CV, and SURE were used. These four classical rules used the same wavelet basis and boundary rule as the proposed method. 

\subsection{Real-data diagnostics}
Let \(\widehat a_{m,i}\) denote the reconstructed acceleration at time \(t_i\) under method \(m\), and define the residual \(e_{m,i}=a_i-\widehat a_{m,i}\).  The residual standard deviation is \(\operatorname{SD}(e_m)\), and the sample-to-sample residual roughness is \(\operatorname{SD}(\Delta e_m)\), where \(\Delta e_{m,i}=e_{m,i+1}-e_{m,i}\) for \(i=1,\ldots,N-1\).  The retained squared-energy ratio is
\[
    R_m =
    \frac{\sum_{i=1}^{N}\widehat a_{m,i}^{\,2}}
         {\sum_{i=1}^{N}a_i^{\,2}}.
\]
Finally, let \(A_m=\max_i|\widehat a_{m,i}|\) denote the reconstructed peak absolute acceleration and let \(T_m=t_{\widehat i_m}\), where \(\widehat i_m=\arg\max_i|\widehat a_{m,i}|\), denote its time.

Table~\ref{tab:seismic-diagnostics} reports the resulting diagnostics.  The WS--Gaussian reconstruction has residual standard deviation \(1.2357\) in the input units.  Its retained squared-energy ratio is \(0.7496\), and its reconstructed peak is \(37.4867\ \mathrm{cm}\,\mathrm{s}^{-2}\) at \(36.04\) s.  WS--Laplace is nearly identical, while the Wendland-only fit has slightly larger residual variability and residual roughness.  The Semicircle-only reconstruction is numerically almost indistinguishable from WS--Gaussian for this record.

\begin{table}[htbp]
    \centering
    \caption{Signal-level diagnostics for the Chino Hills channel-1
    accelerogram.  Residual SD and residual roughness are in
    \(\mathrm{cm}\,\mathrm{s}^{-2}\); retained energy is dimensionless; peak
    acceleration is in \(\mathrm{cm}\,\mathrm{s}^{-2}\); and peak time is in
    seconds.  The estimated noise scale is repeated because it was common to
    all fitted methods.}
    \label{tab:seismic-diagnostics}
    \scriptsize
    \resizebox{\linewidth}{!}{%
    \begin{tabular}{lrrrrrr}
        \toprule
        Method & \(\widehat\sigma\) & Residual SD &
        Residual roughness & Retained energy & Peak \(|a|\) & Peak time \\
        \midrule
        WS--Gaussian      & 0.001038 & 1.235737 & 0.126581 & 0.749592 & 37.486663 & 36.040 \\
        WS--Laplace       & 0.001038 & 1.235750 & 0.126616 & 0.749502 & 37.489605 & 36.040 \\
        Wendland-only     & 0.001038 & 1.259843 & 0.129444 & 0.750354 & 37.419648 & 36.040 \\
        Semicircle-only   & 0.001038 & 1.235738 & 0.126580 & 0.749643 & 37.486663 & 36.040 \\
        Univ               & 0.001038 & 0.011819 & 0.008193 & 0.998916 & 44.250075 & 36.050 \\
        FDR                & 0.001038 & 0.001816 & 0.001837 & 0.999871 & 44.269623 & 36.050 \\
        CV                 & 0.001038 & 0.012125 & 0.008371 & 0.998883 & 44.249556 & 36.050 \\
        SURE               & 0.001038 & 0.001704 & 0.001747 & 0.999880 & 44.269943 & 36.050 \\
        \bottomrule
    \end{tabular}}
\end{table}

The thresholding reconstructions have residual standard deviations ranging from \(0.0017\) to \(0.0121\) in the input units.  Their sample-to-sample roughness ranges from \(0.0017\) to \(0.0084\), and their retained-energy ratios exceed \(0.9988\).  Thus, they remain extremely close to the observed record under these diagnostics and preserve the observed peak of approximately \(44.27\ \mathrm{cm}\,\mathrm{s}^{-2}\) at \(36.05\) s. In contrast, the WS fits remove more high-frequency variation and attenuate the absolute peak to approximately \(37.42\)--\(37.49\) \(\mathrm{cm}\,\mathrm{s}^{-2}\), while preserving its timing to within 0.01 s.  Since no clean signal is available, the larger residual and lower retained energy for WS should be interpreted as stronger shrinkage, not automatically as either improvement or deterioration. Supplementary Table~13 and Supplementary Table~14 include similar results for channels 2 and 3, respectively.

\subsection{Waveform comparison}
The full-record comparison is shown in Figure~\ref{fig:seismic-full-record}. The observed acceleration is drawn in gray in every panel, while the colored curve is the reconstruction from the method named in that panel.  All methods retain the principal wave packet, its approximate arrival time, and the subsequent decaying oscillations.  The WS reconstructions visibly suppress some of the rapid oscillations and reduce the height of the largest extrema, whereas Univ, FDR, CV, and SURE remain nearly superimposed on the observed trace over most of the record. Supplementary Figure~9 and Supplementary Figure~10 present similar plots for channels 2 and 3, respectively.

\begin{figure}[htbp]
    \centering
    \includegraphics[width=\linewidth]{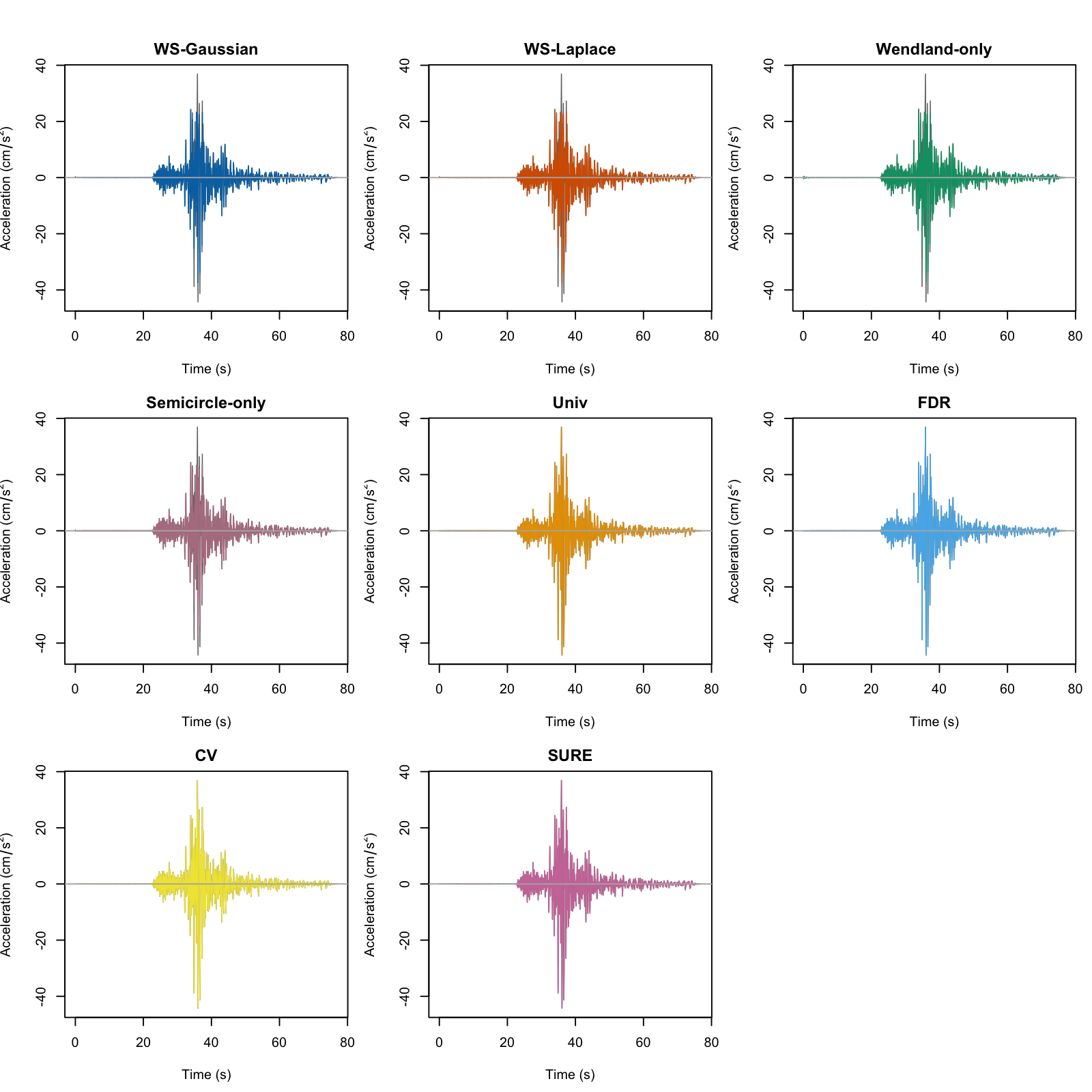}
    \caption{Full-record comparison for the 77-second Chino Hills channel-1
    accelerogram.  In each panel, the gray curve is the observed acceleration
    and the colored curve is the reconstruction from the indicated method.}
    \label{fig:seismic-full-record}
\end{figure}

Figure~\ref{fig:seismic-peak-zoom} provides a closer view of the interval
\(34\)--\(38\) s containing the largest observed absolute acceleration.  The
WS--Gaussian, WS--Laplace, Wendland-only, and Semicircle-only curves preserve the broad waveform while reducing the most extreme excursions. The thresholding curves preserve the peak more nearly exactly, consistent with the retained-energy and peak diagnostics in Table~\ref{tab:seismic-diagnostics}. The near overlap of WS--Gaussian and Semicircle-only in this window is also consistent with their nearly identical signal-level diagnostics. 

\begin{figure}[htbp]
    \centering
    \includegraphics[width=\linewidth]{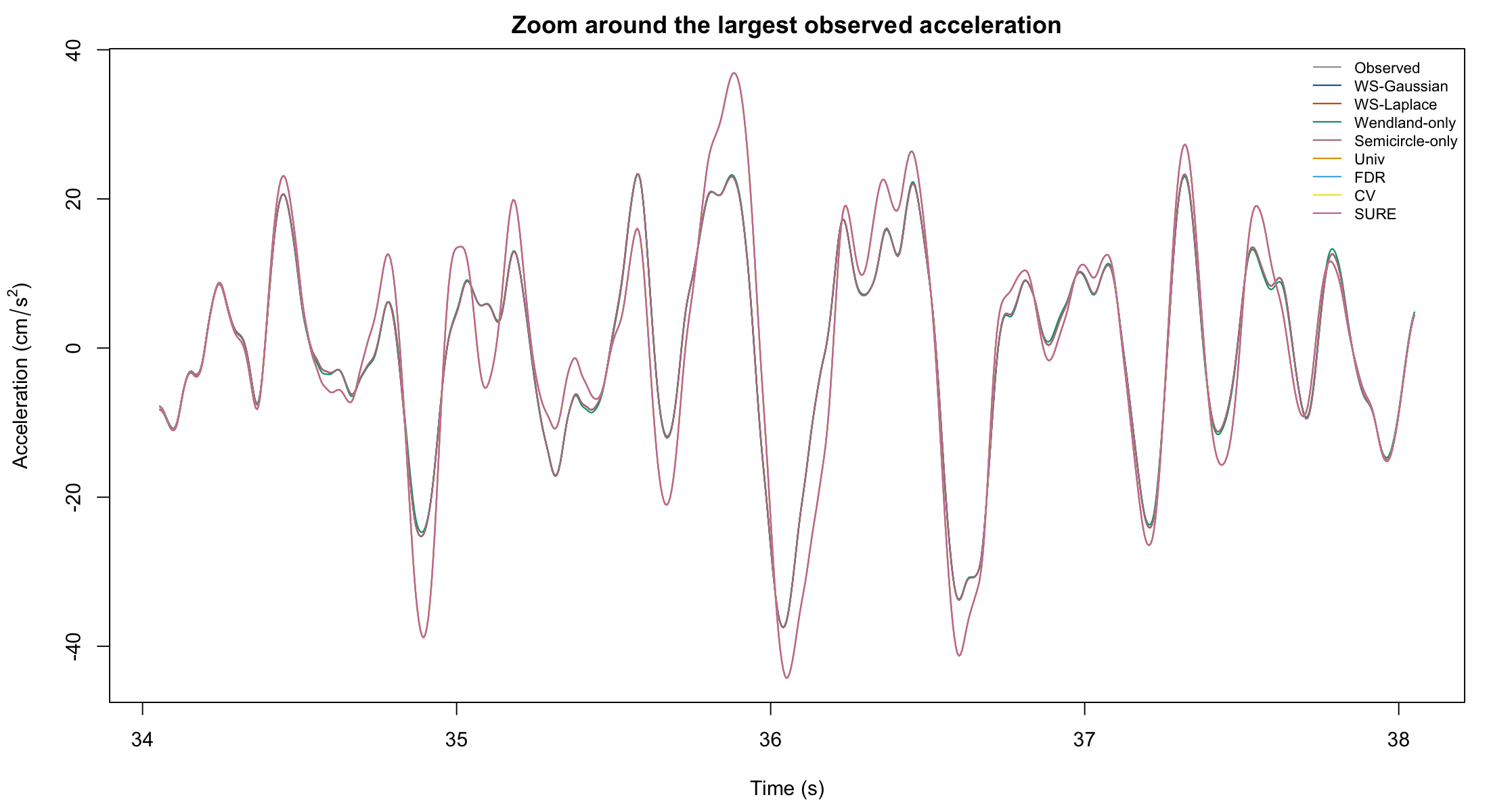}
    \caption{Zoomed comparison over \(34\)--\(38\) s around the largest
    observed absolute acceleration.  The observed trace and eight
    reconstructions are shown in \(\mathrm{cm}\,\mathrm{s}^{-2}\).}
    \label{fig:seismic-peak-zoom}
\end{figure}

\subsection{Interpretation and limitations}
This application demonstrates the operational behavior of the proposed shrinkage model on a processed strong-motion record. WS--Gaussian and WS--Laplace produce almost identical reconstructions for this example, and their output is much closer to the Semicircle-only fit than to the Wendland-only fit.  The proposed fits also provide visibly stronger attenuation of rapid fluctuations than the four classical thresholding rules, while preserving the timing of the dominant acceleration pulse.  These observations are useful for illustrating the model, but they do not establish that one method has lower denoising risk for real seismic data. 

Several limitations should be made explicit.  First, the record is a single horizontal component from one earthquake and one station; the results cannot be generalized to other magnitudes, distances, site conditions, components, or recording instruments.  Second, the CESMD waveform has already undergone instrument and baseline correction and source-reported bandpass filtering, so the analysis assesses shrinkage conditional on that preprocessing. Third, the periodic boundary rule and reverse-padding construction can affect the two endpoints, although the dominant acceleration pulse is well inside the record.  Fourth, the absence of a clean reference prevents direct calculation of MSE, MAE, or true signal-to-noise improvement. 

\section{Conclusion}
\label{sec:conclusion}
We proposed a Bayesian wavelet-shrinkage rule based on a spike at zero and an adaptive mixture of a compactly supported Wendland density and the semicircle density.  The Wendland component concentrates more strongly near zero and vanishes smoothly at the support boundary, whereas the semicircle component is more dispersed and can preserve moderate wavelet coefficients.  The resulting posterior-mean estimator is computable through two component-specific integrals, with explicit support control through \(\beta_j\) and adaptive shape selection through the resolution-specific mixture weight \(\omega_j\).

The simulation results indicate that the proposed WS--Gaussian method is competitive in low-SNR settings.  Its pooled mean MSE was \(0.631\), compared with \(0.624\) for NLP and \(0.879\) for CV, and it was the best-performing non-NLP method in 24 of the 36 design cells, including 11 of the 12 low-SNR cells.  At \(n=1024\) and \(\rho=0.2\), WS--Gaussian had lower MSE than each non-NLP competitor in at least \(93.2\%\) of paired replications.  NLP achieved the best overall reconstruction accuracy, particularly for Doppler and at moderate or favorable SNR, but WS--Gaussian provided a more favorable accuracy--runtime balance than both NLP and the Laplace-working-likelihood version.  Thus, the results do not support universal dominance of a single kernel.  Instead, they suggest that the usefulness of the Wendland--semicircle construction depends on SNR, signal structure, and computational constraints.

The seismic application provides a complementary illustration without a noise-free reference.  WS--Gaussian and WS--Laplace produced nearly identical reconstructions, and the WS--Gaussian output was numerically close to the Semicircle-only reconstruction.  Relative to the observed peak acceleration of approximately \(44.27\ \mathrm{cm}\,\mathrm{s}^{-2}\), the proposed fits reduced the peak to approximately \(37.49\ \mathrm{cm}\,\mathrm{s}^{-2}\) while preserving its timing to within \(0.01\) s.  The classical thresholding methods retained more than \(99.8\%\) of the observed squared energy and preserved the observed peak almost exactly.  These results illustrate the different shrinkage strengths of the methods, but, because the seismic record does not provide a trusted clean signal, they cannot establish which reconstruction has lower real-data risk.

The semi-synthetic seismic validation complements this real-data illustration by treating the processed channel-1 trace as a surrogate truth and adding independent Gaussian noise at amplitude SNRs \(0.5\), \(1\), \(2\), and \(5\).  The noisy-observation baseline had mean NMSEs of approximately \(3.995\), \(0.999\), \(0.249\), and \(0.040\), respectively, closely matching the expected \(1/\mathrm{SNR}^{2}\) pattern.  WS--Gaussian improved on this baseline at SNRs \(0.5\), \(1\), and \(2\), but not at SNR \(5\), where weak noise made denoising less beneficial.  The adaptive WS--Gaussian fit was nearly identical to the Semicircle-only fit, while WS--Laplace exhibited a substantial peak-timing error at SNR \(2\) (mean \(0.412\) s; 72\% of replicates exceeded \(0.1\) s).  These results support the operational use of the method under controlled surrogate-truth conditions, but they do not establish uniform superiority or real-data denoising risk.

The method is formulated for a one-dimensional signal observed on an equally spaced grid and analyzed with an orthogonal discrete wavelet transform.  The primary formulation assumes independent errors with a common variance, whereas the Laplace version is a coefficient-wise working likelihood used for sensitivity analysis; it does not assert that a time-domain Laplace error remains exactly Laplace and independent after the discrete wavelet transform.  For a non-dyadic record, a deterministic extension to a dyadic length is used, and the reconstructed values are truncated to the original observations before application-level diagnostics are calculated.  Conditional on the empirical-Bayes hyperparameters, detail coefficients are modeled independently with symmetric spike-and-slab priors at each resolution, while scaling coefficients are retained rather than shrunk.  Consequently, the coefficient-wise construction does not model dependence among neighboring or parent--child coefficients or across resolutions.

The level-specific support interval \([-\beta_j,\beta_j]\) is an explicit modeling and truncation choice.  When \(\beta_j\) is obtained from an empirical quantile, it is not a formal estimator or guarantee of the true support; coefficients outside the interval are necessarily subject to boundary saturation.  The method may therefore be less suitable when errors are strongly correlated or heteroscedastic, when the signal contains frequent large coefficients outside the fitted support, or when the selected wavelet basis and boundary rule inadequately represent the signal.  The results are consequently basis-, boundary-, hyperparameter-, and support-rule-dependent, and the present formulation does not directly cover two-dimensional images.

The present method estimates the shape trend through a low-dimensional empirical-Bayes parameterization.  A fully Bayesian extension could place a prior on \(\boldsymbol{\eta}\), propagate its uncertainty into the posterior mean, and provide posterior uncertainty intervals for \(\omega_j\) across resolutions.  The resolution trend permits, but does not guarantee, meaningful variation in \(\omega_j\), and numerical integration and bounded optimization introduce additional computational considerations.  Future work should investigate correlated observation errors, blockwise shrinkage, two-dimensional image wavelets, and applications to multiple seismic components and independent earthquake records.  Such studies should combine reconstruction measures with uncertainty quantification, fitted mixture weights, support-scale behavior, realistic noise models, and diagnostics based on known physical features of the observed signal.

\appendix

\section{Proof of Proposition \ref{prop:shrinkage-properties}}
\begin{proof}
For \(r\in\{0,1\}\), consider first the mixture-slab moment
\[
    M_{r,j}(d)=\int_{-\beta_j}^{\beta_j}\theta^r g_j(\theta;\omega_j,\beta_j)L(d\mid\theta)\,d\theta.
\]
Using the substitution \(\theta=-u\), the symmetry of \(g_j\), and the assumed symmetry of the likelihood gives
\[
    M_{0,j}(-d)
    =\int_{-\beta_j}^{\beta_j}g_j(u;\omega_j,\beta_j)L(d\mid u)\,du
    =M_{0,j}(d),
\]
whereas
\[
    M_{1,j}(-d)
    =-\int_{-\beta_j}^{\beta_j}u\,g_j(u;\omega_j,\beta_j)L(d\mid u)\,du
    =-M_{1,j}(d).
\]
The same argument applies separately to the Wendland and semicircle moments.  Moreover, \(L(-d\mid0)=L(d\mid0)\), so the denominator \(D_j(d)\) in \eqref{eq:posterior-denominator} is even.  Consequently, the posterior component probabilities are even functions of \(d\), the component means are odd functions of \(d\), and \eqref{eq:posterior-component-decomposition} implies \(\delta_j(-d)=-\delta_j(d)\).  Setting \(d=0\) yields \(\delta_j(0)=0\), proving (i) and (v).

The posterior distribution of \(\theta_{j,k}\) is supported on \([-\beta_j,\beta_j]\), because both the spike and the continuous slab are supported there.  Hence its posterior mean belongs to this interval, which proves the non-strict bound in (ii).  For finite \(d\), the posterior probability of the spike is strictly positive because \(\pi_j>0\) and \(L(d\mid0)>0\).  The spike is located at the interior point zero, so the posterior mean cannot equal either endpoint of the support.  Therefore \(|\delta_j(d)|<\beta_j\) for finite \(d\).

For the Gaussian likelihood, write \(t=d/\sigma^2\) and factor the likelihood as
\[
    L_N(d\mid\theta,\sigma^2)
    =\frac{1}{\sqrt{2\pi\sigma^2}}
      \exp\left\{-\frac{d^2}{2\sigma^2}\right\}
      \exp\left\{t\theta-\frac{\theta^2}{2\sigma^2}\right\}.
\]
After removing the common factor that does not depend on \(\theta\), the continuous part of the posterior is proportional to
\[
    q_j(\theta)e^{t\theta},
    \qquad
    q_j(\theta)=g_j(\theta;\omega_j,\beta_j)
    \exp\left\{-\frac{\theta^2}{2\sigma^2}\right\}.
\]
For every \(\varepsilon>0\), the assumption on the slab implies that the denominator has positive contribution from \((\beta_j-\varepsilon/2,\beta_j)\), whereas the contribution from \([-\beta_j,\beta_j-\varepsilon]\) is exponentially smaller as \(t\to\infty\).  More precisely,
\[
    \frac{\displaystyle\int_{-\beta_j}^{\beta_j-\varepsilon}q_j(\theta)e^{t\theta}\,d\theta}
    {\displaystyle\int_{\beta_j-\varepsilon/2}^{\beta_j}q_j(\theta)e^{t\theta}\,d\theta}
    \longrightarrow 0.
\]
The posterior mass of the spike at zero is also negligible relative to the contribution from the interval \((\beta_j-\varepsilon/2,\beta_j)\), because \(\beta_j>0\).  Thus, the full posterior mass converges to the upper endpoint in the sense that, for every \(\varepsilon>0\), its mass on \([-\beta_j,\beta_j-\varepsilon]\) tends to zero.  Since the posterior variable is bounded by \(\beta_j\), its posterior mean converges to \(\beta_j\).  The negative limit follows from part (i), proving (iii).

For the fixed-rate Laplace likelihood, if \(d\geq\beta_j\), then \(|d-\theta|=d-\theta\) throughout the support of the slab.  Therefore,
\[
    L_L(d\mid\theta,\lambda)
    =\frac{a}{2}e^{-ad}e^{a\theta},
    \qquad
    L_L(d\mid0,\lambda)=\frac{a}{2}e^{-ad}.
\]
Substitution into the numerator and denominator of \eqref{eq:general-shrinkage-rule}, followed by cancellation of the common factor \(a e^{-ad}/2\), gives the first limit in (iv).  If \(d\leq-\beta_j\), then \(|d-\theta|=\theta-d\), which yields the second limit.  Symmetry of \(g_j\) gives \(A_{0,j}^{-}=A_{0,j}^{+}\) and \(A_{1,j}^{-}=-A_{1,j}^{+}\), so the two limits are negatives of one another.
\end{proof}

\section{Useful identities for implementation}
\label{app:identities}

For \(a>0\), the first six elementary integrals needed for the Wendland component can be generated recursively from
\begin{equation*}
    \int e^{bt}\,dt=\frac{e^{bt}}{b},
    \qquad
    \int t^m e^{bt}\,dt
    =\frac{t^m e^{bt}}{b}-\frac{m}{b}\int t^{m-1}e^{bt}\,dt,
    \qquad b\ne0.
    \label{eq:integration-recursion}
\end{equation*}
The finite-sum expression in \eqref{eq:elementary-integral-closed} is preferable in software because it avoids recursive numerical integration. For the semicircle term, the substitution in \eqref{eq:semicircle-moment-integral} removes the square-root endpoint singularity: the transformed integrand contains \(\sin^2(t)\), which is smooth at \(t=0\) and \(t=\pi\). When \(|d|<\beta\), splitting the integral at \(t_d=\arccos(d/\beta)\) removes the absolute-value kink and improves quadrature accuracy.

\bibliographystyle{jasa}
\bibliography{mybib}

\begin{thebibliography}{21}
\newcommand{\enquote}[1]{``#1''}
\expandafter\ifx\csname natexlab\endcsname\relax\def\natexlab#1{#1}\fi

\bibitem[\protect\citename{Abramovich and Benjamini,
  }1996]{AbramovichBenjamini1996}
Abramovich, F. and Benjamini, Y. (1996).
\newblock \enquote{Adaptive thresholding of wavelet coefficients.}
\newblock {\em Computational Statistics \& Data Analysis\/}, 22, 4, 351--361.

\bibitem[\protect\citename{Angelini and Vidakovic,
  }2004]{AngeliniVidakovic2004}
Angelini, C. and Vidakovic, B. (2004).
\newblock \enquote{{$\Gamma$}-minimax wavelet shrinkage: a robust incorporation
  of information about energy of a signal in denoising applications.}
\newblock {\em Statistica Sinica\/}, 14, 1, 103--125.

\bibitem[\protect\citename{Barrios and dos Santos~Sousa,
  }2025]{BarriosSousa2025}
Barrios, F. A.~C. and dos Santos~Sousa, A.~R. (2025).
\newblock \enquote{Bayesian wavelet shrinkage for low {SNR} data based on the
  {Epanechnikov} kernel.}
\newblock {\em arXiv preprint arXiv:2507.11718\/}.

\bibitem[\protect\citename{Chipman et~al., }1997]{ChipmanKolaczykMcCulloch1997}
Chipman, H.~A., Kolaczyk, E.~D., and McCulloch, R.~E. (1997).
\newblock \enquote{Adaptive Bayesian wavelet shrinkage.}
\newblock {\em Journal of the American Statistical Association\/}, 92, 440,
  1413--1421.

\bibitem[\protect\citename{Daubechies, }1992]{Daubechies1992}
Daubechies, I. (1992).
\newblock {\em Ten lectures on wavelets\/}.
\newblock No.~61 in CBMS-NSF Regional Conference Series in Applied Mathematics.
  Philadelphia, PA: SIAM.

\bibitem[\protect\citename{Donoho and Johnstone, }1994]{DonohoJohnstone1994}
Donoho, D.~L. and Johnstone, I.~M. (1994).
\newblock \enquote{Ideal spatial adaptation by wavelet shrinkage.}
\newblock {\em Biometrika\/}, 81, 3, 425--455.

\bibitem[\protect\citename{Donoho and Johnstone, }1995]{DonohoJohnstone1995}
--- (1995).
\newblock \enquote{Adapting to unknown smoothness via wavelet shrinkage.}
\newblock {\em Journal of the American Statistical Association\/}, 90, 432,
  1200--1224.

\bibitem[\protect\citename{dos S.~Sousa et~al.,
  }2021]{SousaGarciaVidakovic2021}
dos S.~Sousa, A.~R., Garcia, N.~L., and Vidakovic, B. (2021).
\newblock \enquote{Bayesian wavelet shrinkage with beta priors.}
\newblock {\em Computational Statistics\/}, 36, 2, 1341--1363.

\bibitem[\protect\citename{dos Santos~Sousa, }2022]{Sousa2022}
dos Santos~Sousa, A.~R. (2022).
\newblock \enquote{Bayesian wavelet shrinkage with logistic prior.}
\newblock {\em Communications in Statistics - Simulation and Computation\/},
  51, 8, 4700--4714.

\bibitem[\protect\citename{Johnstone and Silverman,
  }2005]{JohnstoneSilverman2005}
Johnstone, I.~M. and Silverman, B.~W. (2005).
\newblock \enquote{Empirical {Bayes} selection of wavelet thresholds.}
\newblock {\em The Annals of Statistics\/}, 33, 4, 1700--1752.

\bibitem[\protect\citename{Mallat, }1999]{Mallat2009}
Mallat, S. (1999).
\newblock {\em A wavelet tour of signal processing\/}.
\newblock 2nd ed. Academic Press.

\bibitem[\protect\citename{Marchesi~Reina and dos S.~Sousa,
  }2026]{ReinaSousa2026}
Marchesi~Reina, J. and dos S.~Sousa, A.~R. (2026).
\newblock \enquote{Wavelet shrinkage based on the raised cosine prior.}
\newblock {\em Journal of Statistical Computation and Simulation\/}, 96, 11,
  2683--2704.

\bibitem[\protect\citename{Nason, }1996]{Nason1996}
Nason, G.~P. (1996).
\newblock \enquote{Wavelet shrinkage using cross-validation.}
\newblock {\em Journal of the Royal Statistical Society: Series B
  (Methodological)\/}, 58, 2, 463--479.

\bibitem[\protect\citename{Rem{\'e}nyi and Vidakovic,
  }2015]{RemenyiVidakovic2015}
Rem{\'e}nyi, N. and Vidakovic, B. (2015).
\newblock \enquote{Wavelet shrinkage with double {Weibull} prior.}
\newblock {\em Communications in Statistics - Simulation and Computation\/},
  44, 1, 88--104.

\bibitem[\protect\citename{Sanyal, }2025]{Sanyal2025}
Sanyal, N. (2025).
\newblock \enquote{Nonlocal prior mixture-based {Bayesian} wavelet regression
  with application to noisy imaging and audio data.}
\newblock {\em Mathematics\/}, 13, 16, 2642.

\bibitem[\protect\citename{Sanyal and Ferreira, }2017]{SanyalFerreira2017}
Sanyal, N. and Ferreira, M. A.~R. (2017).
\newblock \enquote{Bayesian wavelet analysis using nonlocal priors with an
  application to {fMRI} analysis.}
\newblock {\em Sankhy{\=a}: The Indian Journal of Statistics, Series B\/}, 79,
  2, 361--388.

\bibitem[\protect\citename{Vidakovic and Ruggeri, }2001]{VidakovicRuggeri2001}
Vidakovic, B. and Ruggeri, F. (2001).
\newblock \enquote{{BAMS} method: theory and simulations.}
\newblock {\em Sankhy{\=a}: The Indian Journal of Statistics, Series B\/}, 63,
  2, 234--249.

\bibitem[\protect\citename{Vimalajeewa et~al.,
  }2023]{VimalajeewaDasguptaRuggeriVidakovic2023}
Vimalajeewa, D., DasGupta, A., Ruggeri, F., and Vidakovic, B. (2023).
\newblock \enquote{Gamma-minimax wavelet shrinkage for signals with low {SNR}.}
\newblock {\em The New England Journal of Statistics in Data Science\/}, 1, 2,
  159--171.

\bibitem[\protect\citename{Wendland, }1995]{Wendland1995}
Wendland, H. (1995).
\newblock \enquote{Piecewise polynomial, positive definite and compactly
  supported radial functions of minimal degree.}
\newblock {\em Advances in Computational Mathematics\/}, 4, 1, 389--396.

\bibitem[\protect\citename{Wendland, }2005]{Wendland2005}
--- (2005).
\newblock {\em Scattered data approximation\/}.
\newblock No.~17 in Cambridge Monographs on Applied and Computational
  Mathematics. Cambridge: Cambridge University Press.

\bibitem[\protect\citename{Wigner, }1993]{Wigner1955}
Wigner, E.~P. (1993).
\newblock \enquote{Characteristic vectors of bordered matrices with infinite
  dimensions {I}.}
\newblock In {\em The Collected Works of Eugene Paul Wigner: Part A: The
  Scientific Papers\/}, ed. A.~S. Wightman, vol. A/1 of {\em The Collected
  Works of Eugene Paul Wigner\/},  524--540. Berlin, Heidelberg: Springer.

\end{thebibliography}

\clearpage


\renewcommand{\figurename}{Supplementary Figure}
\renewcommand{\tablename}{Supplementary Table}

\vspace*{8cm}
\begin{center}
{\bfseries \Huge Supplement to ``Resolution-Adaptive Compact-Support Priors for Bayesian Wavelet Denoising" by Nilotpal Sanyal}
\end{center}

\newpage

\section{Estimated noise-scale diagnostic}
\label{app:noise-scale}

Throughout this supplement, WS denotes Wendland--semicircle; SNR denotes signal-to-noise ratio; MSE and MAE denote mean squared error and mean absolute error; SD denotes standard deviation; and the benchmark labels Univ, FDR, CV, SURE, BAMS, and NLP denote universal thresholding, false-discovery-rate thresholding, cross-validation, Stein's unbiased risk estimate, the Bayesian adaptive multiresolution shrinker, and a nonlocal prior, respectively.

Supplementary Figure~\ref{fig:simulation-sigma} reports the replicate-level distributions of
\(\widehat\sigma\) from the WS--Gaussian fits by signal, SNR, and sample size.
The estimated noise scale is largest at \(\rho=0.2\), decreases at
\(\rho=1\), and is smallest at \(\rho=3\).  The close clustering across sample
sizes within each signal--SNR combination indicates that the scale estimate is
not materially driven by the sample size in this simulation.

\begin{figure}[htbp]
    \centering
    \includegraphics[width=\linewidth]{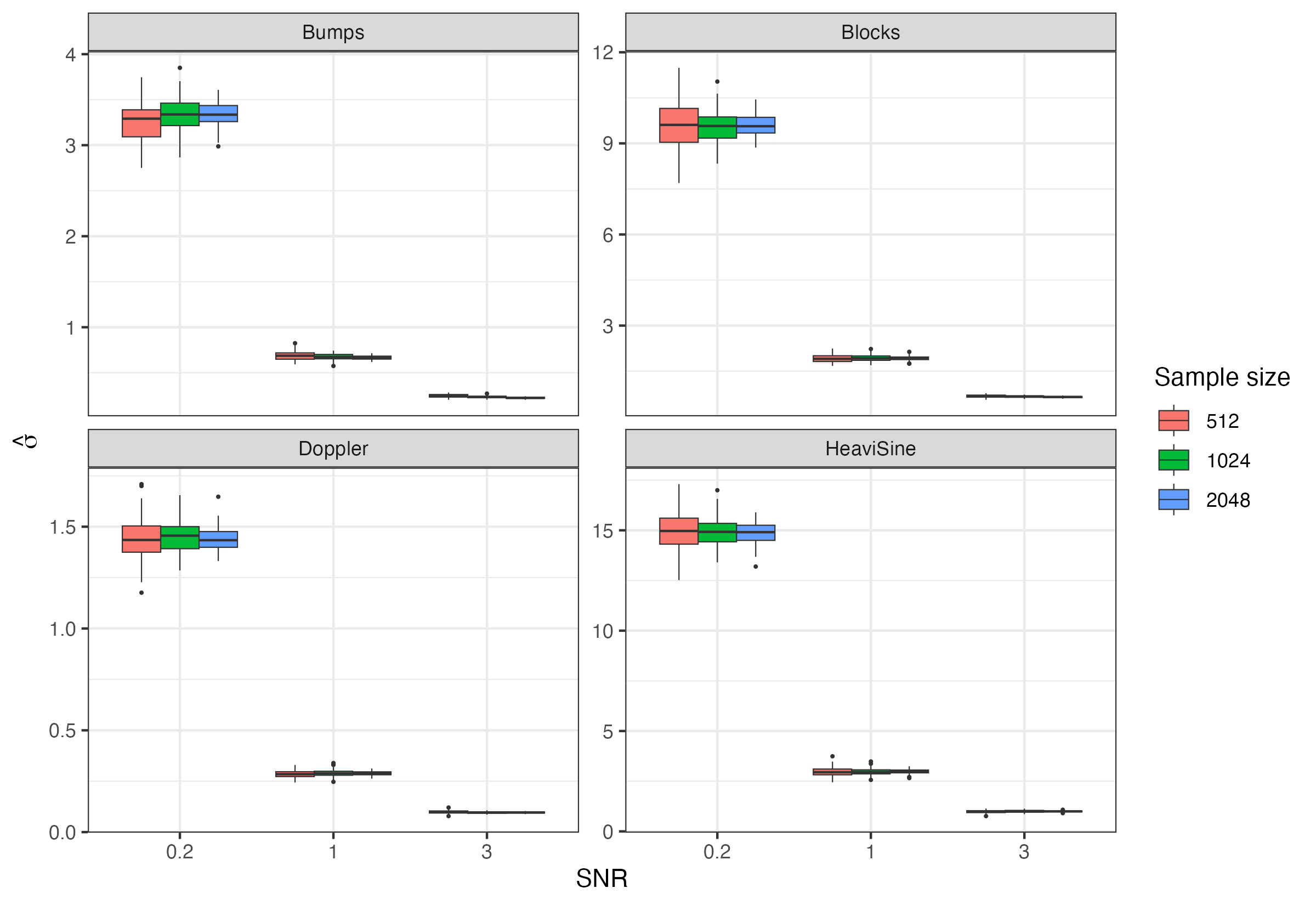}
    \caption{Replicate-level distributions of the estimated noise scale
    \(\widehat\sigma\) from WS--Gaussian fits, by signal, SNR, and sample size.
    Each box summarizes 100 independent replications.}
    \label{fig:simulation-sigma}
\end{figure}

\clearpage


\section{Supplementary simulation analyses}
\label{sec:supp-simulation}
This supplement reports the numerical, sensitivity, and robustness analyses that support the primary Gaussian-error simulation in Section~6 of the main manuscript.  Unless stated otherwise, each reported mean is based on 100 common-random-number replications per design cell, and entries in parentheses in the MSE tables are Monte Carlo standard errors.  WS--G and WS--L denote the Gaussian- and Laplace-likelihood versions of the adaptive Wendland--semicircle fit.

For clarity, the simulation tables in this supplement use computationally generated observations from the deterministic test functions and error mechanisms described in Section~6 of the main manuscript. They are not derived from the seismic record except in the explicitly labeled seismic and semi-synthetic sections.

\subsection{Quadrature assessment}
\label{sec:supp-quadrature}
The proposed estimator contains one-dimensional component integrals.  We evaluated quadrature orders \(q\in\{32,48,64,96,128\}\) for both coefficientwise likelihoods at \(n=1024\) and \(\rho\in\{0.2,1,3\}\), using \(q=128\) as the numerical reference.  For each order, the reporting script recorded the 95th percentile of the relative MSE change, the reconstruction NRMSE relative to \(q=128\), the component-conditional-mean RMSE after scaling by \(\widehat\sigma\), the maximum and median absolute reconstruction differences relative to \(q=128\), the failure and numerical-fallback frequencies, the coefficient-level check for small, moderate, near-boundary, and beyond-support coefficients, the median runtime, and the warning rate.  Table~\ref{tab:supp-quadrature} reports these diagnostics and identifies the smallest order satisfying all prespecified tolerances.  The selected order, \(q=128\), was then used in the remaining analyses; it was not selected from the primary evaluation outcomes.

\begin{table}[p]
    \centering
    \caption{Quadrature diagnostics relative to (q=128).  The relative MSE change, reconstruction NRMSE, and component-conditional-mean RMSE are 95th-percentile summaries over the quadrature study; the selected order is the smallest order satisfying all prespecified tolerances.\\[-5pt]}
    \label{tab:supp-quadrature}
    \renewcommand{\arraystretch}{1}
    \resizebox{\linewidth}{!}{%
    \begin{tabular}{llrrrrrrr}
        \toprule
        Likelihood & $q$ & Relative MSE change & Reconstruction NRMSE & Component RMSE/$\widehat\sigma$ & Median runtime (s) & Warnings (\%) & Stable & Selected $q$ \\
        \midrule
        gaussian & 32 & 6.94e-05 & 5.54e-05 & 5.68e-02 & 0.215 & 0.0 & No & 128 \\
        gaussian & 48 & 1.34e-06 & 1.70e-06 & 2.15e-02 & 0.221 & 0.0 & No & 128 \\
        gaussian & 64 & 7.87e-08 & 5.70e-07 & 1.19e-02 & 0.230 & 0.0 & No & 128 \\
        gaussian & 96 & 9.94e-09 & 8.04e-08 & 1.65e-03 & 0.244 & 0.0 & No & 128 \\
        gaussian & 128 & 0.00e+00 & 0.00e+00 & 0.00e+00 & 0.257 & 0.0 & Yes & 128 \\
        laplace & 32 & 2.78e-05 & 3.75e-03 & 3.74e-02 & 0.241 & 0.0 & No & 128 \\
        laplace & 48 & 1.19e-05 & 1.62e-03 & 1.50e-02 & 0.246 & 0.0 & No & 128 \\
        laplace & 64 & 6.87e-06 & 9.93e-04 & 8.44e-03 & 0.254 & 0.0 & No & 128 \\
        laplace & 96 & 3.45e-06 & 4.95e-04 & 3.48e-03 & 0.266 & 0.0 & No & 128 \\
        laplace & 128 & 0.00e+00 & 0.00e+00 & 0.00e+00 & 0.276 & 0.0 & Yes & 128 \\
        \bottomrule
    \end{tabular}
    }
\end{table}

Supplementary Figure~\ref{fig:supp-quadrature} displays the three numerical-error diagnostics as functions of \(q\). The figure makes visible whether increasing the quadrature order continues to change the fitted reconstruction or only increases computation.

\begin{figure}[htbp]
    \centering
    \includegraphics[width=\linewidth]{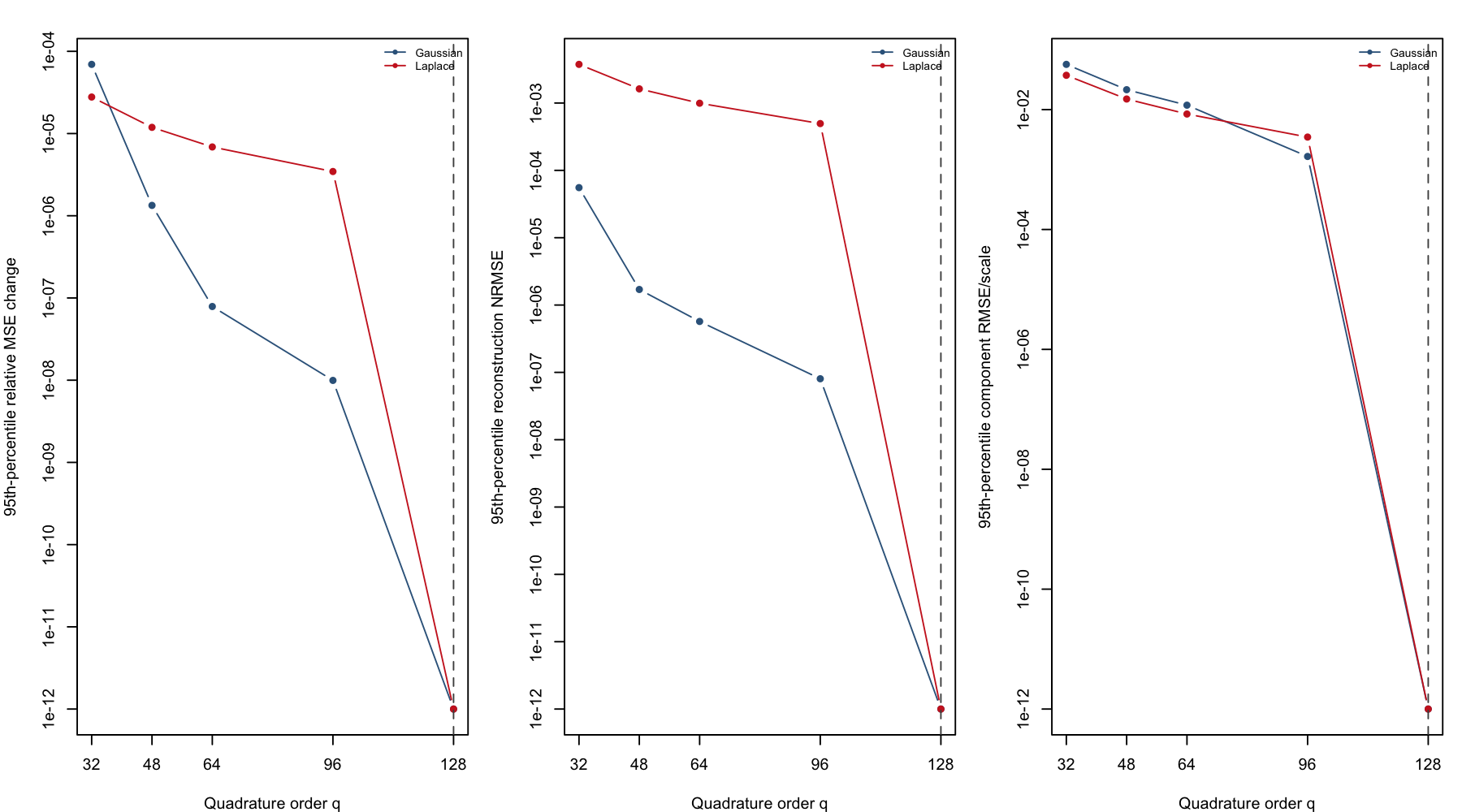}
    \caption{Quadrature diagnostics relative to \(q=128\) for the Gaussian and Laplace likelihoods.  Curves show the 95th-percentile relative MSE change, reconstruction NRMSE, and component-conditional-mean RMSE scaled by \(\widehat\sigma\); the vertical marker indicates the selected order.}
    \label{fig:supp-quadrature}
\end{figure}

\subsection{Endpoint and monotonicity checks}
\label{sec:supp-endpoints}
The implementation explicitly allows the spike probability and the Wendland mixture weight to attain their endpoints.  When \(\pi_j=0\), the posterior spike probability is exactly zero; when \(\pi_j=1\), the posterior mean is exactly zero.  Similarly, \(\omega_j=1\) and \(\omega_j=0\) reduce the slab to the Wendland-only and Semicircle-only cases, respectively.  We evaluated the posterior-mean map on a dense symmetric grid for all combinations of \(\pi_j\in\{0,0.5,1\}\), \(\omega_j\in\{0,0.5,1\}\), and both likelihoods.  Table~\ref{tab:supp-numerical-checks} reports oddness, monotonicity, support, finiteness, and endpoint diagnostics.  For the Gaussian-likelihood rows, the monotonicity results are consistent with the analytical identity \(\delta_j'(d)=\operatorname{Var}(\theta_j\mid d)/\sigma^2\geq0\).  The corresponding Laplace-likelihood rows are empirical numerical checks, for which this Gaussian variance identity does not apply.

\begin{table}[p]
    \centering
    \caption{Numerical checks for the posterior-mean map over endpoint spike probabilities and endpoint slab weights.  The minimum first difference checks nondecreasing behavior on the evaluation grid; the two endpoint columns check the exact \(\pi_j=0\) and \(\pi_j=1\) rules.\\[-5pt]}
    \label{tab:supp-numerical-checks}
    \tiny
    \renewcommand{\arraystretch}{0.92}
    \resizebox{\linewidth}{!}{%
    \begin{tabular}{llrrrrccc}
        \toprule
        Likelihood & \(\pi_j\) & \(\omega_j\) & Max oddness error & Min first difference & Max support excess & Finite & \(\pi_j=0\) & \(\pi_j=1\) \\
        \midrule
        gaussian & 0.0 & 0.0 & 7.55e-15 & 9.14e-04 & 0.00e+00 & Yes & Yes & -- \\
        gaussian & 0.0 & 0.5 & 7.99e-15 & 9.74e-04 & 0.00e+00 & Yes & Yes & -- \\
        gaussian & 0.0 & 1.0 & 5.77e-15 & 2.18e-03 & 0.00e+00 & Yes & Yes & -- \\
        gaussian & 0.5 & 0.0 & 7.55e-15 & 9.14e-04 & 0.00e+00 & Yes & -- & -- \\
        gaussian & 0.5 & 0.5 & 7.55e-15 & 9.74e-04 & 0.00e+00 & Yes & -- & -- \\
        gaussian & 0.5 & 1.0 & 5.77e-15 & 2.18e-03 & 0.00e+00 & Yes & -- & -- \\
        gaussian & 1.0 & 0.0 & 0.00e+00 & 0.00e+00 & 0.00e+00 & Yes & -- & Yes \\
        gaussian & 1.0 & 0.5 & 0.00e+00 & 0.00e+00 & 0.00e+00 & Yes & -- & Yes \\
        gaussian & 1.0 & 1.0 & 0.00e+00 & 0.00e+00 & 0.00e+00 & Yes & -- & Yes \\
        laplace & 0.0 & 0.0 & 1.55e-15 & 0.00e+00 & 0.00e+00 & Yes & Yes & -- \\
        laplace & 0.0 & 0.5 & 2.16e-13 & 0.00e+00 & 0.00e+00 & Yes & Yes & -- \\
        laplace & 0.0 & 1.0 & 5.18e-13 & 0.00e+00 & 0.00e+00 & Yes & Yes & -- \\
        laplace & 0.5 & 0.0 & 1.11e-15 & 0.00e+00 & 0.00e+00 & Yes & -- & -- \\
        laplace & 0.5 & 0.5 & 1.28e-13 & 0.00e+00 & 0.00e+00 & Yes & -- & -- \\
        laplace & 0.5 & 1.0 & 2.76e-13 & 0.00e+00 & 0.00e+00 & Yes & -- & -- \\
        laplace & 1.0 & 0.0 & 0.00e+00 & 0.00e+00 & 0.00e+00 & Yes & -- & Yes \\
        laplace & 1.0 & 0.5 & 0.00e+00 & 0.00e+00 & 0.00e+00 & Yes & -- & Yes \\
        laplace & 1.0 & 1.0 & 0.00e+00 & 0.00e+00 & 0.00e+00 & Yes & -- & Yes \\
        \bottomrule
    \end{tabular}
    }
\end{table}

\subsection{Fit-level quality control and computing environment}
\label{sec:supp-reproducibility}
Table~\ref{tab:supp-qc} summarizes the quality-control results.  The expanded Gaussian study contained 665 warning-bearing fits (\(1.73\%\)), and the Laplace study contained 193 warning-bearing fits (\(0.67\%\)).  The warning-source diagnostic identified all warnings as arising from the \texttt{FDR} method.  Of the 665 Gaussian-study warning-bearing fits, 662 occurred at \(\rho=0.2\) across \(n\in\{512,1024,2048\}\), and 3 occurred for \textit{Bumps} at \(n=1024,\rho=0.5\); all 193 Laplace-study warnings occurred at \(\rho=0.2\) across the three sample sizes.  All warning-bearing fits retained valid status, all reported fit metrics were finite, and no MSE value was non-finite.  Thus, the warnings were nonfatal numerical diagnostics and did not exclude or invalidate any estimate included in the reported summaries.
\begin{table}[p]
    \centering
    \caption{Quality-control summary for the saved replicate-level outputs.  Warning counts count fits with at least one recorded warning, and non-finite MSE counts include failed or invalid fits.\\[-5pt]}
    \label{tab:supp-qc}
    \scriptsize
    \renewcommand{\arraystretch}{0.92}
    \begin{tabular}{lrrrr}
        \toprule
        Study & Recorded fits & Failed fits & Fits with warnings & Non-finite MSE \\
        \midrule
        quadrature & 12,000 & 0 & 0 & 0 \\
        primary\_expanded\_snr & 38,400 & 0 & 665 & 0 \\
        component\_adaptivity & 28,800 & 0 & 0 & 0 \\
        hyperparameters & 21,600 & 0 & 0 & 0 \\
        basis\_j0 & 7,200 & 0 & 0 & 0 \\
        laplace\_errors & 28,800 & 0 & 193 & 0 \\
        \bottomrule
    \end{tabular}
\end{table}

\begin{table}
\centering
\caption{Computing environment used for the simulation and semi-synthetic analyses.}
\label{tab:supp-computing-environment}
\begin{tabular}{ll}
\toprule
Item & Specification \\
\midrule
Processor/node & Intel(R) Xeon(R) Gold 5120 CPU @ 2.20 GHz \\
Operating system & \texttt{Linux 5.14.0-284.169.1.el9\_2.x86\_64} \\
R version & 4.5.1 \\
Parallel backend & MPI via \texttt{Rmpi} \\
\texttt{snowfall} & \texttt{1.84.6.3} \\
\texttt{wavethresh} & \texttt{4.7.3} \\
\texttt{Rmpi} & \texttt{0.7.3.4} \\
\texttt{NLPwavelet} & \texttt{1.1} \\
\bottomrule
\end{tabular}
\end{table}

\subsection{Component and adaptivity comparisons}
\label{sec:supp-component}
The component/adaptivity study compared, for each likelihood, the adaptive mixture, the constant mixture, the Wendland-only endpoint, and the Semicircle-only endpoint.  The full Gaussian cellwise NMSE results are given in Supplementary Table~\ref{tab:supp-component}; the main manuscript reports the compact aggregate comparison in Table~4 and shows representative fitted \(\widehat\omega_j\) profiles in Figure~4.  The four models have the same empirical support-scale rule and differ only in how the two slab components are combined.  This isolates the contribution of mixing from the contribution of allowing the mixture weight to vary across resolution.

\begin{table}[p]
    \centering
    \caption{Mean NMSE for the four Gaussian WS component and adaptivity specifications in the component/adaptivity study.  The smallest value in each cell is shown in bold.\\[-5pt]}
    \label{tab:supp-component}
    \scriptsize
    \renewcommand{\arraystretch}{0.92}
    \resizebox{\linewidth}{!}{%
    \begin{tabular}{lllrrrr}
        \toprule
        Signal & n & SNR & WS--G & WS--G constant & Wendland-only & Semicircle-only \\
        \midrule
        Bumps & 512 & 0.2 & 1.036 & 1.037 & \textbf{1.030} & 1.032 \\
        Bumps & 512 & 1.0 & 0.526 & 0.526 & 0.636 & \textbf{0.526} \\
        Bumps & 512 & 3.0 & 0.114 & 0.114 & 0.140 & \textbf{0.114} \\
        Blocks & 512 & 0.2 & 0.906 & 0.909 & 0.962 & \textbf{0.905} \\
        Blocks & 512 & 1.0 & 0.246 & 0.246 & 0.314 & \textbf{0.246} \\
        Blocks & 512 & 3.0 & 0.084 & 0.084 & 0.095 & \textbf{0.084} \\
        Doppler & 512 & 0.2 & 0.923 & 0.924 & 0.972 & \textbf{0.922} \\
        Doppler & 512 & 1.0 & 0.184 & 0.184 & 0.236 & \textbf{0.184} \\
        Doppler & 512 & 3.0 & \textbf{0.038} & 0.038 & 0.046 & 0.038 \\
        HeaviSine & 512 & 0.2 & 0.544 & 0.545 & 0.692 & \textbf{0.543} \\
        HeaviSine & 512 & 1.0 & 0.054 & 0.054 & 0.080 & \textbf{0.054} \\
        HeaviSine & 512 & 3.0 & 0.013 & 0.013 & 0.017 & \textbf{0.013} \\
        Bumps & 1024 & 0.2 & 0.965 & 0.965 & 0.992 & \textbf{0.963} \\
        Bumps & 1024 & 1.0 & 0.339 & 0.339 & 0.408 & \textbf{0.339} \\
        Bumps & 1024 & 3.0 & 0.066 & 0.066 & 0.075 & \textbf{0.066} \\
        Blocks & 1024 & 0.2 & 0.709 & 0.710 & 0.832 & \textbf{0.709} \\
        Blocks & 1024 & 1.0 & 0.169 & 0.169 & 0.200 & \textbf{0.169} \\
        Blocks & 1024 & 3.0 & 0.057 & 0.057 & 0.066 & \textbf{0.057} \\
        Doppler & 1024 & 0.2 & 0.819 & 0.821 & 0.914 & \textbf{0.818} \\
        Doppler & 1024 & 1.0 & 0.120 & 0.120 & 0.155 & \textbf{0.120} \\
        Doppler & 1024 & 3.0 & \textbf{0.020} & 0.020 & 0.026 & 0.020 \\
        HeaviSine & 1024 & 0.2 & \textbf{0.297} & 0.297 & 0.387 & 0.297 \\
        HeaviSine & 1024 & 1.0 & 0.029 & 0.029 & 0.042 & \textbf{0.029} \\
        HeaviSine & 1024 & 3.0 & 0.009 & 0.009 & 0.012 & \textbf{0.009} \\
        Bumps & 2048 & 0.2 & 0.917 & 0.917 & 0.966 & \textbf{0.916} \\
        Bumps & 2048 & 1.0 & 0.208 & 0.208 & 0.241 & \textbf{0.208} \\
        Bumps & 2048 & 3.0 & 0.038 & \textbf{0.038} & 0.043 & 0.038 \\
        Blocks & 2048 & 0.2 & 0.519 & 0.519 & 0.646 & \textbf{0.519} \\
        Blocks & 2048 & 1.0 & \textbf{0.129} & 0.129 & 0.145 & 0.129 \\
        Blocks & 2048 & 3.0 & \textbf{0.035} & 0.035 & 0.040 & 0.035 \\
        Doppler & 2048 & 0.2 & 0.621 & 0.621 & 0.771 & \textbf{0.621} \\
        Doppler & 2048 & 1.0 & \textbf{0.069} & 0.069 & 0.088 & 0.069 \\
        Doppler & 2048 & 3.0 & \textbf{0.010} & 0.010 & 0.013 & 0.010 \\
        HeaviSine & 2048 & 0.2 & \textbf{0.188} & 0.188 & 0.239 & 0.188 \\
        HeaviSine & 2048 & 1.0 & 0.018 & 0.018 & 0.025 & \textbf{0.018} \\
        HeaviSine & 2048 & 3.0 & 0.006 & 0.006 & 0.007 & \textbf{0.006} \\
        \bottomrule
    \end{tabular}
    }
\end{table}

Supplementary Table~\ref{tab:supp-component} shows that the Semicircle-only endpoint has the lowest NMSE in 26 of the 36 Gaussian cells, whereas the adaptive mixture is lowest in 8 cells.  The fitted weights in Figure~5 of the main manuscript are positive but very close to zero, indicating that the adaptive fit remains near the Semicircle-only endpoint.  The modest resolution pattern should therefore be interpreted as weak data-driven variation around that endpoint, rather than as evidence of substantial scale-varying use of both slab components.  The fixed endpoint fits consequently provide important benchmarks for interpreting the adaptive mixture.

Across the component/adaptivity study, the adaptive and constant-mixture fits had optimizer-bound-hit rates of 84\% and 91\%, respectively, and mean runtimes of 0.28 and 0.17 s per replicate; these diagnostics were recorded independently of the accuracy summaries.

\subsection{Hyperparameters and bounded-support behavior}
\label{sec:supp-hyperparameters}

\subsubsection{Sensitivity to \(\gamma\), \(\varkappa\), and \(\ell\)}
\label{sec:supp-hyperparameters-detail}
At \(n=1024\), we varied one fixed quantity at a time around the default \((\gamma,\varkappa,\ell)=(2.4,0.99,2)\).  The values were \(\gamma\in\{1.8,2.1,2.4,2.7,3.0\}\), \(\varkappa\in\{0.95,0.975,0.99\}\), and \(\ell\in\{1.5,2,3\}\).  The purpose was to evaluate stability of the main conclusions, not to select a separate value for each signal.  Table~\ref{tab:supp-hyperparameter} reports the Gaussian WS mean NMSE averaged over all four signals and the three original SNR values, together with its percentage change from the default.  The corresponding Laplace fits are retained in the saved replicate-level output and are available from the same reporting script. Across these prespecified grids, the mean NMSE changed from \(-4.7\%\) to \(+6.4\%\) over \(\gamma\), remained within \(-0.1\%\) to \(+0.8\%\) over \(\varkappa\), and ranged from \(-2.5\%\) to \(+3.9\%\) over \(\ell\).

\begin{table}[p]
    \centering
    \caption{Gaussian WS hyperparameter sensitivity at (n=1024), averaged over the four signals and the three original SNR values.  Each parameter is varied one at a time around the default (\(\gamma,\varkappa,\ell)=(2.4,0.99,2)\).\\[-5pt]}
    \label{tab:supp-hyperparameter}
    \small
    \renewcommand{\arraystretch}{0.92}
    \begin{tabular}{lrrr}
        \toprule
        Parameter & Value & Mean NMSE & Change from default (\%) \\
        \midrule
        gamma & 1.8 & 0.286 & -4.7 \\
        gamma & 2.1 & 0.291 & -2.9 \\
        gamma & 2.4 & 0.300 & 0.0 \\
        gamma & 2.7 & 0.309 & 3.2 \\
        gamma & 3.0 & 0.319 & 6.4 \\
        kappa & 0.950 & 0.302 & 0.8 \\
        kappa & 0.975 & 0.299 & -0.1 \\
        kappa & 0.990 & 0.300 & 0.0 \\
        ell & 1.5 & 0.292 & -2.5 \\
        ell & 2.0 & 0.300 & 0.0 \\
        ell & 3.0 & 0.312 & 3.9 \\
        \bottomrule
    \end{tabular}
\end{table}

Supplementary Figure~\ref{fig:supp-hyperparameters} plots the same relative changes.  The results show whether the conclusion depends on a narrow tuning choice or remains stable over the prespecified perturbations.

\begin{figure}[htbp]
    \centering
    \includegraphics[width=\linewidth]{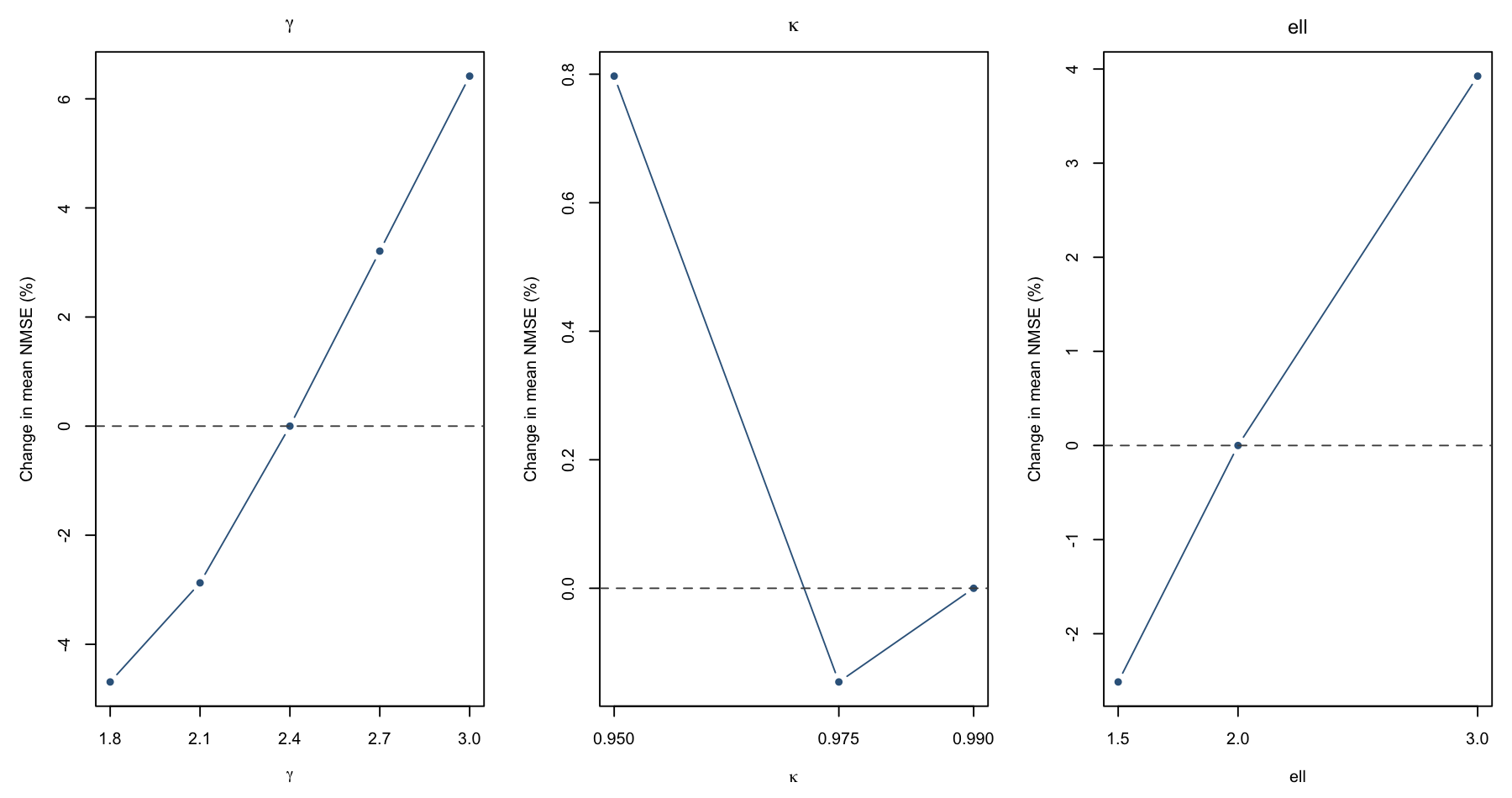}
    \caption{Relative change in mean Gaussian WS NMSE from the default hyperparameter setting, shown for one-at-a-time changes in \(\gamma\), \(\varkappa\), and \(\ell\) at \(n=1024\).  Each point averages the four signals and three original SNR values.}
    \label{fig:supp-hyperparameters}
\end{figure}

\subsubsection{Support exceedance and saturation}
\label{sec:supp-support}
The compact support of the slab implies that posterior means cannot exceed the fitted support scale in magnitude.  To assess the practical effect of this saturation, we transformed each known test signal to the wavelet domain and recorded, for every fitted resolution, the proportion of true coefficients outside \([-\widehat\beta_j,\widehat\beta_j]\), the fraction of replications in which at least one coefficient was outside the interval, and the proportion of true wavelet energy outside the interval.  Table~\ref{tab:supp-support} reports these quantities for the default WS--Gaussian fit at \(n=1024\); Supplementary Figure~\ref{fig:supp-support} shows the mean coefficient-level exceedance rate by signal and SNR.

\begin{table}[p]
    \centering
    \caption{Support-exceedance diagnostics for the default WS--Gaussian fit at (n=1024).  The first column is the mean coefficient-level exceedance percentage, the second is the percentage of replications with at least one exceedance at a level, and the third is the percentage of true wavelet energy outside the fitted supports; values are averaged over fitted detail levels.\\[-5pt]}
    \label{tab:supp-support}
    \scriptsize
    \renewcommand{\arraystretch}{0.92}
    \resizebox{\linewidth}{!}{%
    \begin{tabular}{lrrrr}
        \toprule
        Signal & SNR & Outside coefficients (\%) & Any outside (\%) & Outside energy (\%) \\
        \midrule
        Bumps & 0.2 & 5.9 & 7.8 & 6.5 \\
        Bumps & 1.0 & 8.8 & 35.8 & 15.9 \\
        Bumps & 3.0 & 10.5 & 64.1 & 25.2 \\
        Blocks & 0.2 & 3.8 & 10.1 & 8.4 \\
        Blocks & 1.0 & 9.6 & 25.3 & 17.7 \\
        Blocks & 3.0 & 9.3 & 42.1 & 19.6 \\
        Doppler & 0.2 & 8.3 & 13.0 & 11.4 \\
        Doppler & 1.0 & 10.2 & 26.5 & 19.1 \\
        Doppler & 3.0 & 8.9 & 36.2 & 20.4 \\
        HeaviSine & 0.2 & 3.1 & 6.6 & 4.7 \\
        HeaviSine & 1.0 & 9.3 & 15.3 & 11.6 \\
        HeaviSine & 3.0 & 9.6 & 21.2 & 14.0 \\
        \bottomrule
    \end{tabular}
    }
\end{table}

Across the four signals and three SNR values, coefficient-level exceedance ranged from \(3.1\%\) to \(10.5\%\), the percentage of replications with at least one exceedance ranged from \(6.6\%\) to \(64.1\%\), and the proportion of true wavelet energy outside the fitted support ranged from \(4.7\%\) to \(25.2\%\).  At coarse resolutions, the empirical \(0.99\) support quantile is estimated from relatively few coefficients and should therefore be interpreted as a data-adaptive truncation rule, not as a guarantee that every future coefficient will lie within \([-\widehat\beta_j,\widehat\beta_j]\).  The reported exceedance and energy summaries quantify how often this distinction matters for the simulated signals.

\begin{figure}[htbp]
    \centering
    \includegraphics[width=\linewidth]{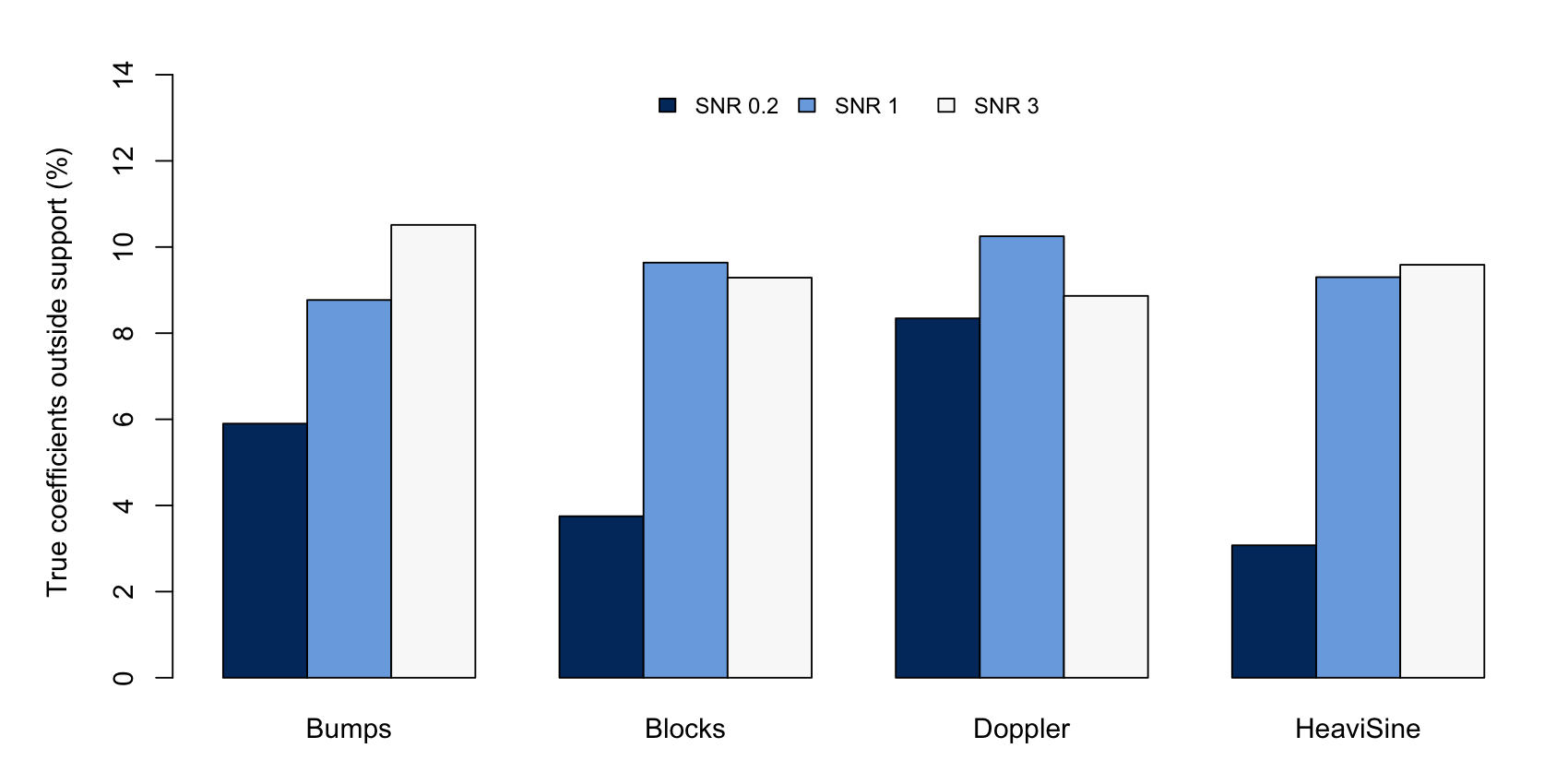}
    \caption{Mean proportion of true wavelet coefficients outside the fitted support interval for the default WS--Gaussian fit at \(n=1024\), by signal and amplitude SNR.}
    \label{fig:supp-support}
\end{figure}

\subsection{Wider SNR range and wavelet sensitivity}
\label{sec:supp-basis}
The added \(n=1024\) SNR values \(\rho\in\{0.5,2,5\}\) are reported in Supplementary Table~\ref{tab:supp-extended-snr} for selected methods (WS--G, WS--L, CV, and NLP).  This compact selection represents both proposed likelihoods and two benchmark procedures; the complete method comparison for the original Gaussian-error grid is reported in Table~2 of the main manuscript, while the expanded SNR profiles are shown in Figure~4 of the main manuscript.  These results show whether the low-SNR comparison is an isolated endpoint phenomenon or part of a broader transition as the noise level decreases.

\begin{table}[p]
    \centering
    \caption{Mean NMSE (Monte Carlo standard error) for selected methods (WS--G, WS--L, CV, and NLP) at the additional \(n=1024\) amplitude SNR values in the expanded primary Gaussian-error study.  The table provides a compact selection representing both proposed likelihoods and two benchmark procedures; the complete method comparison is reported in the primary Gaussian-error results.\\[-5pt]}
    \label{tab:supp-extended-snr}
    \scriptsize
    \renewcommand{\arraystretch}{0.92}
    \resizebox{\linewidth}{!}{%
    \begin{tabular}{llrrrr}
        \toprule
        Signal & SNR & WS--G & WS--L & CV & NLP \\
        \midrule
        Bumps & 0.5 & 0.690 (0.005) & 0.921 (0.003) & 0.679 (0.006) & 0.459 (0.005) \\
        Bumps & 2.0 & 0.128 (0.001) & 0.320 (0.002) & 0.102 (0.001) & 0.075 (0.001) \\
        Bumps & 5.0 & 0.029 (0.000) & 0.089 (0.001) & 0.029 (0.000) & 0.015 (0.000) \\
        Blocks & 0.5 & 0.334 (0.004) & 0.494 (0.006) & 0.390 (0.005) & 0.237 (0.004) \\
        Blocks & 2.0 & 0.092 (0.001) & 0.151 (0.001) & 0.078 (0.001) & 0.056 (0.000) \\
        Blocks & 5.0 & 0.027 (0.000) & 0.069 (0.000) & 0.020 (0.000) & 0.014 (0.000) \\
        Doppler & 0.5 & 0.284 (0.005) & 0.593 (0.008) & 0.387 (0.005) & 0.204 (0.005) \\
        Doppler & 2.0 & 0.041 (0.001) & 0.104 (0.001) & 0.049 (0.000) & 0.027 (0.000) \\
        Doppler & 5.0 & 0.008 (0.000) & 0.024 (0.000) & 0.010 (0.000) & 0.006 (0.000) \\
        HeaviSine & 0.5 & 0.095 (0.003) & 0.144 (0.003) & 0.163 (0.004) & 0.064 (0.004) \\
        HeaviSine & 2.0 & 0.013 (0.000) & 0.019 (0.000) & 0.021 (0.000) & 0.010 (0.000) \\
        HeaviSine & 5.0 & 0.005 (0.000) & 0.009 (0.000) & 0.006 (0.000) & 0.003 (0.000) \\
        \bottomrule
    \end{tabular}
    }
\end{table}

To assess sensitivity to the wavelet representation, we compared the default Daubechies extremal-phase basis with a least-asymmetric Daubechies basis, and compared \(J_0=0\) with \(J_0=1\), at \(n=1024\).  Table~\ref{tab:supp-basis} reports the Gaussian WS mean NMSE and its change from the default.  This is a limited robustness check rather than an exhaustive comparison of wavelet families.

\begin{table}[p]
    \centering
    \caption{Gaussian WS sensitivity to the coarsest shrunk resolution and wavelet basis at (n=1024), averaged over the four signals and three original SNR values.\\[-5pt]}
    \label{tab:supp-basis}
    \scriptsize
    \renewcommand{\arraystretch}{0.92}
    \begin{tabular}{lrrr}
        \toprule
        Specification & Mean NMSE & Change from default (\%) & Runtime (s) \\
        \midrule
        WS--G default, J0=0 & 0.300 & 0.0 & 0.248 \\
        WS--G default, J0=1 & 0.285 & -4.9 & 0.243 \\
        WS--G least-asymmetric & 0.306 & 2.1 & 0.251 \\
        \bottomrule
    \end{tabular}
\end{table}

Relative to the default specification, setting \(J_0=1\) reduced the mean NMSE by \(4.9\%\) (relative change \(-4.9\%\)), whereas using the least-asymmetric basis increased it by \(2.1\%\) (relative change \(+2.1\%\)).  These modest changes indicate that the main conclusion is not highly sensitive to the examined coarsest-resolution or basis choice.

\subsection{Laplace-error sensitivity}
\label{sec:supp-laplace}
The Laplace-error experiment used the same test functions, sample sizes, SNR values, wavelet implementation, hyperparameter rules, competing procedures, and 100 replications as the original Gaussian grid.  The only change was the data-generating error distribution.  Specifically, \(\varepsilon_i\sim\operatorname{Laplace}(0,b)\) with \(b=\sigma/\sqrt{2}\), so that \(\operatorname{Var}(\varepsilon_i)=\sigma^2\).  The errors were generated in the time domain before the DWT.  Therefore, the transformed coefficients need not be independent Laplace variables; WS--Laplace is interpreted as a coefficientwise working likelihood in this sensitivity analysis, not as an exact transformed-data model.

\subsubsection{Cellwise and aggregate accuracy}
Supplementary Table~\ref{tab:supp-laplace-mse} reports the detailed Laplace-error MSE results.  Supplementary Figure~\ref{fig:supp-laplace-profile} shows the SNR profiles at \(n=1024\), and Supplementary Table~\ref{tab:supp-laplace-overall} gives descriptive averages over the 36 cells.  These summaries allow the effect of a non-Gaussian error distribution to be separated from the effect of changing the working likelihood.

\begin{table}[p]
    \centering
    \caption{Mean MSE across 100 replications under the Laplace-error data-generating mechanism.  Entries in parentheses are Monte Carlo standard errors, and the smallest MSE in each cell is shown in bold.\\[-5pt]}
    \label{tab:supp-laplace-mse}
    \scriptsize
    \renewcommand{\arraystretch}{0.92}
    \resizebox{\linewidth}{!}{%
    \begin{tabular}{llrrrrrrrr}
        \toprule
        Signal & SNR & WS--G & WS--L & Univ & FDR & CV & SURE & BAMS & NLP \\
        \midrule
        \multicolumn{10}{c}{\(n=512\)} \\
                \midrule
        Bumps & 0.2 & 0.536 (0.016) & 0.440 (0.003) & 0.457 (0.003) & 0.460 (0.004) & 0.463 (0.004) & 0.954 (0.046) & \textbf{0.439 (0.003)} & 1.097 (0.058) \\
        Bumps & 1.0 & 0.217 (0.002) & 0.343 (0.002) & 0.361 (0.002) & 0.296 (0.003) & 0.189 (0.003) & 0.185 (0.003) & 0.372 (0.002) & \textbf{0.140 (0.002)} \\
        Bumps & 3.0 & 0.046 (0.001) & 0.135 (0.002) & 0.171 (0.001) & 0.084 (0.001) & 0.052 (0.001) & 0.036 (0.000) & 0.222 (0.002) & \textbf{0.027 (0.000)} \\
        Blocks & 0.2 & 3.874 (0.121) & 3.590 (0.036) & 3.848 (0.035) & 3.859 (0.036) & 3.834 (0.035) & 8.337 (0.430) & \textbf{3.558 (0.034)} & 7.245 (0.483) \\
        Blocks & 1.0 & 0.836 (0.011) & 1.365 (0.009) & 1.596 (0.014) & 1.341 (0.018) & 0.942 (0.011) & 0.987 (0.014) & 1.463 (0.010) & \textbf{0.755 (0.018)} \\
        Blocks & 3.0 & 0.281 (0.003) & 0.493 (0.004) & 0.690 (0.005) & 0.433 (0.005) & 0.225 (0.003) & 0.229 (0.003) & 0.663 (0.007) & \textbf{0.172 (0.002)} \\
        Doppler & 0.2 & 0.096 (0.003) & \textbf{0.082 (0.001)} & 0.090 (0.001) & 0.089 (0.001) & 0.090 (0.001) & 0.210 (0.009) & 0.082 (0.001) & 0.182 (0.011) \\
        Doppler & 1.0 & 0.015 (0.000) & 0.028 (0.000) & 0.036 (0.000) & 0.029 (0.000) & 0.018 (0.000) & 0.020 (0.000) & 0.045 (0.001) & \textbf{0.014 (0.000)} \\
        Doppler & 3.0 & 0.003 (0.000) & 0.007 (0.000) & 0.010 (0.000) & 0.006 (0.000) & 0.003 (0.000) & 0.003 (0.000) & 0.015 (0.000) & \textbf{0.002 (0.000)} \\
        HeaviSine & 0.2 & \textbf{5.873 (0.299)} & 6.391 (0.167) & 7.822 (0.175) & 8.258 (0.164) & 7.610 (0.187) & 18.900 (1.166) & 7.216 (0.154) & 15.560 (1.270) \\
        HeaviSine & 1.0 & \textbf{0.491 (0.017)} & 0.806 (0.014) & 1.252 (0.022) & 1.273 (0.045) & 0.852 (0.018) & 1.161 (0.041) & 0.833 (0.013) & 0.742 (0.047) \\
        HeaviSine & 3.0 & \textbf{0.113 (0.002)} & 0.146 (0.002) & 0.293 (0.004) & 0.257 (0.006) & 0.171 (0.003) & 0.190 (0.004) & 0.297 (0.005) & 0.125 (0.005) \\
                \addlinespace
        \multicolumn{10}{c}{\(n=1024\)} \\
                \midrule
        Bumps & 0.2 & 0.496 (0.008) & 0.444 (0.002) & 0.464 (0.002) & 0.466 (0.003) & 0.466 (0.003) & 1.018 (0.036) & \textbf{0.443 (0.002)} & 1.033 (0.044) \\
        Bumps & 1.0 & 0.137 (0.001) & 0.276 (0.002) & 0.296 (0.002) & 0.217 (0.002) & 0.136 (0.001) & 0.138 (0.001) & 0.333 (0.002) & \textbf{0.106 (0.001)} \\
        Bumps & 3.0 & 0.026 (0.000) & 0.073 (0.001) & 0.100 (0.001) & 0.050 (0.000) & 0.026 (0.000) & 0.024 (0.000) & 0.149 (0.001) & \textbf{0.018 (0.000)} \\
        Blocks & 0.2 & \textbf{3.136 (0.083)} & 3.327 (0.026) & 3.654 (0.026) & 3.689 (0.027) & 3.627 (0.029) & 8.634 (0.343) & 3.375 (0.023) & 7.074 (0.391) \\
        Blocks & 1.0 & 0.598 (0.006) & 1.045 (0.011) & 1.203 (0.011) & 1.003 (0.013) & 0.733 (0.007) & 0.770 (0.009) & 1.250 (0.009) & \textbf{0.555 (0.012)} \\
        Blocks & 3.0 & 0.193 (0.002) & 0.386 (0.002) & 0.491 (0.003) & 0.313 (0.003) & 0.169 (0.001) & 0.171 (0.001) & 0.493 (0.003) & \textbf{0.127 (0.001)} \\
        Doppler & 0.2 & 0.076 (0.001) & \textbf{0.076 (0.001)} & 0.086 (0.001) & 0.086 (0.001) & 0.085 (0.001) & 0.184 (0.007) & 0.078 (0.000) & 0.153 (0.007) \\
        Doppler & 1.0 & \textbf{0.010 (0.000)} & 0.018 (0.000) & 0.025 (0.000) & 0.020 (0.000) & 0.014 (0.000) & 0.015 (0.000) & 0.028 (0.000) & 0.010 (0.000) \\
        Doppler & 3.0 & 0.002 (0.000) & 0.004 (0.000) & 0.006 (0.000) & 0.004 (0.000) & 0.002 (0.000) & 0.002 (0.000) & 0.012 (0.000) & \textbf{0.002 (0.000)} \\
        HeaviSine & 0.2 & \textbf{3.526 (0.155)} & 4.775 (0.153) & 6.280 (0.156) & 6.754 (0.175) & 5.883 (0.162) & 15.401 (0.904) & 6.233 (0.163) & 11.381 (0.803) \\
        HeaviSine & 1.0 & \textbf{0.278 (0.008)} & 0.488 (0.015) & 0.852 (0.015) & 0.794 (0.023) & 0.579 (0.011) & 0.882 (0.032) & 0.644 (0.013) & 0.620 (0.035) \\
        HeaviSine & 3.0 & \textbf{0.075 (0.001)} & 0.111 (0.001) & 0.197 (0.002) & 0.167 (0.002) & 0.120 (0.001) & 0.142 (0.003) & 0.279 (0.003) & 0.091 (0.004) \\
                \addlinespace
        \multicolumn{10}{c}{\(n=2048\)} \\
                \midrule
        Bumps & 0.2 & 0.434 (0.004) & \textbf{0.430 (0.001)} & 0.450 (0.001) & 0.452 (0.001) & 0.450 (0.001) & 0.905 (0.032) & 0.431 (0.001) & 0.972 (0.164) \\
        Bumps & 1.0 & 0.085 (0.001) & 0.182 (0.001) & 0.212 (0.001) & 0.152 (0.001) & 0.099 (0.001) & 0.101 (0.001) & 0.272 (0.001) & \textbf{0.073 (0.001)} \\
        Bumps & 3.0 & 0.015 (0.000) & 0.043 (0.000) & 0.059 (0.000) & 0.032 (0.000) & 0.017 (0.000) & 0.017 (0.000) & 0.097 (0.001) & \textbf{0.012 (0.000)} \\
        Blocks & 0.2 & \textbf{2.162 (0.046)} & 2.774 (0.036) & 3.141 (0.037) & 3.108 (0.040) & 3.011 (0.041) & 6.825 (0.251) & 2.912 (0.033) & 5.820 (0.249) \\
        Blocks & 1.0 & 0.464 (0.003) & 0.738 (0.005) & 0.936 (0.005) & 0.773 (0.005) & 0.582 (0.004) & 0.626 (0.007) & 0.960 (0.009) & \textbf{0.419 (0.009)} \\
        Blocks & 3.0 & 0.117 (0.001) & 0.274 (0.002) & 0.342 (0.002) & 0.214 (0.002) & 0.122 (0.001) & 0.123 (0.001) & 0.391 (0.002) & \textbf{0.088 (0.001)} \\
        Doppler & 0.2 & \textbf{0.055 (0.001)} & 0.068 (0.001) & 0.080 (0.001) & 0.078 (0.001) & 0.076 (0.001) & 0.154 (0.005) & 0.074 (0.000) & 0.138 (0.006) \\
        Doppler & 1.0 & \textbf{0.006 (0.000)} & 0.012 (0.000) & 0.017 (0.000) & 0.013 (0.000) & 0.009 (0.000) & 0.011 (0.000) & 0.022 (0.000) & 0.008 (0.000) \\
        Doppler & 3.0 & \textbf{0.001 (0.000)} & 0.002 (0.000) & 0.004 (0.000) & 0.002 (0.000) & 0.002 (0.000) & 0.002 (0.000) & 0.010 (0.000) & 0.001 (0.000) \\
        HeaviSine & 0.2 & \textbf{2.513 (0.093)} & 2.515 (0.083) & 4.287 (0.093) & 4.212 (0.094) & 3.980 (0.089) & 13.394 (0.629) & 3.800 (0.165) & 11.031 (0.610) \\
        HeaviSine & 1.0 & \textbf{0.184 (0.004)} & 0.233 (0.005) & 0.569 (0.009) & 0.498 (0.009) & 0.402 (0.007) & 0.687 (0.023) & 0.417 (0.009) & 0.546 (0.024) \\
        HeaviSine & 3.0 & \textbf{0.052 (0.001)} & 0.087 (0.001) & 0.135 (0.002) & 0.114 (0.002) & 0.086 (0.001) & 0.110 (0.002) & 0.283 (0.003) & 0.072 (0.003) \\
        \bottomrule
    \end{tabular}
    }
\end{table}

\begin{figure}[htbp]
    \centering
    \includegraphics[width=\linewidth]{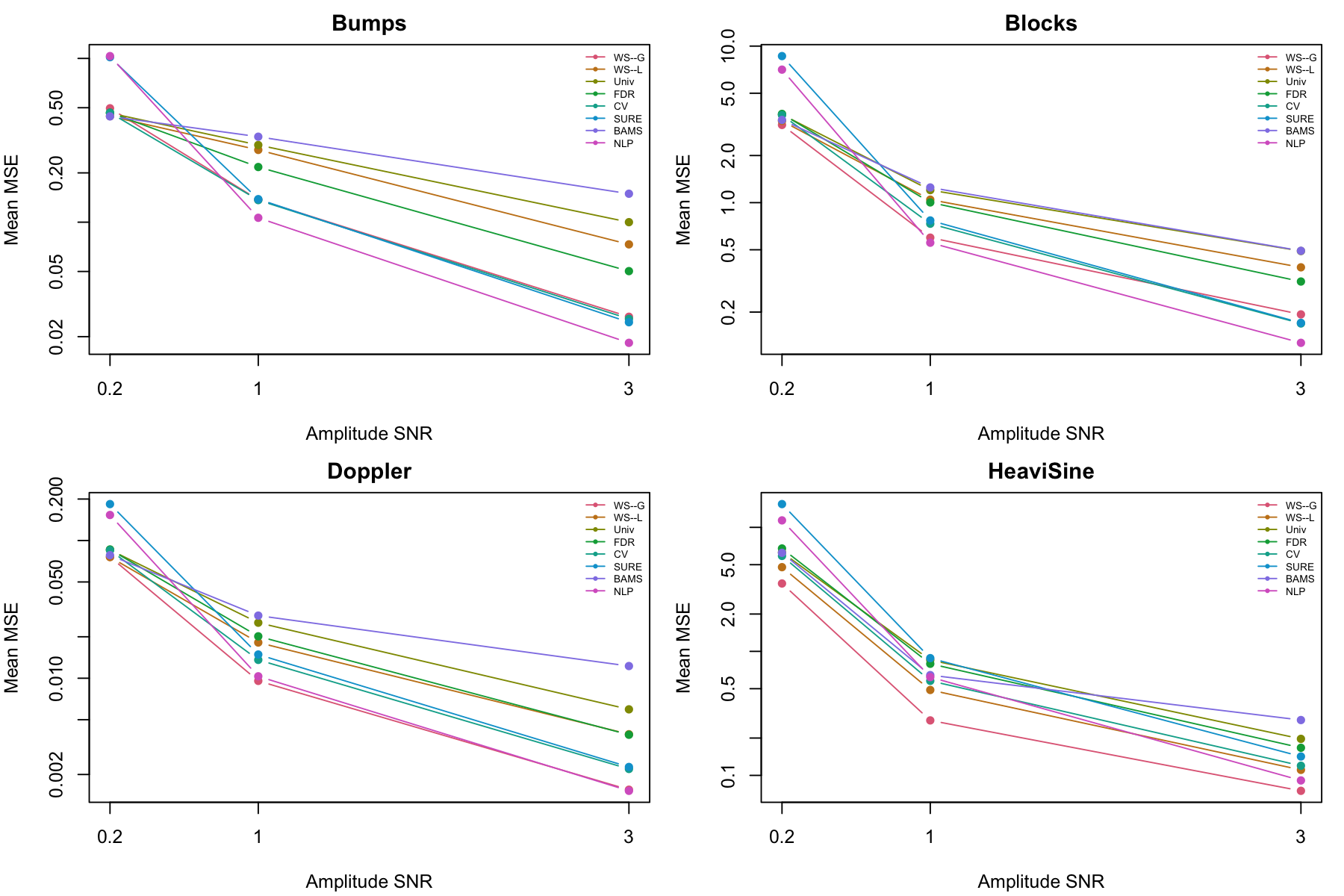}
    \caption{Mean MSE under Laplace errors as a function of amplitude SNR at \(n=1024\), shown separately for the four test signals.  Each point averages 100 replications.}
    \label{fig:supp-laplace-profile}
\end{figure}

\begin{table}[p]
    \centering
    \caption{Descriptive averages over the 36 Laplace-error design cells.  NMSE is normalized by the variance of the corresponding true signal, and runtime is the mean elapsed time in seconds per fitted replicate.\\[-5pt]}
    \label{tab:supp-laplace-overall}
    \scriptsize
    \renewcommand{\arraystretch}{0.92}
    \resizebox{\linewidth}{!}{%
    \begin{tabular}{lrrrrr}
        \toprule
        Method & Mean MSE & Mean NMSE & Mean MAE & Mean output SNR & Runtime (s) \\
        \midrule
        WS--G & 0.751 & 0.343 & 0.441 & 3.801 & 0.273 \\
        WS--L & 0.895 & 0.404 & 0.536 & 2.852 & 1.795 \\
        Univ & 1.125 & 0.469 & 0.618 & 2.338 & 0.002 \\
        FDR & 1.100 & 0.431 & 0.590 & 2.657 & 0.003 \\
        CV & 0.976 & 0.382 & 0.532 & 3.234 & 0.011 \\
        SURE & 2.265 & 0.776 & 0.648 & 2.995 & 0.003 \\
        BAMS & 1.061 & 0.474 & 0.605 & 2.144 & 0.011 \\
        NLP & 1.845 & 0.685 & 0.570 & 3.674 & 125.927 \\
        \bottomrule
    \end{tabular}
    }
\end{table}

Supplementary Table~\ref{tab:supp-laplace-overall} shows that the pooled WS--G mean NMSE was \(0.343\), compared with \(0.404\) for WS--L and \(0.685\) for NLP.  At \(n=1024\) and \(\rho=0.2\), the paired WS--G minus WS--L NMSE interval includes zero, whereas the corresponding interval for WS--G minus NLP is clearly negative, indicating lower WS--G NMSE than NLP in this slice.

\subsubsection{Paired comparisons and replicate variability}
Supplementary Table~\ref{tab:supp-laplace-paired} reports paired MSE and NMSE differences at \(n=1024\) and \(\rho=0.2\), pooled over the four signals.  The interval is formed from the paired replicate differences.  Supplementary Figure~\ref{fig:supp-laplace-box} displays the corresponding replicate distributions, while Supplementary Figure~\ref{fig:supp-laplace-runtime} displays the runtime distributions on a logarithmic scale.

\begin{table}[p]
    \centering
    \caption{Paired comparison with WS--Gaussian under Laplace errors at (n=1024) and amplitude SNR \(0.2\), pooled over the four signals.  Negative differences favor WS--Gaussian.\\[-5pt]}
    \label{tab:supp-laplace-paired}
    \scriptsize
    \renewcommand{\arraystretch}{0.92}
    \resizebox{\linewidth}{!}{%
    \begin{tabular}{lccc}
        \toprule
        Competitor & MSE difference [95\% CI] & NMSE difference [95\% CI] & WS--G win rate (\%) \\
        \midrule
        WS--L & -0.347 [-0.440, -0.253] & -0.018 [-0.038, 0.001] & 60.0 \\
        Univ & -0.812 [-0.947, -0.678] & -0.123 [-0.145, -0.102] & 76.2 \\
        FDR & -0.940 [-1.105, -0.776] & -0.142 [-0.165, -0.120] & 76.0 \\
        CV & -0.707 [-0.823, -0.591] & -0.110 [-0.129, -0.090] & 76.0 \\
        SURE & -4.501 [-5.151, -3.850] & -1.325 [-1.410, -1.241] & 100.0 \\
        BAMS & -0.724 [-0.875, -0.572] & -0.069 [-0.094, -0.045] & 62.5 \\
        NLP & -3.102 [-3.604, -2.600] & -1.022 [-1.107, -0.938] & 86.8 \\
        \bottomrule
    \end{tabular}
    }
\end{table}

\begin{figure}[htbp]
    \centering
    \includegraphics[width=\linewidth]{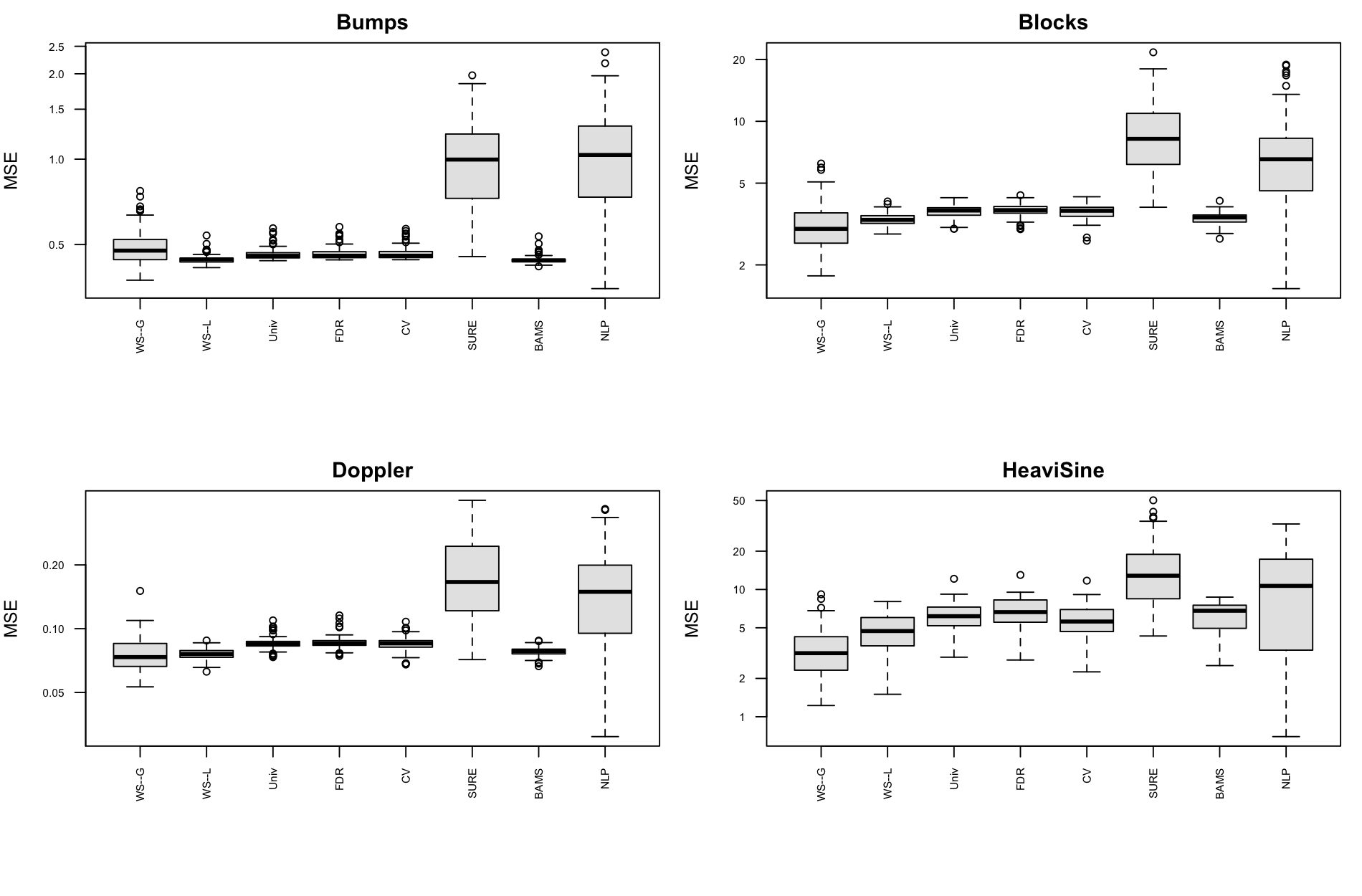}
    \caption{Replicate-specific MSE distributions under Laplace errors at \(n=1024\) and amplitude SNR \(\rho=0.2\), shown separately for Bumps, Blocks, Doppler, and HeaviSine.}
    \label{fig:supp-laplace-box}
\end{figure}

\begin{figure}[htbp]
    \centering
    \includegraphics[width=0.85\linewidth]{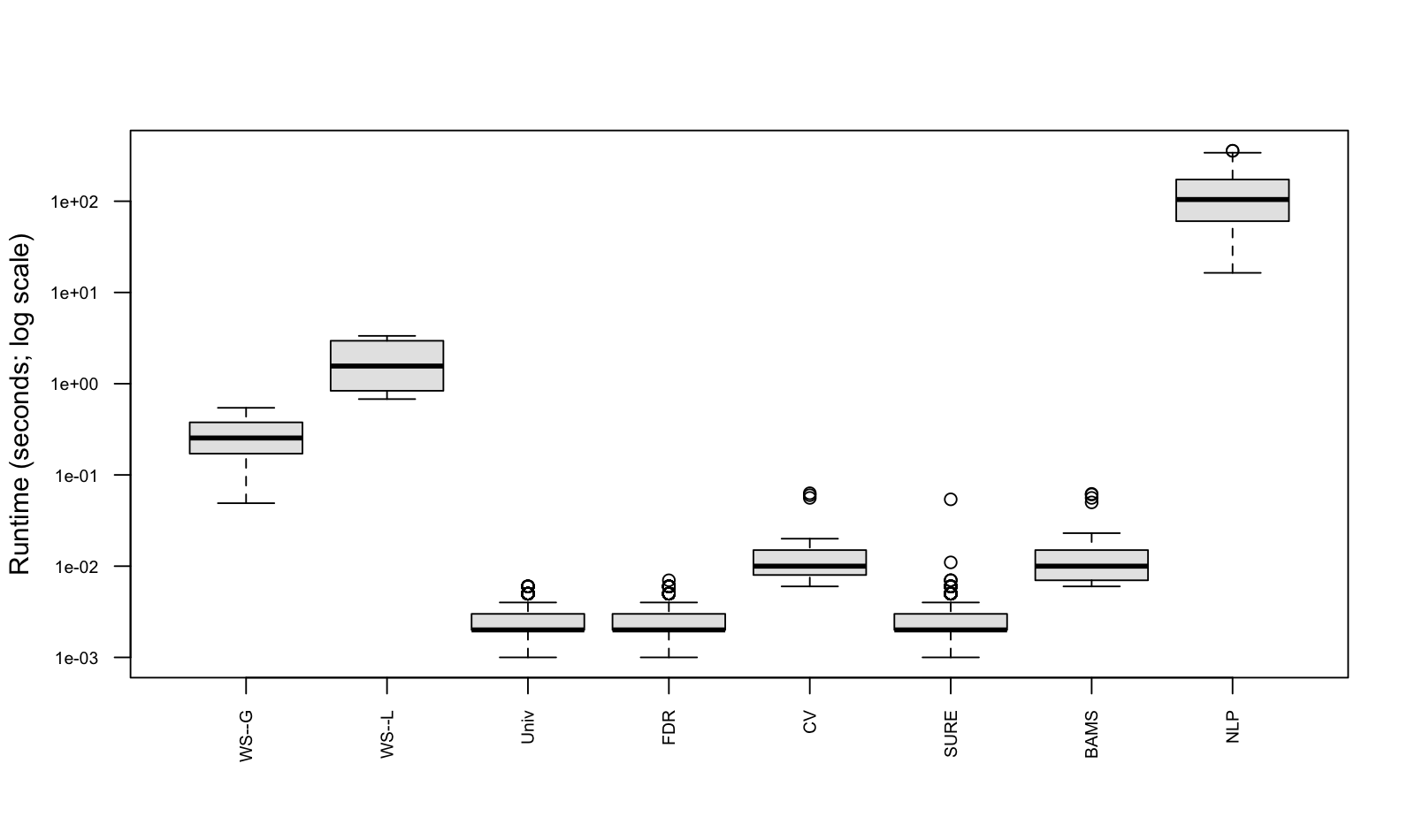}
    \caption{Per-fit runtime distributions under Laplace errors across the 36 design cells, displayed on a logarithmic scale.}
    \label{fig:supp-laplace-runtime}
\end{figure}

\subsubsection{Noise-scale estimates}
Supplementary Figure~\ref{fig:supp-laplace-sigma} reports the fitted noise-scale estimates from the WS--Gaussian fits.  The scale estimator is a robust wavelet-domain diagnostic; under the Laplace-error experiment, it is not expected to equal the Gaussian-model likelihood scale coefficient by coefficient.

\begin{figure}[htbp]
    \centering
    \includegraphics[width=\linewidth]{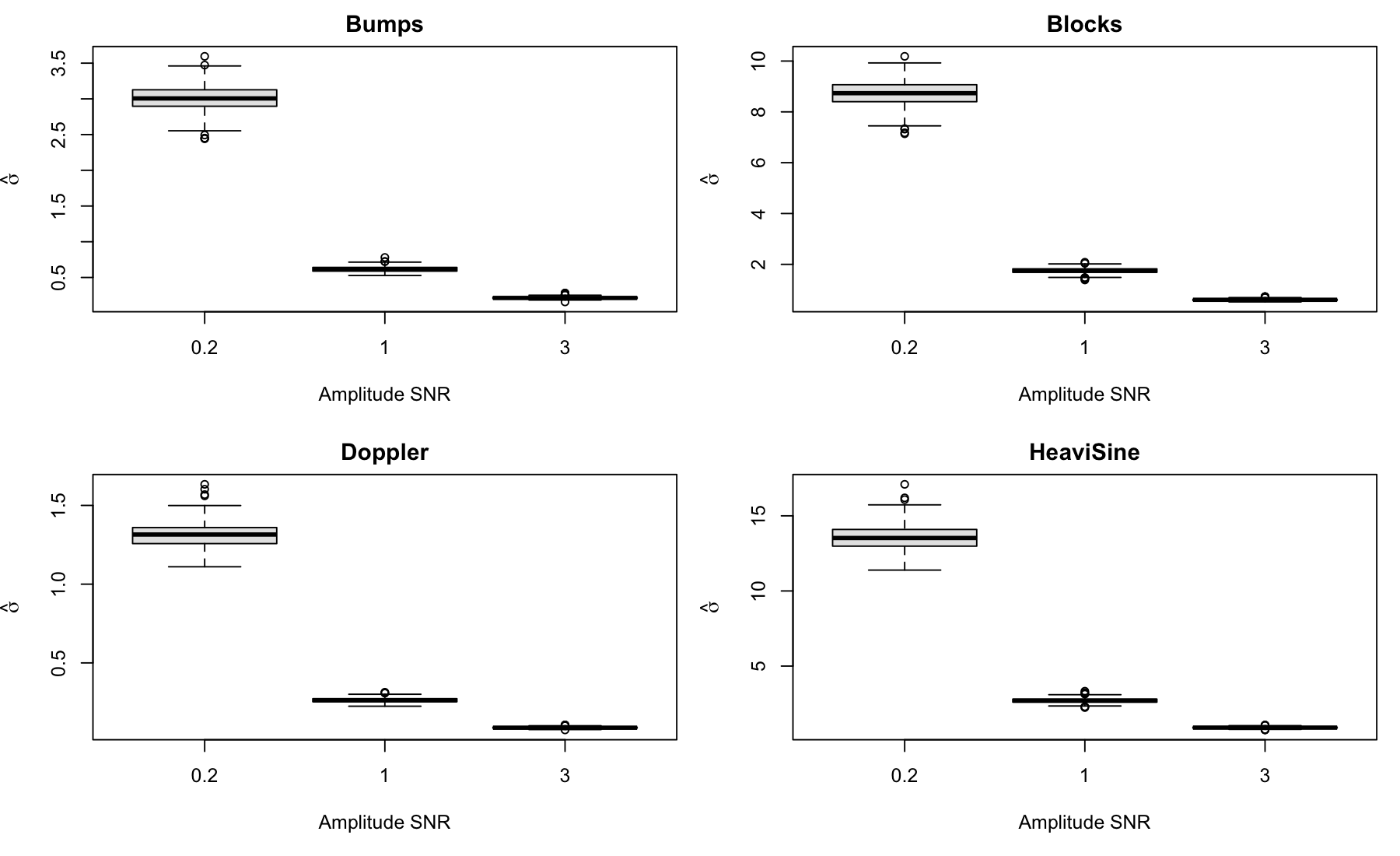}
    \caption{Replicate-level distributions of the estimated noise scale \(\widehat\sigma\) from WS--Gaussian fits under Laplace errors, by signal and amplitude SNR.}
    \label{fig:supp-laplace-sigma}
\end{figure}

The Laplace-error results are therefore interpreted as a robustness analysis.  They assess whether the framework remains operational when the observation errors are symmetric and heavier-tailed than Gaussian errors, without implying that either working likelihood dominates for every signal or SNR.

\clearpage

\section{Supplementary seismic analyses}
\label{sec:supp-seismic}

\subsection{Seismic trace data supplementary analysis}
\label{sec:supp-seismic-real}

\begin{table}[htbp]
    \centering
    \caption{Signal-level diagnostics for the Chino Hills channel-2
    accelerogram.  Residual SD and residual roughness are in
    \(\mathrm{cm}\,\mathrm{s}^{-2}\); retained energy is dimensionless; peak
    acceleration is in \(\mathrm{cm}\,\mathrm{s}^{-2}\); and peak time is in
    seconds.  The estimated noise scale is repeated because it was common to
    all fitted methods.\\[-5pt]}
    \label{tab:seismic-diagnostics-channel2}
    \scriptsize
    \resizebox{\linewidth}{!}{%
    \begin{tabular}{lrrrrrr}
        \toprule
        Method & \(\widehat\sigma\) & Residual SD &
        Residual roughness & Retained energy & Peak \(|a|\) & Peak time \\
        \midrule
        WS--Gaussian      & 0.002316 & 0.368729 & 0.064033 & 0.856550 & 15.459820 & 36.145000 \\
        WS--Laplace       & 0.002316 & 0.368825 & 0.064373 & 0.856030 & 15.452390 & 36.145000 \\
        Wendland-only     & 0.002316 & 0.390173 & 0.067782 & 0.859128 & 15.605580 & 36.145000 \\
        Semicircle-only   & 0.002316 & 0.368730 & 0.064028 & 0.856593 & 15.460290 & 36.145000 \\
        Univ              & 0.002316 & 0.024427 & 0.017773 & 0.993997 & 18.634640 & 36.640000 \\
        FDR               & 0.002316 & 0.003972 & 0.003962 & 0.999208 & 18.662150 & 36.640000 \\
        CV                & 0.002316 & 0.025056 & 0.018160 & 0.993816 & 18.633550 & 36.640000 \\
        SURE              & 0.002316 & 0.002983 & 0.003172 & 0.999430 & 18.663210 & 36.640000 \\
        \bottomrule
    \end{tabular}}
\end{table}

\begin{table}[htbp]
    \centering
    \caption{Signal-level diagnostics for the Chino Hills channel-3
    accelerogram.  Residual SD and residual roughness are in
    \(\mathrm{cm}\,\mathrm{s}^{-2}\); retained energy is dimensionless; peak
    acceleration is in \(\mathrm{cm}\,\mathrm{s}^{-2}\); and peak time is in
    seconds.  The estimated noise scale is repeated because it was common to
    all fitted methods.\\[-5pt]}
    \label{tab:seismic-diagnostics-channel3}
    \scriptsize
    \resizebox{\linewidth}{!}{%
    \begin{tabular}{lrrrrrr}
        \toprule
        Method & \(\widehat\sigma\) & Residual SD &
        Residual roughness & Retained energy & Peak \(|a|\) & Peak time \\
        \midrule
        WS--Gaussian      & 0.001096 & 0.458692 & 0.053394 & 0.880936 & 29.600850 & 36.500000 \\
        WS--Laplace       & 0.001096 & 0.458666 & 0.053465 & 0.880819 & 29.600410 & 36.500000 \\
        Wendland-only     & 0.001096 & 0.520664 & 0.062985 & 0.881834 & 29.438360 & 36.500000 \\
        Semicircle-only   & 0.001096 & 0.458695 & 0.053395 & 0.880986 & 29.600860 & 36.500000 \\
        Univ              & 0.001096 & 0.012352 & 0.008607 & 0.998606 & 32.972220 & 36.505000 \\
        FDR               & 0.001096 & 0.001921 & 0.001961 & 0.999833 & 32.970040 & 36.505000 \\
        CV                & 0.001096 & 0.012670 & 0.008796 & 0.998563 & 32.972330 & 36.505000 \\
        SURE              & 0.001096 & 0.001968 & 0.001998 & 0.999828 & 32.970000 & 36.505000 \\
        \bottomrule
    \end{tabular}}
\end{table}

\begin{figure}[htbp]
    \centering
    \includegraphics[width=\linewidth]{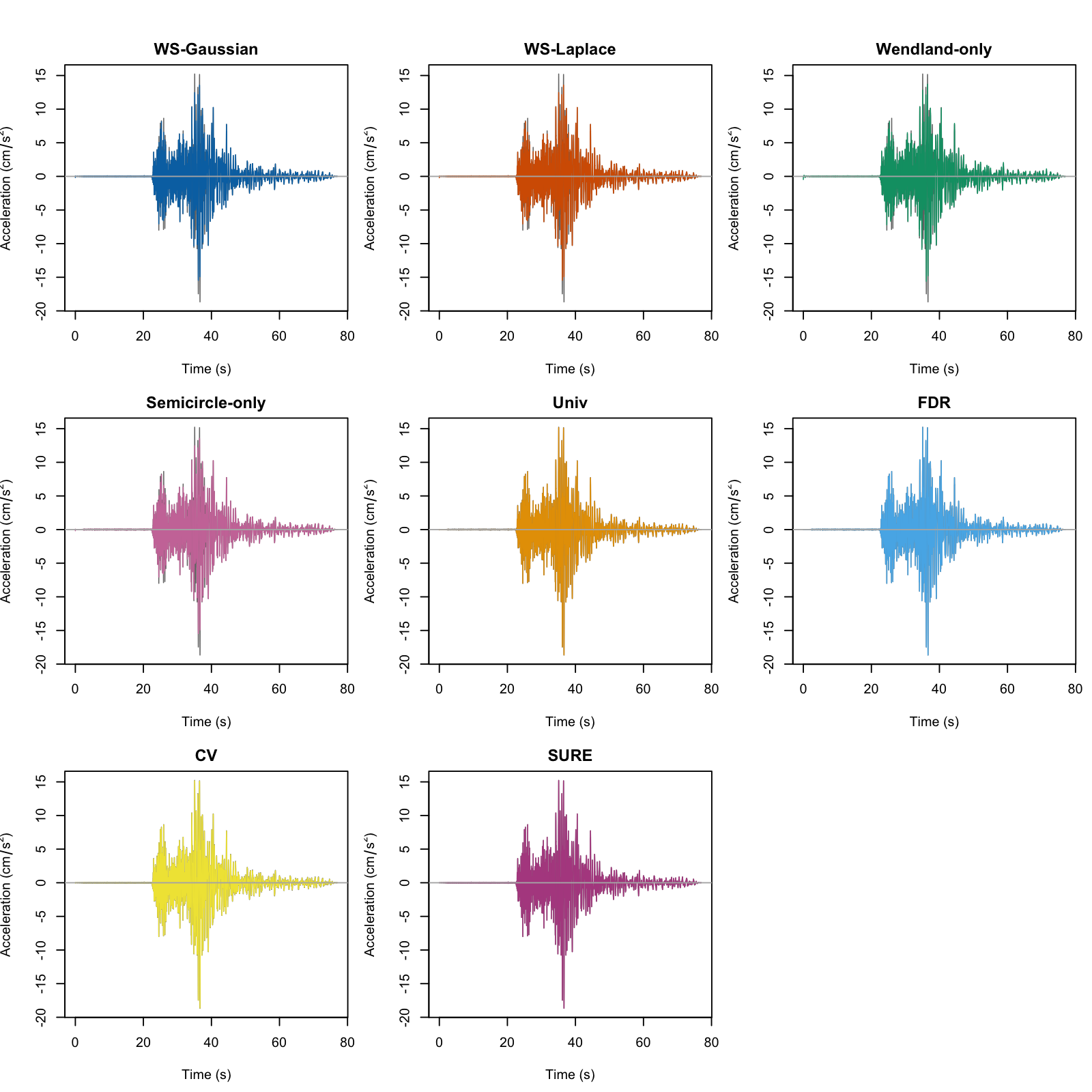}
    \caption{Full-record comparison for the 77-second Chino Hills channel-2
    accelerogram.  In each panel, the gray curve is the observed acceleration
    and the colored curve is the reconstruction from the indicated method.}
    \label{fig:seismic-full-record-channel2}
\end{figure}

\begin{figure}[htbp]
    \centering
    \includegraphics[width=\linewidth]{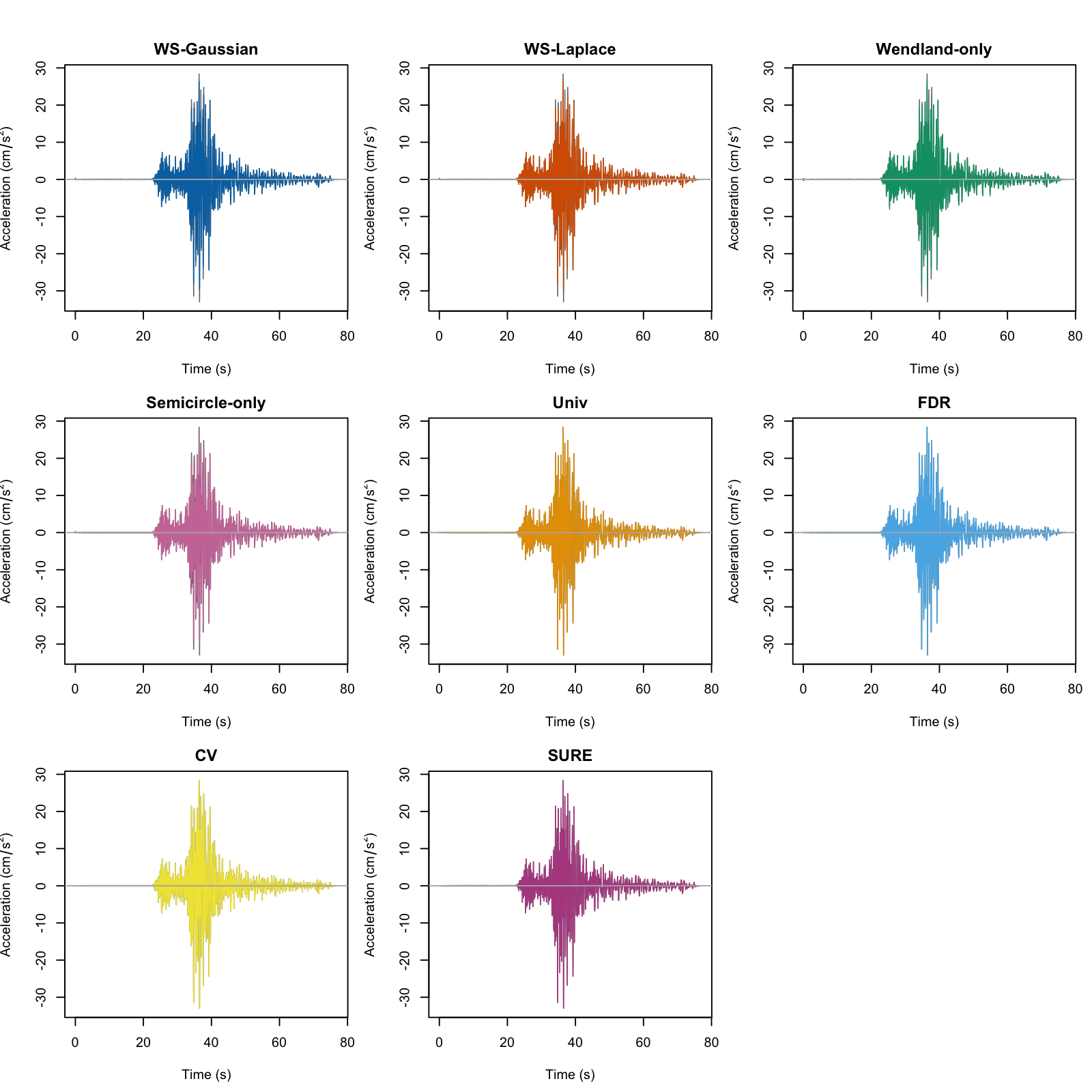}
    \caption{Full-record comparison for the 77-second Chino Hills channel-3
    accelerogram.  In each panel, the gray curve is the observed acceleration
    and the colored curve is the reconstruction from the indicated method.}
    \label{fig:seismic-full-record-channel3}
\end{figure}

\subsection{Semi-synthetic seismic validation}
\label{sec:supp-seismic-semisynthetic}

This section evaluates recovery under controlled noise using the processed
channel-1 accelerogram from the Chino Hills earthquake as a surrogate truth.
This experiment is not a real-data accuracy analysis: the processed trace is
treated as known only to permit objective comparison with the reconstructed
signals.  The results are therefore conditional on the processed trace being
an adequate representation of the underlying acceleration signal.

The trace contains \(N=15{,}400\) observations sampled at
\(\Delta t=0.005\) seconds.  For the wavelet transform, the record was
reverse-padded to the next dyadic length,
\(\widetilde N=16{,}384\), by appending the reverse of its terminal
observations.  All reported accuracy measures were computed after
reconstruction using only the original \(N\) observations; padded values were
not included in any metric.

Let \(a_i^\star\) denote the processed surrogate truth and let
\[
    y_i=a_i^\star+\varepsilon_i,\qquad
    \varepsilon_i\sim N(0,\sigma_\varepsilon^2),
    \qquad
    \sigma_\varepsilon=\frac{\operatorname{sd}(a_1^\star,\ldots,a_N^\star)}
    {\rho},
\]
where \(\rho\in\{0.5,1,2,5\}\) is the amplitude signal-to-noise ratio.
Thus, the noise standard deviation is scaled relative to the sample standard
deviation of the surrogate truth.  100 independent replications were
generated at each SNR.  Within each replication, every method was applied to
the same noisy realization, allowing paired comparisons. The fitted procedures used \(q=48\), matching the seismic application settings in Table~6 of the main manuscript; the \(q=128\) selection applied to the separate simulation-suite analyses.

The comparison was restricted to selected methods used for this application:
WS--Gaussian with adaptive, constant, Wendland-only, Semicircle-only, and
pooled three-level noise-scale specifications; WS--Laplace with adaptive and
pooled three-level specifications; and the Univ, FDR, CV, and SURE thresholding
rules.  The noisy observation itself was also included as an identity
baseline.  The baseline is not a fitted method and represents the error
obtained by making no denoising modification.

For an estimate \(\widehat a\), the reported metrics are
\[
    \operatorname{MSE}(\widehat a)
    =\frac{1}{N}\sum_{i=1}^{N}(\widehat a_i-a_i^\star)^2,
    \qquad
    \operatorname{NMSE}(\widehat a)
    =\frac{\operatorname{MSE}(\widehat a)}
    {\operatorname{var}(a_1^\star,\ldots,a_N^\star)}.
\]
The peak-magnitude error is
\[
    \left|
    \max_{1\leq i\leq N}|\widehat a_i|
    -
    \max_{1\leq i\leq N}|a_i^\star|
    \right|,
\]
and the peak-timing error is
\[
    \Delta t\left|
    \operatorname*{arg\,max}_{1\leq i\leq N}|\widehat a_i|
    -
    \operatorname*{arg\,max}_{1\leq i\leq N}|a_i^\star|
    \right|.
\]

Each saved replicate-level row is keyed by study, signal, \(n\), amplitude SNR \(\rho\), replication seed, data-generating error law, likelihood, wavelet basis, \(J_0\), and hyperparameter specification.  For fitted WS procedures, the saved diagnostics include \(\widehat{\eta}_0\), \(\widehat{\eta}_1\), the level-specific \(\widehat\omega_j\), \(\widehat\beta_j\), \(\widehat\pi_j\), \(\widehat\sigma\), optimizer convergence and bound-hit flags, quadrature order, elapsed time, warning/fallback messages, and finite-metric indicators.  Common random numbers were used within each scenario, and the same 100 replication seeds were supplied to all methods being compared.  No fit was retried or excluded because of a warning; the counts in Supplementary Table~\ref{tab:supp-qc} are based on the saved first-pass outputs.

Supplementary Table~\ref{tab:supp-seismic-semisynthetic} summarizes the
results.  For the fitted procedures, entries are replicate means with Monte
Carlo standard errors in parentheses.  The noisy-observation rows are shown
as descriptive baseline means.  The baseline NMSEs are approximately
\(1/\rho^2\), as expected from the amplitude-SNR construction: they are
approximately \(3.995\), \(0.9995\), \(0.2494\), and \(0.0399\) at
\(\rho=0.5,1,2,\) and \(5\), respectively.

At SNRs \(0.5\), \(1\), and \(2\), WS--Gaussian has lower NMSE than the noisy
observation, with mean NMSEs \(0.269\), \(0.166\), and \(0.115\), respectively.
At SNR \(5\), however, its mean NMSE is \(0.091\), exceeding the noisy
baseline value of approximately \(0.040\).  Thus, in this semi-synthetic
experiment, WS--Gaussian improves the noisy record in the lower- and
moderate-SNR settings but can introduce additional error when the noise is
already weak.  CV, FDR, and SURE obtain lower NMSE than the noisy baseline at every examined SNR.

The adaptive WS--Gaussian, constant, and Semicircle-only fits are nearly
identical across the four SNR values, whereas the Wendland-only fit is
consistently less accurate.  The pooled three-level noise-scale modification
has only a negligible effect.  The near equality of the adaptive and
Semicircle-only fits is consistent with the fitted adaptive weights being
very close to the Semicircle-only endpoint for this trace; this experiment
does not demonstrate substantial resolution-varying use of both slab
components.

The WS--Laplace fit has larger mean NMSE than WS--Gaussian at every SNR.
Its peak-timing behavior is particularly unfavorable at SNR \(2\): the mean
absolute timing error is \(0.412\) seconds, with median \(0.570\) seconds and
72 of the 100 replications having an error greater than \(0.1\) seconds.
This pattern is consistent with repeated selection of a different dominant peak rather than a single extreme replicate. Peak-timing error should therefore
be interpreted together with the waveform and peak-magnitude diagnostics.

\begin{table}[p]
    \centering
    \caption{Semi-synthetic seismic recovery using the processed channel-1 trace as a surrogate truth.  The table compares selected methods over 100 replications at each amplitude SNR.  For fitted procedures, entries are replicate means with Monte Carlo standard errors in parentheses.  The noisy-observation identity baseline is reported by its replicate mean only as a descriptive reference; it is not a fitted method.  MSE is in \((\mathrm{cm}\,\mathrm{s}^{-2})^2\), peak-magnitude error in \(\mathrm{cm}\,\mathrm{s}^{-2}\), and peak-timing error in seconds; the two peak errors are absolute errors relative to the surrogate truth.}
    \label{tab:supp-seismic-semisynthetic}
    \scriptsize
    \renewcommand{\arraystretch}{0.90}
    \resizebox{\linewidth}{!}{%
    \begin{tabular}{llrrrr}
        \toprule
        SNR & Method & MSE & NMSE & Peak magnitude error & Peak timing error \\
        \midrule
        0.5 & WS--G & 5.055 (0.023) & 0.269 (0.001) & 9.111 (0.235) & 0.172 (0.026) \\
        0.5 & WS--G constant & 5.055 (0.023) & 0.269 (0.001) & 9.111 (0.235) & 0.172 (0.026) \\
        0.5 & Wendland-only & 5.463 (0.028) & 0.291 (0.001) & 12.572 (0.210) & 0.137 (0.024) \\
        0.5 & Semicircle-only & 5.055 (0.023) & 0.269 (0.001) & 9.111 (0.235) & 0.172 (0.026) \\
        0.5 & WS--G pooled-3 MAD & 5.058 (0.024) & 0.269 (0.001) & 9.120 (0.235) & 0.172 (0.026) \\
        0.5 & WS--L & 11.960 (0.061) & 0.637 (0.003) & 22.641 (0.175) & 0.014 (0.001) \\
        0.5 & WS--L pooled-3 MAD & 11.970 (0.061) & 0.638 (0.003) & 22.664 (0.176) & 0.014 (0.001) \\
        0.5 & Univ & 7.959 (0.025) & 0.424 (0.001) & 13.131 (0.155) & 0.027 (0.001) \\
        0.5 & FDR & 6.833 (0.027) & 0.364 (0.001) & 11.178 (0.169) & 0.024 (0.001) \\
        0.5 & CV & 4.381 (0.019) & 0.233 (0.001) & 5.940 (0.218) & 0.036 (0.010) \\
        0.5 & SURE & 4.428 (0.020) & 0.236 (0.001) & 6.218 (0.223) & 0.037 (0.010) \\
        0.5 & Noisy observation & 75.019 & 3.995 & 14.942 & 0.284 \\
        1.0 & WS--G & 3.122 (0.011) & 0.166 (0.001) & 8.218 (0.094) & 0.215 (0.027) \\
        1.0 & WS--G constant & 3.122 (0.011) & 0.166 (0.001) & 8.218 (0.094) & 0.215 (0.027) \\
        1.0 & Wendland-only & 3.222 (0.011) & 0.172 (0.001) & 9.715 (0.079) & 0.406 (0.026) \\
        1.0 & Semicircle-only & 3.122 (0.011) & 0.166 (0.001) & 8.218 (0.094) & 0.215 (0.027) \\
        1.0 & WS--G pooled-3 MAD & 3.128 (0.011) & 0.167 (0.001) & 8.226 (0.095) & 0.226 (0.028) \\
        1.0 & WS--L & 5.258 (0.015) & 0.280 (0.001) & 8.410 (0.173) & 0.010 (0.006) \\
        1.0 & WS--L pooled-3 MAD & 5.270 (0.014) & 0.281 (0.001) & 8.389 (0.173) & 0.010 (0.006) \\
        1.0 & Univ & 4.007 (0.011) & 0.213 (0.001) & 7.073 (0.105) & 0.017 (0.000) \\
        1.0 & FDR & 3.089 (0.011) & 0.165 (0.001) & 5.642 (0.110) & 0.014 (0.000) \\
        1.0 & CV & 1.808 (0.006) & 0.096 (0.000) & 2.526 (0.139) & 0.020 (0.008) \\
        1.0 & SURE & 1.837 (0.007) & 0.098 (0.000) & 2.882 (0.141) & 0.021 (0.008) \\
        1.0 & Noisy observation & 18.767 & 1.000 & 6.095 & 0.152 \\
        2.0 & WS--G & 2.152 (0.005) & 0.115 (0.000) & 7.139 (0.066) & 0.005 (0.000) \\
        2.0 & WS--G constant & 2.152 (0.005) & 0.115 (0.000) & 7.138 (0.065) & 0.005 (0.000) \\
        2.0 & Wendland-only & 2.195 (0.006) & 0.117 (0.000) & 8.063 (0.058) & 0.010 (0.006) \\
        2.0 & Semicircle-only & 2.156 (0.005) & 0.115 (0.000) & 7.135 (0.066) & 0.005 (0.000) \\
        2.0 & WS--G pooled-3 MAD & 2.155 (0.005) & 0.115 (0.000) & 7.140 (0.065) & 0.005 (0.000) \\
        2.0 & WS--L & 3.403 (0.007) & 0.181 (0.000) & 8.593 (0.041) & 0.412 (0.026) \\
        2.0 & WS--L pooled-3 MAD & 3.411 (0.006) & 0.182 (0.000) & 8.594 (0.041) & 0.418 (0.025) \\
        2.0 & Univ & 1.879 (0.004) & 0.100 (0.000) & 3.598 (0.056) & 0.010 (0.000) \\
        2.0 & FDR & 1.279 (0.004) & 0.068 (0.000) & 2.397 (0.059) & 0.010 (0.000) \\
        2.0 & CV & 0.685 (0.002) & 0.036 (0.000) & 0.930 (0.057) & 0.007 (0.000) \\
        2.0 & SURE & 0.699 (0.003) & 0.037 (0.000) & 1.070 (0.063) & 0.007 (0.000) \\
        2.0 & Noisy observation & 4.683 & 0.249 & 2.527 & 0.034 \\
        5.0 & WS--G & 1.704 (0.002) & 0.091 (0.000) & 7.520 (0.028) & 0.008 (0.000) \\
        5.0 & WS--G constant & 1.704 (0.002) & 0.091 (0.000) & 7.507 (0.028) & 0.008 (0.000) \\
        5.0 & Wendland-only & 1.741 (0.002) & 0.093 (0.000) & 7.869 (0.026) & 0.007 (0.000) \\
        5.0 & Semicircle-only & 1.704 (0.002) & 0.091 (0.000) & 7.497 (0.028) & 0.008 (0.000) \\
        5.0 & WS--G pooled-3 MAD & 1.707 (0.002) & 0.091 (0.000) & 7.525 (0.028) & 0.008 (0.000) \\
        5.0 & WS--L & 2.048 (0.002) & 0.109 (0.000) & 7.067 (0.024) & 0.004 (0.000) \\
        5.0 & WS--L pooled-3 MAD & 2.058 (0.002) & 0.110 (0.000) & 7.070 (0.023) & 0.004 (0.000) \\
        5.0 & Univ & 0.601 (0.001) & 0.032 (0.000) & 1.116 (0.029) & 0.009 (0.000) \\
        5.0 & FDR & 0.349 (0.001) & 0.019 (0.000) & 0.794 (0.030) & 0.006 (0.000) \\
        5.0 & CV & 0.166 (0.000) & 0.009 (0.000) & 0.390 (0.028) & 0.004 (0.000) \\
        5.0 & SURE & 0.173 (0.001) & 0.009 (0.000) & 0.458 (0.029) & 0.005 (0.000) \\
        5.0 & Noisy observation & 0.750 & 0.040 & 0.755 & 0.004 \\
        \bottomrule
    \end{tabular}%
    }
\end{table}

Supplementary Figure~\ref{fig:supp-seismic-profiles} displays the mean NMSE
and mean absolute peak-timing error across the four SNR values.  The figure
provides a graphical summary of the fitted-method trends, whereas the noisy
observation baseline is reported numerically in
Supplementary Table~\ref{tab:supp-seismic-semisynthetic}.  These results are
a controlled surrogate-truth validation and should complement, rather than
replace, the real seismic trace analysis, for which no noise-free reference
waveform is available.

\begin{figure}[p]
    \centering
    \includegraphics[width=\linewidth]{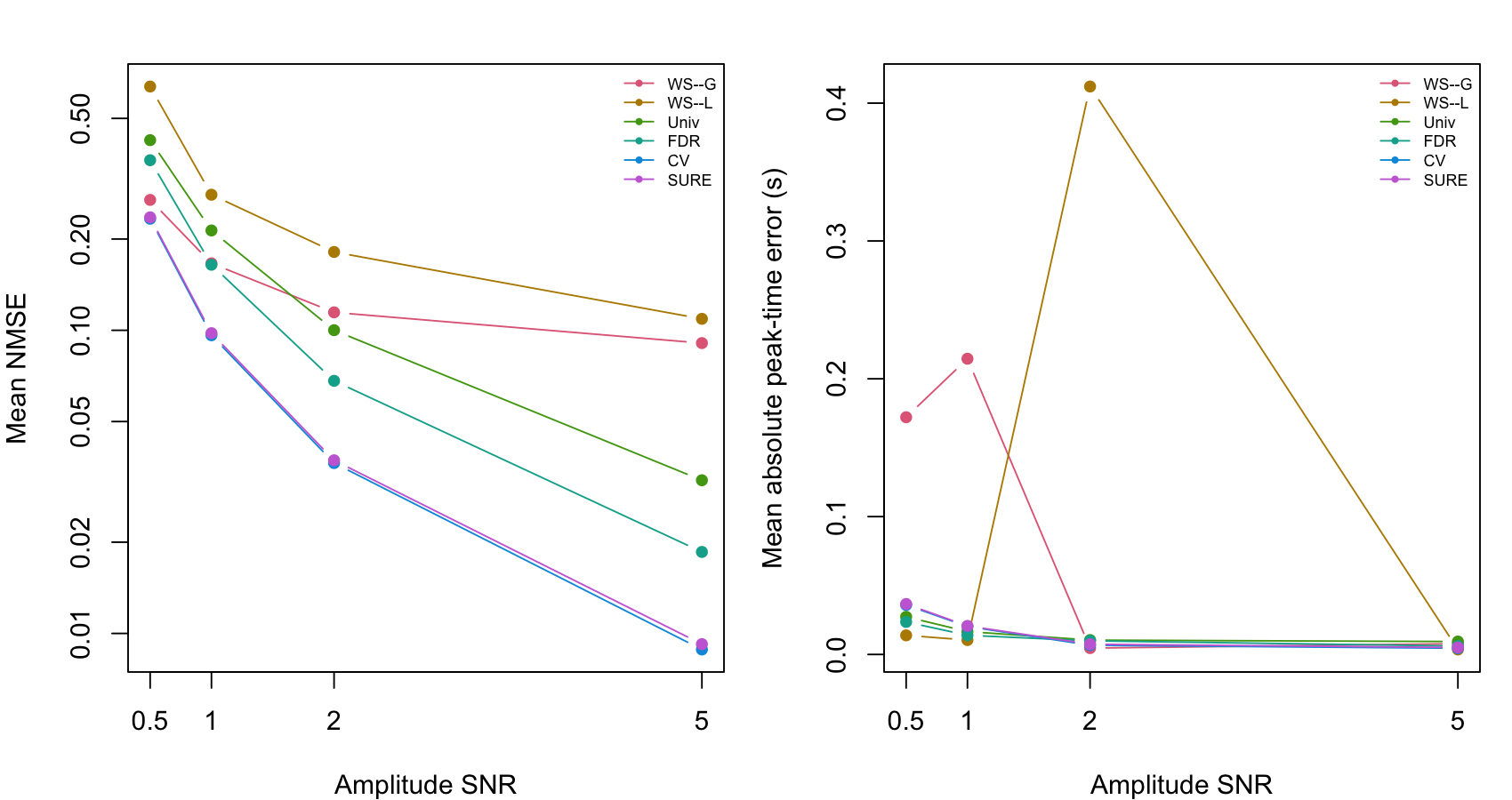}
    \caption{Semi-synthetic seismic validation using the processed channel-1
    accelerogram as a surrogate truth.  The left panel shows mean NMSE and the
    right panel shows mean absolute peak-timing error over 100 replications at
    amplitude SNR values \(0.5\), \(1\), \(2\), and \(5\).  Lower values indicate
    better recovery.}
    \label{fig:supp-seismic-profiles}
\end{figure}

\end{document}